\documentclass[preprint, 12pt,3p, proof, authoryear]{elsarticle}
\makeatletter
\def\ps@pprintTitle{%
 \let\@oddhead\@empty
 \let\@evenhead\@empty
 \def\@oddfoot{}%
 \let\@evenfoot\@oddfoot}
\makeatother

\usepackage{tikz}
\usetikzlibrary{patterns}
\usepackage{longtable, booktabs,multirow,lipsum,array}
\usepackage[authoryear, round]{natbib}
\usepackage{amssymb}
\usepackage{amsmath}
\usepackage{subfig}
\usepackage{enumerate}
\usepackage{mwe}
\usepackage[toc, header]{appendix}
\usepackage{mathtools}
\usepackage{graphicx}
\usepackage{soul}
\usepackage{hyperref}
\usepackage{amsfonts}
\usepackage{graphicx}
\usepackage[utf8]{inputenc}
\usepackage[all]{xy}
\usepackage{csquotes}
\usepackage{ulem}
\usepackage{amsmath}
\usepackage{amsthm}
\usepackage{import}
\usepackage{adjustbox}
\usepackage{caption}
\usepackage{pdflscape}
\usepackage{tikzpeople}
\usepackage{wasysym}
\usepackage{setspace}
\usepackage{booktabs}
\usepackage{arydshln}
\usepackage{comment}
\usepackage{mathbbol}
\usepackage[flushleft]{threeparttable}
\usepackage{algorithm,algpseudocode}
\usepackage{multibib}
\usepackage{titlesec}
\newcites{Main}{References}
\newcites{SM}{Appendix References}
\newcommand\independent{\protect\mathpalette{\protect\independenT}{\perp}}
\def\independenT#1#2{\mathrel{\rlap{$#1#2$}\mkern2mu{#1#2}}}

\newtheorem{lemma}{Lemma}
\newtheorem{proposition}{Proposition}

\newtheorem{theorem}{Theorem}
\newtheorem{corollary}{Corollary}
\newtheorem{definition}{Definition}

\newtheorem*{lemma*}{Lemma}
\newtheorem*{proposition*}{Proposition}
\newtheorem*{remark*}{Remark}
\newtheorem*{theorem*}{Theorem}
\newtheorem*{corollary*}{Corollary}
\newtheorem*{definition*}{Definition}
\newtheorem*{assumption*}{Assumption}

\usetikzlibrary{shapes,decorations.pathreplacing,arrows,calc,arrows.meta,fit,positioning}
\tikzset{
    -Latex,auto,node distance =1 cm and 1 cm,semithick,
    state/.style ={ellipse, draw, minimum width = 0.7 cm},
    point/.style = {circle, draw, inner sep=0.04cm,fill,node contents={}},
    el/.style = {inner sep=2pt, align=left, sloped},
    >={Stealth}
}

\usepackage{titletoc}

\newcommand\DoToC{%
  \startcontents
  \printcontents{}{1}{\textbf{ }\vskip3pt\hrule\vskip5pt}
  \vskip3pt\hrule\vskip5pt
}

\begin{document}

\begin{frontmatter}

\title{Optimal regimes with limited resources}

\author[]{Aaron L. Sarvet  \corref{cor1}}
\author[]{Julien D. Laurendeau}
\author[]{Mats J. Stensrud}

\address[]{Department of Mathematics, Ecole Polytechnique Fédérale de Lausanne, Switzerland}

\cortext[cor1]{\textbf{Contact information for corresponding author:}\\
Aaron L. Sarvet, Department of Mathematics, Ecole Polytechnique Fédérale de Lausanne, Switzerland. \url{aaron.sarvet@epfl.ch}}

\begin{abstract}
Policy-makers are often faced with the task of distributing a limited supply of resources. To support decision-making in these settings, statisticians are confronted with two challenges: estimands are defined by allocation strategies that are functions of features of all individuals in a cluster; and relatedly the observed data are neither independent nor identically distributed when individuals compete for resources. Existing statistical approaches are inadequate because they ignore at least one of these core features. As a solution, we develop theory for a general policy class of dynamic regimes for clustered data, covering existing results in classical and interference settings as special cases. We cover policy-relevant estimands and articulate realistic conditions compatible with resource-limited observed data. We derive identification and inference results for settings with a finite number of individuals in a cluster, where the observed dataset is viewed as a single draw from a super-population of clusters.  We also consider asymptotic estimands when the number of individuals in a cluster is allowed to grow; under explicit conditions, we recover previous results, thereby clarifying when the use of existing methods is permitted. Our general results lay the foundation for future research on dynamic regimes for clustered data, including the longitudinal cluster setting.    
\end{abstract}

\end{frontmatter}

\date{May 2022}

\newpage

\doublespacing

\section{Introduction}
\label{sec: Introduction}

Consider the canonical average treatment effect (ATE) $\Psi \coloneqq \mathbb{E}[Y_i^{a=1} - Y_i^{a=0}]$, which is the prevailing target of inference in applied and methodological studies. The ATE is the expected difference between outcomes for a single individual $i$ under an assignment to treatment ($a=1$) and control ($a=0$). Similarly, many other common estimands are defined by a single individual's outcome under different treatments. For example, existing definitions of the optimal dynamic regime $g^{\mathbf{opt}}$ are based on a function that maps an individual's pre-treatment features to a single treatment value, see e.g. \citetMain{murphy2003optimal,robins2004optimal}. In practice, the values of $\Psi$ and $g^{\mathbf{opt}}$ are used to guide policies not just for one individual, but for whole populations. The tacit premise is that each individual is assumed to be exchangeable with any other, and that their potential outcomes are drawn from a common distribution characterized by $\Psi$ and $g^{\mathbf{opt}}$.  

However, in many real-world settings, this premise is immediately violated; policy makers must use non-iid data, e.g., because of causal connections between individuals in a cluster. Relatedly, they have interest in interventions that depend on the features of \textit{all} individuals in a cluster.  For example, many real-world policies involve regimes for \textit{prioritizing} individuals and are familiarly called waiting-list rules, triaging rules or prioritization mechanisms,  among other terms. These rules are frequently used in national health care systems across the world \citepMain{lakshmi2013application, januleviciute2013impact, mullen2003prioritising,  hadorn1997new, feek2000rationing}. Today such rules are used for the provision of primary care \citepMain{breton2017assessing}, non-emergency surgeries \citepMain{australian2015national, powers2023managing, curtis2010waiting, de2007western, valente2021new, solans2013developing}, organ transplantation \citepMain{OPTNSRTR20172019}, palliative care \citepMain{russell2020triaging}, and diagnostic testing with X-rays, MRIs, or endoscopies \citepMain{luigi2013oecd}, among other scarce resources. To be explicit, consider the widely cited New York Ventilator Guidelines for ventilator allocation, which were specified by \citetMain{nyc2015ventilator} in preparation for a possible influenza pandemic. These guidelines describe regimes for allocating ventilators to a cluster of individuals according to their Sequential Organ Failure Assessment (SOFA) scores. 

Existing results on dynamic regimes \citepMain{orellana2010dynamic,  van2005history, murphy2001cpprg} ostensibly provide a framework to account for resource limitations in causal inference. Recently, this framework has been used to study optimal resource-constrained regimes in statistics \citepMain{luedtke2016optimal, qiu2020optimal, luedtke2022individualized, canigliaestimating, zhao2023positivity, xu2023optimal, imai2023experimental}, econometrics \citepMain{bhattacharya2012inferring,kitagawa2018should, mbakop2021model, athey2021policy, sun2021treatment, pellatt2022pac, adusumilli2019dynamically, liu2022policy},  and machine learning, e.g., \citepMain{zhou2024optimal, luedtke2019asymptotically, badanidiyuru2018bandits, tran2012knapsack}.
For example, \citetMain{luedtke2016optimal} and \citetMain{qiu2020optimal} described semiparametric efficient estimators for the optimal \textit{individualized} dynamic regime (hereby, IR) subject to fixed constraints. Later, \citetMain{athey2021policy} described an efficient estimation strategy using regret bounds \citepMain{manski2004statistical}. \citetMain{sun2021treatment} generalize the notion of a resource constraint to one of a fixed budget, to accomodate settings where individual-specific costs of treatment randomly vary. \citetMain{canigliaestimating} parameterized a marginal structural model \citepMain{robins2000marginal} and subsequently optimized the parameters defining an IR under the resource constraint. 

We argue that these IR proposals are ill-posed for the typical resource-limited setting because they are restricted to individualized policies and iid data. There is indeed a growing literature that goes beyond such a restricted framework, which allows for interference \citepMain{hudgens2008toward, sobel2006randomized} or more generally causally connected individuals \citepMain{van2014causal} or social networks \citepMain{ogburn2020causal, manski2013identification}. One active strand of this literature concerns inference on classical estimands defined by individualized policies (e.g., \citetMain{savje2021average}). A second active strand concerns inference on estimands commonly called ``spillover'' or indirect effects \citepMain{hu2022average, forastiere2021identification, vanderweele2011effect}. Arguably these strands are related: either dependencies between individuals are regarded explicitly as a nuisance that complicates inference, or interest lies in quantifying the magnitude of this nuisance. In contrast, very few works explicitly consider dynamic regimes in limited resource settings. The literature on design-based inference is perhaps one exception; their parameters are defined by regimes corresponding to treatment allocation plans of experimental designs, which subsume, for example, Bernoulli designs, complete randomization designs, and paired-randomization designs, see \citetSM{savje2021average} for a review. However these works have been formulated exclusively for experimental data, and their experimental designs only correspond to static stochastic regimes. In the context of observational data, \citetSM{van2014causal} and \citetSM{ogburn2020causal} consider dynamic regimes in a network setting. Similar to the exposure mappings of \citetMain{aronow2017estimating} and \citetMain{savje2024causal} in experimental settings, \citetSM{van2014causal} and \citetSM{ogburn2020causal}'s regimes are individualized in the sense that each individual's treatment is restricted to be a fixed function of at most a randomizer, their own covariates, and the covariates of a bounded number of network ties. 
Collectively, these works are not directly relevant to our setting. First, we will consider observational data where the design, i.e., assignment of treatment, is unknown; thus the proposed methodologies of the designed-based literature are infeasible and their estimands are of questionable practical relevance. Second, we will evaluate regimes that are dynamic with respect to covariates of \textit{all} individuals in a cluster;  returning to the ventilator allocation example \citepMain{nyc2015ventilator}, whether or not an individual receives a ventilator unit in an ICU during a pandemic is not merely a function of that individual's own SOFA score, but of the SOFA scores of a cluster of patients in, say, the hospital.   Such regimes have, to our knowledge, not been considered in the interference literature.

As a solution, we work at the intersection of these two fast-growing literatures: to  consider  dynamic and optimal regimes in limited resource settings, we ground our work in frameworks traditionally applied in interference and network settings. Specifically, we articulate assumptions that both encode characteristic features of the resource-limited setting and permit articulation of estimands corresponding to real-world questions. Insodoing, we elaborate a general policy class of dynamic regimes for clustered data, covering special cases considered by the existing literature on interference \citepMain{savje2021average, van2014causal, ogburn2020causal}.  We further derive identification and estimation results for settings with either a finite or a near-infinite number of individuals in a cluster. Our elaborated results extend to longitudinal settings wherein limited resources frequently appear, and thereby we set the stage for future research. 

As an illustrative case study, we consider the estimand from \citetMain{qiu2020optimal} and \citetMain{luedtke2016optimal}, who aimed to learn an optimal regime, dynamic with respect to covariates $L$, subject to investigator-specified constraints on the marginal probability of treatment. While their proposed estimand and model violate characteristic features of the resource-limited setting,  we locate their statistical parameter as a large population limiting case of a functional identifying a more appropriate causal estimand in an elaborated model.  Thereby, we justify the use of existing statistical procedures for a more realistic setting and also clarify when the use of such procedures would fail.


\section{Dynamic regimes under resource limitations} \label{sec: elab}

\subsection{Data structure}
Consider an observed data setting where baseline clinical features $L_i$, binary treatments $A_i$, and outcome $Y_i$ are observed for $i=1,\dots,n$ individuals. In an iid setting, each vector $O_i\equiv \{L_i, A_i, Y_i\}$ is viewed as an independent draw from a common law $P$ with support $\mathcal{L} \times \{0,1\} \times \mathbb{R}$.\footnote{We consider $L_{i}$ discrete, in order to simplify notation, but the results can be generalized to continuous covariates.} We use superscripts to denote counterfactual (potential outcome) variables under a treatment regime. For example, $Y_{i}^{g}$ is the counterfactual outcome that would occur under regime $g$ for individual $i$. We use plus symbols ($+$) to distinguish \textit{natural} values \citepMain{robins2004effects} of counterfactual intervention variables under a regime  ($A_i^g$) from \textit{assigned} values of such variables ($A_{i}^{g+}$). When we consider iid settings we drop individual indices. Except when emphasis is needed, we omit subscripts that indicate when a function depends on a particular law, say, $P$. 

\subsection{Conventional individualized regimes}
An IR can be defined as a mapping ($g$) of an individual $i$'s characteristics $l$ and a randomizer $\delta \in [0,1]$ to a single treatment level, $a$. We refer to such regimes as \textit{individualized}, as they are a function solely of a single individual's features. Suppose $O^F_i\equiv \{O_i^g \mid g \in \Pi\}$ collects an individual $i$'s counterfactuals across all IRs $g$ in an unconstrained policy class $\Pi$, and that $P^F$ denotes the joint distribution over these variables. To reduce notation, we omit superscripts $F$ distinguishing $P$ vs. $P^F$, except as needed.

A given IR results in only a proportion of a target population receiving treatment, $\mathbb{E}[A^{g+}] \in [0,1]$.  Thus, a policy maker can consider a class of regimes defined to satisfy the constraint, $\Pi(\kappa)\coloneqq \{g\in \Pi \mid \mathbb{E}[A^{g+}] \leq \kappa\}$, and find the regime $g^{\mathbf{opt}}$ therein that optimizes some expected utility:  $g^{\mathbf{opt}} \coloneqq   \underset{g \in \Pi(\kappa) }{\arg \max } \  \mathbb{E}[Y^g]$. 
As reviewed in \citetMain{qiu2020optimal}, determination of $g^{\mathbf{opt}}$ admits a closed form solution \citepMain{dantzig1957discrete} up to an equivalence class defined by the counterfactual propensities,  $q^g(l)\coloneqq \mathbb{E}[A^{g+} \mid L=l]$. Letting $\Delta(l) \coloneqq  \mathbb{E}[Y^{a=1} - Y^{a=0} \mid L=l]$, and $\eta \coloneqq \max\Big\{\inf\{ c \in \mathbb{R} \mid P(\Delta(L) > c) \leq \kappa\}, 0\Big\}$, then $g^{\mathbf{opt}}$ is chosen as a regime such that
\begin{align*}
    q^{\mathbf{opt}}(l) = 
        \begin{cases}
            \frac{\kappa - P(\Delta(L) > \eta)}
                 {P(\Delta(L) = \eta)} 
                & \text{ if } \Delta(l) = \eta \text{ and } \Delta(l) > 0,\\
           I(\Delta(l) > \eta) & \text{ otherwise. } 
        \end{cases}
\end{align*}
Heuristically, $g^{\mathbf{opt}}$ assigns treatment to all individuals who would experience the greatest and strictly-positive conditional ATEs (CATEs) and who collectively constitute a proportion of the population less than or equal to the resource constraint $\kappa$. If this proportion is less than $\kappa$, then treatment is randomly assigned to the remaining group with the largest strictly-positive CATE so that $q^{\mathbf{opt}}$ exactly satisfies the constraint in expectation.

\subsection{Why iid data models are inadequate in settings with resource limitations}

The constrained optimization approach assumes iid observed data structures. However, if existing data arise from a setting where resources are already limited, as will often be the case, then iid models will immediately be violated. Consider, for example, the causal directed acyclic graph (DAG) in Figure \ref{fig: CFRLM}. Throughout we will use such graphs to depict Finest Fully Randomized Causally Interpreted Structured Tree Graph (FFRCISTG) models of \citetMain{robins1986new}, a nonparameteric structural equation model generalizing that of \citetMain{pearl2009causality}. The sub-DAG  involving only the nodes $O_j\equiv\{L_j, A_j, Y_j\}$  canonically represents causal relations between an individual's variables for a setting without unmeasured confounding, which is often used to justify identification of causal parameters based on graphical criteria. Although rarely done in practice, DAGs can depict relations between variables in a sample, $\{O_1,\dots, O_n\}$ \citepMain{ogburn2014causal}. Under an iid observed data model, the sub-DAG involving only the black arrows would be a valid causal graph, wherein the graphical motif for individual $j$ is repeated for all $n$ individuals, and there are no causal connections between individuals. However, if a resource limitation is operative in the observed data, then this iid sub-DAG will not faithfully represent the causal structure: if one individual uses or occupies a scarce treatment, such as a ventilator unit in an ICU, then this treatment is not available to an hitherto untreated individual. We can graphically depict these causal connections  with the additional red arrows . 
This causal structure contradicts independence between an arbitrary pair of individuals $\{O_i, O_j\}$; minimally, there is a direct causal path between one individual's covariates and another individual's outcome, mediated by the resource limitation, for example, 
$L_1\rightarrow A_1 \textcolor{red}{\boldsymbol{\rightarrow}} A_j \rightarrow Y_j$. 
These causal connections are related to the concept of interference between individuals, see for example \citetMain{sobel2006randomized, hudgens2008toward, savje2021average}, and we provide additional remarks in Section \ref{sec: ID}. Other inter-individual causal connections will also commonly arise in resource-limited settings; for example, a treatment provider will often consider the clinical features of other individuals in, say, an emergency room or military field hospital, when deciding whether to provide a scarce treatment unit. Such causal connections are depicted by the additional orange arrows in Figure \ref{fig: CFRLM}.

\begin{figure}[h] 
\centering
\begin{tikzpicture}[inner sep=0.3, outer sep=0.3, scale=.5]
        \node (Aj)    at  (4   *3, 0  ) {$A_j$};
        \node (A1)    at  (4   *3, 4   ) {$A_{1}$};
        \node (An)    at  (4   *3,-4   ) {$A_{n}$};

        \node (Lj)    at  (1   *3, 0  ) {$L_j$};
        \node (L1)    at  (1   *3, 4   ) {$L_{1}$};
        \node (Ln)    at  (1   *3,-4   ) {$L_{n}$};

        \node (Yj)    at  (7   *3, 0)    {$Y_j$};
        \node (Y1)    at  (7   *3, 4   ) {$Y_{1}$};
        \node (Yn)    at  (7   *3,-4   ) {$Y_{n}$};

        \node (dots)    at  (1   *3, 2  )  {$\vdots$};
        \node (dots)    at  (1   *3,-2   ) {$\vdots$};
        \node (dots)    at  (4   *3, 2  )  {$\vdots$};
        \node (dots)    at  (4   *3,-2   ) {$\vdots$};        
        \node (dots)    at  (7   *3, 2  )  {$\vdots$};
        \node (dots)    at  (7   *3,-2   ) {$\vdots$};

\begin{scope}

 
        \path (Lj)          edge[bend right=15]  (Yj);
        \path (Lj)          edge[              ]  (Aj);
        \path (Aj)          edge  (Yj); 
        
        \path (L1)          edge[bend right=15 ]  (Y1);
        \path (L1)          edge                (A1);
        \path (A1)          edge                   (Y1);
        
        \path (Ln)          edge[bend right=15 ]  (Yn);
        \path (Ln)          edge                  (An);
        \path (An)          edge                   (Yn); 
        
        \path (A1)     edge[bend left=25, color=red, opacity=0.8 ]  (Aj);
        \path (A1)     edge[bend left=45, color=red, opacity=0.8]   (An);
        \path (Aj)     edge[bend left=25, color=red, opacity=0.8]   (An); 

        \path (L1)     edge[ color=orange, opacity=0.8]   (Aj);
        \path (L1)     edge[ color=orange, opacity=0.8]   (An);
        \path (Lj)     edge[ color=orange, opacity=0.8]   (A1);
        \path (Lj)     edge[ color=orange, opacity=0.8]   (An);
        \path (Ln)     edge[ color=orange, opacity=0.8]   (A1);
        \path (Ln)     edge[ color=orange, opacity=0.8]   (Aj);
 \end{scope}       
        
\end{tikzpicture}

\caption{DAG for a cluster of size $n$. 
} 
\label{fig: CFRLM}
\end{figure}
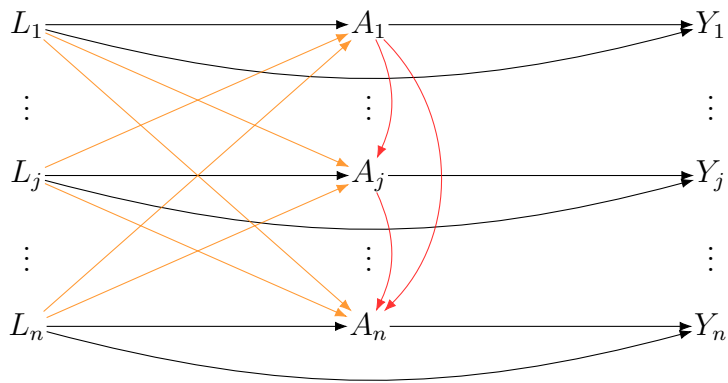

\subsection{A new policy-class for resource-limited settings}
\label{sec:policyclass}

We use bold font to denote a set of variables from $i=1,..,n$, for example,  $\mathbf{O}_n \equiv \{O_1,\dots, O_n\}$ denotes the factual data for all $n$ individuals in a cluster comprising an observed data sample, and let $\mathbb{P}_n$ denote its distribution. Estimands defined by $g\in\Pi(\kappa)$ are \textit{individualized} in the sense that a treatment provider implementing this regime for a particular individual $i$ need \textit{only} consider the clinical features of that particular individual, $L_i$. Instead, we consider regimes $G_n$ that are functions mapping all $n$ individuals' characteristics $\mathbf{L}_n$ and a randomizer $\delta \in [0,1]$ to a vector of $n$ treatments, $\mathbf{A}_n \in \{0, 1\}^n$. These randomizers can, for example, constitute lottery results or random-number-generators used to break ties under equipoise. We refer to a regime $G_n$ as a \textit{cluster} dynamic regime  (hereby CR), as they are a function of \textit{all} individuals' features in the cluster, and we let $\Pi_n$ be the set of all possible CRs, analogous to $\Pi$ for IRs.  We let $\mathbf{O}_n^F\coloneqq \{\mathbf{O}^{G_n} \mid G_n \in \Pi_n\}$ collect all $n$ individuals' counterfactuals under all possible CRs $G_n$, and we let $\mathbb{P}_n^F$ denote its joint distribution.

In contrast to $\Pi(\kappa)$, we specifically consider the policy class $$\Pi^F_n(\kappa_n) \equiv \Big\{G_n: \lVert G_n(\mathbf{l}_n, \delta) \rVert \leq \kappa_n  \text{ for all} \ \mathbf{l}_n, \delta \Big\},$$
and its subclass
$$\Pi_n(\kappa_n) \equiv \Big\{G_n: \lVert G_n(\mathbf{l}_n, \delta) \rVert = \kappa_n  \text{ for all} \ \mathbf{l}_n, \delta \Big\},$$
where $ \kappa_n \in \{0,1, \dots, n\}$ represents the number of treatment units available.   In words,  $\Pi^F_n(\kappa_n)$ is the class of CRs where no more than $\kappa_n$ individuals receive treatment and $\Pi_n(\kappa_n)$ is the sub-class where exactly $\kappa_n$ receive treatment.  To simplify notation, we focus on the class $\Pi_n(\kappa_n)$ and defer relaxations and extensions to Appendix \ref{appsec: gencdtrs}.

A feature of $\Pi^F_n(\kappa_n)$ and $\Pi_n(\kappa_n)$ is that they are not dependent on the law $\mathbb{P}_n^F$. In contrast, $\Pi_{P^F}(\kappa) \coloneqq \Pi(\kappa)$ crucially depends on the law $P^F$; the class of IRs guaranteed to satisfy a constraint $\kappa$ across all laws contains only the regimes that randomly assign treatment with some probability less than or equal to $\kappa$, including the trivial deterministic regime that assigns $a=0$ to all individuals. In Appendix \ref{appsec: misc}, we elaborate on this law dependence, and its corresponding consequences for robustness and transportability.

\subsubsection{Formulating cluster regimes based on rankings} \label{sec: rank}

In practice, regimes for limited resource settings are often formulated as rules for assigning a rank ordering to a finite cluster of individuals, and this rank ordering determining the assignment of the finite number of treatment units. To accord with this formulation, we define an unmeasured variable $\mathbf{R}_n$ that denotes this rank ordering. We then consider a elaborated CR $G^*_n\equiv\{G^*_{n, \mathbf{R}}, G^*_{n, \mathbf{A}}\}$ that involves both an intervention on $\mathbf{R}_n$ and $\mathbf{A}_n$ via $G^*_{n, \mathbf{R}}$, and $G^*_{n, \mathbf{A}}$, respectively. The assigned ranks $\mathbf{R}_n^{G_n^*+}$ determine the order in which individuals are assigned treatment, $\mathbf{A}_n^{G_n^*+}$, which proceeds until all $\kappa_n$ treatment units are allocated. 
A regime that assigned treatment whenever a resource was available, $A_{i}^{G^*_{n}+} = I(\kappa_n\geq R_i^{G_n^*+})$, would guarantee exact utilization of the available treatments. Thus each $G_n \in \Pi_n(\kappa_n)$ may be characterized by a rank-function (or prioritization rule) $G^*_{n, \mathbf{R}}$. When policy-makers and subject matter researchers reason about candidate regimes in terms of such prioritization rules, then this property may facilitate sound translation of policy questions into appropriate analyses.

\section{Identifying resource-limited estimands} \label{sec: ID}

\subsection{Structural conditions}\label{subsec: dimred}

We argued that iid statistical models for $\mathbb{P}_n$ are inadequate for capturing the defining features of resource-limited settings. 
Here we describe a statistical model for $\mathbb{P}_n$ that is coherent for these settings, and demonstrate that relevant estimands are identifiable and estimable under this model using available data. We do so by first articulating a general unrestricted non-parametric structural equation model (SEM) \citep{pearl2009causality, richardson2013single}. Under an unrestricted SEM, data are generated through recursive evaluation of the following structural functions for all $i=1,\dots,n$:
\begin{align*}
    L_{i}        & = f_{L_{i}}        (Pa(L_{i})        , \epsilon_{L_{i}}), \\
    A_{i}        & = f_{A_{i}}         (Pa(A_{ i})        , \epsilon_{A_{i}}), \\
    Y_{i}        & = f_{Y_{i}}        (Pa(Y_{i})        , \epsilon_{Y_{i}}). 
\end{align*}
For example, $A_i$ is determined by the structural function $f_{A_i}$ of variables in the parent set $Pa(A_i)$, and an error term, $\epsilon_{A_i}$. We suppose that elements of $\mathbf{A}_n$ are generated subsequent to $\mathbf{L}_n$ and prior to $\mathbf{Y}_n$. Formally, we let $\mathbf{f}_n \equiv \{f_{L_i}, f_{A_i}, f_{Y_i}\mid i\in\{1,\dots,n\}\}$ denote the set of all relevant structural equations, and $\boldsymbol{\epsilon}_n \equiv \{\epsilon_i\mid i\in\{1,\dots,n\}\}$ denote the set of all error terms following a joint distribution $\mathbb{P}^{\epsilon}_n$, with $\epsilon_{i} \coloneqq \{\epsilon_{L_i}, \epsilon_{A_i}, \epsilon_{Y_i}\}$. Following \citet{pearl2009causality}, a particular causal system $\mathbb{S}_n$ is then taken to be a pairing of a set of structural equations and a distribution for their associated error terms, $\mathbb{S}_n \equiv \{\mathbf{f}_n, \mathbb{P}^{\epsilon}_n$\}.

Unlike SEMs used in settings with iid observations \citepMain{pearl2009causality, richardson2013single}, the above structural equations allow the inclusion of variables encoding information on individual $j\ne i$ in the structural equations for some $V_{i}$. For example, this general SEM allows that $L_{j}\subseteq Pa(A_{i})$, that is, it allows the covariates for individual $j$ to have direct effects on the treatment of individual $i$. Furthermore the structural equations are themselves indexed by $i$, reflecting that this general SEM does not assume any invariance between structural equations for $V_{j}$ and $V_{i}$. For example,  $V_{j}$ might differ from $V_{i}$ even if $Pa(V_{j})=Pa(V_{i})$ and $ \epsilon_{V_{j}}=\epsilon_{V_{i}}$. 

We then  define a statistical model $\mathcal{M}_n$, i.e. a collection of laws $\mathbb{P}_n$, through restrictions on the SEM, i.e. through restrictions on $\mathbf{f}_n$ and $\mathbb{P}^{\epsilon}_n$. A particular causal system $\mathbb{S}_n$ induces a joint distribution for $\mathbf{O}^F_n$, that is $\mathbb{P}_n \equiv \mathbb{P}_n(\mathbb{S}_n)$. Thus, if we take $\mathcal{S}_n$ to be a set of causal systems, then its corresponding induced statistical model is defined as $\mathcal{M}_n\equiv \{\mathbb{P}_n(\mathbb{S}_n) \mid \mathbb{S}_n\in\mathcal{S}_n\}$.

To identify causal parameters of interest, we introduce two types of conditions. The first type are weak structural conditions for the consistent estimation of parameters of $\mathbf{O}_n$. These conditions, labelled $A1-A3$ below, are a strict subset of those that are implicit whenever data are described as iid, but they must be made explicit in settings where iid models are unsatisfactory. 

\begin{itemize}    
    \item [\textit{A1}.] \textbf{Conditional noninterference:} For all $i$,  $Pa(L_{i}) \cup Pa(Y_{i}) \subseteq O_i$ 
and $\epsilon_{Y_{i}}\independent \boldsymbol{\epsilon}_n \setminus \epsilon_i$. Furthermore, $\{\epsilon_{L_j}: j\in \{1,\dots,n\}\}$ are mutually independent.
\item [\textit{A2}.] \textbf{Structural invariance:} For all $i\neq j$, $f_{L_{i}} = f_{L_{j}}$ and $f_{Y_{i}} = f_{Y_{j}}.$
\item [\textit{A3}.] \textbf{Conditional identically-distributed errors:} For all $i \neq j$, and for all $pa(y)$, 
        $\Big(\epsilon_{L_{i}} \Big) \sim \Big(\epsilon_{L_{j}}\Big)$, and
        $\Big(\epsilon_{Y_{i}} \Big| Pa(Y_{i})=pa(y)\Big) \sim \Big(\epsilon_{Y_{j}} \Big| Pa(Y_{j})=pa(y)\Big)$.
\end{itemize}

The second type permits identification of causal parameters of  $\mathbf{O}_n^F$ solely in terms of parameters of $\mathbf{O}_n$. This condition, labelled $B$,  is isomorphic to those classically assumed for identification of individual-level parameters in the statistical causal inference literature. 

\begin{itemize}    
    \item [\textit{B}.] \textbf{No individual-level confounding of outcomes:} For all $i$,  $Y_i^{a_i} \independent A_i \mid L_i$.
\end{itemize}

We denote the statistical model for $\mathbb{P}_n$ that is induced by an NPSEM under conditions \textit{A} and \textit{B} by $\mathcal{M}_n^{AB}$.

\subsection{Remarks on conditions}\label{subsec: IDcondRemarks}

We illustrate the structural implications of Condition \textit{A1} with DAGs in Figure \ref{fig: CondIll1} of Appendix \ref{appsec: misc}. Condition \textit{A1} will in general be violated by a causal effect of any variable for individual $j$, $O_j$, on the covariates or outcome for individual $i$, $\{L_i, Y_i\}$, except via paths intersected by the treatment of individual $i$, $A_i$. For example, \textit{A1} does not preclude an effect depicted by the path $L_j \rightarrow A_i \rightarrow Y_i$: such an effect would be expected in a limited resource setting. However, a direct path $A_j \rightarrow Y_i$ would violate condition  \textit{A1}. 
Furthermore, Condition \textit{A1} will in general be violated by unmeasured common causes of $Y_i$ and any element in $O_j$, or any unmeasured common cause of $L_i$ and $L_j$. This is encoded by restrictions on the dependence structure of the error terms for $Y_i$ and $L_i$ with those error terms for individual $j$. However, Condition \textit{A1} leaves unrestricted the causal and error-dependence structure between the treatment variables of all individuals in the cluster, $\mathbf{A}_n$, and between the covariates $\mathbf{L}_n$ and those treatments. Thus, extensive causal connections between individuals are permitted under this condition.  Conditions \textit{A2} and \textit{A3} assert that covariate and outcome structural equations are invariant in $i$ and that the conditional distributions of their error terms are identical across individuals. Together, Conditions \textit{A1}-\textit{A3} provide a structure and regularity to the joint distribution of $\mathbf{O}_n$ that is nevertheless markedly weaker than standard iid models.

We graphically illustrate the structural implications of Condition \textit{B} with the DAGs in Figure \ref{fig: CondIll2}  of Appendix \ref{appsec: misc}. In the context of an SEM following Conditions \textit{A1}-\textit{A3}, Condition \textit{B} additionally precludes classical within-individual confounding structures, and also places additional restrictions on the causal structure between the covariates $\mathbf{L}_n$ and treatments $\mathbf{A}_n$. In particular, when there is an unmeasured common cause of $L_i$ and $Y_i$, unmeasured common causes of $L_i$ and $A_j$ will result in a variant of the classical M-bias causal structure \citepMain{greenland1999causal}. In this case, the M-bias arises because of the causal connections between treatment units $A_j$ and $A_i$, characteristic of the resource-limited setting. 

By definition, $Y_i^{a_i}$ is the potential outcome of individual $i$ under an intervention \textit{solely} on that individual's own treatment (and no other individuals). Therefore, Condition \textit{B} is amenable to interrogation on conventional graphs constructed for a single individual, for example the single-world intervention graphs (SWIG) of \citetMain{richardson2013single}, wherein this condition might be contradicted by graphical structures using standard rules for d-separation \citepMain{pearl2009causality}, see Figure \ref{fig: CondIll3} for an example. However, under Condition \textit{A1} we might immediately reformulate Condition \textit{B} in terms of a counterfactual outcome under intervention to \textit{all} individuals' treatments, $Y_i^{\mathbf{a}_n} \independent A_i \mid L_i$. This follows from the restrictions on the structural equations for $Y_i$, whereby it follows that $Y^{a_i}_i= Y_i^{\mathbf{a}_n}$ for all $i$, due to the exclusion of $\mathbf{a}_n\setminus a_i$ or any of their descendants besides $A_i$ from the arguments of $f_{Y_i}$. In this sense, counterfactuals $Y_i^{a_i}$ generated by such an SEM satisfy the stable-unit treatment value assumption of \citetMain{rubin1990comment}. However, other counterfactuals from this model, defined by interventions on other variables, will in general not. For example, an intervention on the patients' rank-ordering for treatment $\mathbf{R}_n$, without subsequent intervention on patients' treatment assignment would not satisfy SUTVA, i.e./ $Y^{r_i}_i \neq Y_i^{\mathbf{r}_n}$. While we do not consider further such interventions, they are the focus of ongoing work.

\subsection{Identification results}\label{subsec: IDres}

Under $\mathcal{M}^{AB}_n$, covariates and outcomes enjoy a conditional iid property, see Appendix \ref{appsec: proofs} for a formal result. For example, conditional on the measured past, the observed outcomes are independent and the probability that any two individuals with the same covariate history $a, l$ will have the particular outcome value $y$ is described by a common measure. We let $Q_{Y}(y \mid a, l) :=  \mathbb{P}_n(Y_i=y \mid L_i=l, A_i=a)$ denote this common measure  and likewise define  $Q_{L}(l) :=  \mathbb{P}_n(L_i=l)$.  
Thus far we have not discussed any restrictions on the relations between the parameters $q_i(a \mid l) \coloneqq \mathbb{P}_n(A_i=a \mid L_i=l)$, across individuals $i=1,..,n$.  Likewise, $q^*_i(a \mid l) \coloneqq \mathbb{P}_n(A^{G_n+}_i=a \mid L_i=l)$ is similarly unrestricted. There are no guarantees that $q_i=q_j$ for an arbitrary $\mathbb{P}_n$ in $\mathcal{M}^{AB}_n$ or that $q^*_i=q^*_j$ for some arbitrary $G_n\in\Pi_n(\kappa_n)$.  In Section \ref{sec: estim}, we discuss conditions under which we can construct consistent estimators of individual-level parameters $Q_L$ and $Q_Y$ from $\mathbf{O}_n$. When such is the case, we will say we have identified a causal parameter whenever we express it as a known function of at most parameters of $Q_L$ and $Q_Y$.

Analogous to classical approaches in causal inference \citepMain{richardson2013single}, we define an \textit{individual-level} g-formula density function $f_i^{G_n}$, with respect to an individual $i$ and $CR$ $G_n$, 
\begin{align}
      f_{i}^{G_n}(o) = & Q_{Y}(y \mid a, l) q^{*}_{i}(a \mid l)  Q_L(l), \label{eq: patlevgform}
\end{align}
which will appear in our identification results. We emphasize that $f_{i}^{G_n}$ is distinguished from a classical g-formula in that it will generally differ for each individual $i$, via the intervention density $q^{*}_{i}$. 

We also define lower-dimensional summaries of $\mathbf{O}_n$, corresponding to the frequency distributions of the values of $O_i$ in the cluster. First, let $\mathbb{L}_{n}$ be the random vector indicating the numbers of individuals in the cluster with each of the possible values $l \in \mathcal{L}$. Second, let $\mathbb{L}_{n}(l)$  denote the one-dimensional random variable indicating the number of individuals in the cluster with specific covariate value $l$. 
Then $\mathbb{L}_{n}$ represents the random marginal frequency table for the cluster with fixed $n$, and $\mathbb{L}_{n}(l)$ represents the random value of one of its cells corresponding to value $l$. Likewise, $\mathbb{A}_{n}$ and $\mathbb{A}_{n}(a)$ have analogous interpretations. Elaborating this notational convention, let $B_i \coloneqq \{A_i, L_i\}$ so that $\mathbb{B}_{n}$ is the random contingency table for joint frequencies of treatment and covariates in the cluster, which we may consider conditional on some or all of its margins, for example, conditional on $\mathbb{L}_{n}= \mathbb{l}_n$. As short hand, let $\mathbb{Q}_{Y, n}$ be the conditional density function of $\mathbb{O}_n$ given $\mathbb{B}_n$ and let $\mathbb{Q}_{L, n}$ be the marginal density of $\mathbb{L}_n$. Additionally, let $\mathbb{q}^*_n$ denote the conditional density of $\mathbb{B}_{n}^{G_n+}$ given $\mathbb{L}_n$. See Appendix \ref{appsec: proofs} for identities of these compositional parameters in terms of $Q_L$ and $Q_Y$ under $\mathcal{M}^{AB}_n$. 

We thus also define the \textit{compositional} g-formula density function  $\mathbb{f}^{G_n}$ as 
\begin{align} \label{eq: compgform}
   \mathbb{f}^{G_n}(\mathbb{o}_{n}) \coloneqq \mathbb{Q}_{Y,n}(\mathbb{o}_{n} \mid \mathbb{b}_{n})
                \mathbb{q}^*_n(\mathbb{b}_{n} \mid  \mathbb{l}_{n})\mathbb{Q}_{L,n}(\mathbb{l}_{n}).
\end{align}
We refer to $\mathbb{f}^{G_n}$ as the compositional g-formula, because it is isomorphic to the individual-level g-formula, except with functions for compositional parameters in place of individual-level parameters. 

The following theorem links the \textit{individual-level} and \textit{compositional} g-formulae with causal parameters under $G_n$, where we consider arbitrary real-valued functions $h$ and $h'$ of $\mathbb{O}_n$, and $O_i$, respectively.

\begin{theorem} \label{theorem: CDTRID}
Consider a law $\mathbb{P}_n\in\mathcal{M}^{AB}_n$  
and a regime $G_n \in \Pi_n(\kappa_n)$. 
Then   $$\mathbb{E}[ h(\mathbb{O}_n^{G_n+})] =  \sum\limits_{\mathbb{o}_{n}}h(\mathbb{o}_n)  \mathbb{f}^{G_n}(\mathbb{o}_{n}) \ 
 \text{and} \ \mathbb{E}[ h'(O_i^{G_n+})] =  \sum\limits_{o}h'(o)f_{i}^{G_n}(o),$$ 
whenever the right hand sides of the equations are well-defined. 
\end{theorem}
A proof of Theorem \ref{theorem: CDTRID} is provided in Appendix \ref{appsec: proofs}. We also provide a suitable positivity condition for the identifying functionals in Theorem \ref{theorem: CDTRID}. Let $\overline{q}_n(a \mid l)  \coloneqq \frac{1}{n}\sum\limits_{i=1}^nq_i(a\mid l)$ denote the average conditional density function of $A_i$ given $L_i$ over the indices $i$, and likewise define $\overline{q}^*_n$. Then, these functionals will be well-defined when the following condition holds:

\begin{itemize}
    \item [$C0.$] \textbf{Weak finite-cluster positivity} If $\overline{q}^*_n(a \mid l)Q_L(l) > 0 $ then $\overline{q}_n(a \mid l)>0,$ for all $a,l.$
\end{itemize}
In words, whenever a treatment and covariate history $a, l$ has positive probability under $G_n$ for at least one individual $i$, then there must exist at least one individual $j$, possibly different from $i$, that has a positive probability of receiving that treatment level $a$ given $l$ under the observed data law. 

Analogous positivity conditions are typically considered in the causal inference literature for g-formula identification of parameters defined by IRs under iid models, see Expression (62) in \citetMain{richardson2013single}. In these settings, positivity implies the existence of an asymptotically consistent estimator under mild regularity conditions. However, stronger conditions will be necessary for the existence of such an estimator when a single draw from $\mathbb{P}_n$ comprises the observed data.  

\subsubsection{$L$-rank-preserving regimes}

Theorem \ref{theorem: CDTRID} provides an identification result for an arbitrary CR in $\Pi_n(\kappa_n)$. As noted in Section \ref{sec: rank}, any such regime can be characterized by a function $G^*_{n, \mathbf{R}} : \mathcal{L}^n \times [0,1] \rightarrow S(\{1,\dots, n\})$.

As an example, consider a pair of clusters with $n=4$ and $\kappa_n=1$ and $\mathcal{L}\equiv\{0,1\}$, where in the first cluster  $\mathbf{L}_n=(1, 1, 0, 0)$ and in the second cluster $\mathbf{L}_n=(0, 0, 1, 1)$. Then consider a regime $G_n$ such that for both such clusters, the rank orderings of $\mathbf{r}_n=(1,2,3,4)$ and $\mathbf{r}_n=(2,1,3,4)$ are assigned with equal probability. Each of these clusters has the same composition: $\mathbb{L}_n=(2,2)$. However, in the first cluster, individuals with $L_i=1$ each have a probability of 0.5 of treatment assignment and those with $L_i=0$ have 0 probability of treatment assignment, whereas in the second cluster these probabilities are exchanged.

In contrast to the preceding example, realistic policies will often correspond to regimes that assign treatment in a way that is invariant across clusters with the same composition. To formalize this class of regimes, we then consider a subclass of policies, $ \Pi^d_n(\kappa_n, L_i, \Lambda)\subset \Pi_n(\kappa_n)$, defined as follows:

\begin{definition}[$L$-rank-preserving CRs]\label{def: LRPregimes}
The class of CRs $\Pi^d_n(\kappa_n, L_i, \Lambda)$ is the subset of regimes $G_n$ in  $\Pi_n(\kappa_n)$ such that for all $i, j$ and all $\mathbb{P}_n$ the following properties hold:
\begin{enumerate}
    \item [$(1)$] There exists a function $\Lambda:  \mathcal{L}_{i} \rightarrow \mathbb{R}$ where $\Lambda(L_{i}) > \Lambda(L_{j})$ implies $ A_i^{G_n+} \geq A_j^{G_n+}$, $\mathbb{P}_n$-almost surely; and
    \item [($2)$] $\mathbb{P}_n(A^{G_n+}_i = 1 \mid L_i=l, \mathbb{L}_n=\mathbb{l}_n) = p^{G_n}(l, \mathbb{l}_n)\in [0,1]$ for all $l, \mathbb{l}_n$.
    
\end{enumerate}
\end{definition}

$L$-rank-preserving CRs constitute an important class of regimes. For example, the prioritization rules deployed in national health care systems are of this class \citepMain{lakshmi2013application, januleviciute2013impact, mullen2003prioritising,  hadorn1997new, feek2000rationing}, as they have the feature of sorting individuals into coarsened rank groups via some function $\Lambda$, i.e., property $(1)$. Furthermore, $L$-rank-preserving CRs capture some notion of equity that is not in general satisfied by CRs, as illustrated by the previous example; alternatively, $L$-rank-preserving CRs ensure that individuals with the same covariate value $l$ will have the same probability of receiving treatment, $q^*_i=q^*_j$, for all $i,j$. Finally, the optimal CR will be an $L$-rank-preserving CR under $\mathcal{M}^{AB}_n$. We provide a formal result in Section \ref{sec: estim}.

We now present an additional identification result for the expectations of $\overline{Y}^{{G}_n}_n \coloneqq \frac{1}{n}\sum\limits_{y} y\mathbb{Y}_n^{{G}_n}(y)$, the cluster-average potential outcome. The resulting identification formula has a unique formulation since $\overline{Y}^{{G}_n}_n$ can be re-expressed as a simple linear combination of individual outcomes, $\frac{1}{n}\sum\limits Y_i^{{G}_n}$.

 \begin{proposition}
     \label{theorem: YbarID}
    Consider a law $\mathbb{P}_n \in\mathcal{M}^{AB}_n$ and $G_n \in \Pi^d_n(\kappa_n, L_i, \Lambda)$. Then $$\mathbb{E}[ \overline{Y}^{{G}_n}_n] =  \sum\limits_{o\in\mathcal{O}_i}yf_{i}^{G_n}(o).$$
\end{proposition}
In subsequent sections, we study the asymptotic properties of these functionals  as the cluster size $n$ grows and find correspondences with identification approaches in the constrained optimization literature. We also consider a more general version of Proposition \ref{theorem: YbarID} in Appendix \ref{appsec: gencdtrs}, for a larger class of CRs defined solely by property (2) of Definition \ref{def: LRPregimes}.

\subsubsection{Identifying intervention densities}

Theorem \ref{theorem: CDTRID} and Proposition \ref{theorem: YbarID} do not by themselves constitute identification results because $q_n^*$ and $\mathbb{q}^*_n$ are by definition counterfactual densities. Under classical IRs, this distinction is trivial because $q_i^*$ would be known  and invariant to $\mathbb{P}_n$, following immediately by definition of the regime. In contrast, the $q_i^*$ and $\mathbb{q}_n^*$ of CRs are not immediately apparent. 
In Appendix \ref{appsec: gencdtrs}, we review identification of $\mathbb{q}^*_n$ and $q_n^*$ under general regimes $G_n$ and unrestricted non-parametric models, and illustrate simplifications under regimes in special classes. We provide a simple closed form expression for $q^*_{i}$ under $L$-rank preserving regimes, $\Pi^d_n(\kappa_n, L_i, \Lambda)$. This allows specification of $G_n$ in terms of ${\Lambda}$, which as we argued, characterizes many regimes of practical interest. 
 
\section{Large-cluster estimands} \label{sec: largecluster}

Finite cluster estimands are, in general, identified by functionals that depend on a fixed cluster size, say $n^*$. 
In this section, we articulate ``large-cluster'' estimands defined as the limits of particular sequences of finite-cluster estimands indexed by cluster size $n$. We will show that these large-cluster estimands can correspond to parameters defined by the IRs already considered in the extant literature, and indeed can be identified by identical functionals of observed data. 

\subsection{Asymptotic preliminaries}

We define an arbitrary joint distribution of $\{\mathbf{O}^F_1, \mathbf{O}^F_2, \mathbf{O}^F_3, \dots \}$, the joint vector of covariates, treatments, and outcomes for clusters of increasing size $n$. We denote this distribution by  $\mathbb{P}_{0}$, which we call an \textit{asymptotic law}. We let $\mathcal{M}_{0}^U$ denote the set of all such laws that can be factorized as $\mathbb{P}_{0} = \mathbb{P}_{1}\mathbb{P}_{2}\mathbb{P}_{3}\cdots$, and we locate $\mathbb{P}_n$ as a factor of this joint law, $\mathbb{P}_n \equiv \mathbb{P}_n(\mathbb{P}_0)$. 

For a given asymptotic law $\mathbb{P}_{0} \in \mathcal{M}_{0}^U$, we consider the asymptotic properties of sequences of estimands indexed by $n$. 
To facilitate consideration of these sequences, we restrict the space $\mathcal{M}_{0}^U$ via the following conditions, where the primitive is the unrestricted space of joint causal systems $\mathcal{S}_0 \equiv \mathcal{S}_1\mathcal{S}_2\mathcal{S}_3\cdots$. We let $\mathbb{P}_n^{\epsilon_{L_i}}$ and $\mathbb{P}_n^{\epsilon_{Y_i \mid pa(y)}}$ denote the marginal law of $\epsilon_{L_i}$ and the conditional law of $\epsilon_{Y_i}$ given $Pa(Y_i)=pa(y)$ for a cluster of size $n$, respectively. 
\begin{itemize}
    \item [$D1.$] $\mathbb{P}_n \in \mathcal{M}_n^{AB}$, for all $n\in\mathbb{N}^+.$
    \item [$D2.$] $\{f_{L_i}, \mathbb{P}_n^{\epsilon_{L_i}}, f_{Y_i}, \mathbb{P}_n^{\epsilon_{Y_i \mid pa(y)}}\} \equiv \{f_{L_i}, \mathbb{P}_{n'}^{\epsilon_{L_i}}, f_{Y_i}, \mathbb{P}_{n'}^{\epsilon_{Y_i \mid pa(y)}}\}$ for all  $n, n' \in\mathbb{N}^+.$
\end{itemize}
Condition $D1$ implies that each marginal law $\mathbb{P}_n$ of $\mathbb{P}_0$ is induced by a causal system $\mathbb{S}_n$ that satisfies the conditions of $\mathcal{M}_n^{AB}$. Condition $D2$ states that the structural equations for $L_i$ and $Y_i$ are invariant, as are the marginal and conditional distributions of their error terms, respectively, across $\mathbb{S}_1, \mathbb{S}_2, \mathbb{S}_3,\dots$. 
A consequence of conditions  $D1$ and $D2$ is that, for a given $\mathbb{P}_0$,  the single functions $Q_L$ and $Q_Y$ will describe the marginal and conditional densities of $L_i$ and $Y_i$, respectively, for all $i$, $n$. We henceforth take limits with respect to such an asymptotic law $\mathbb{P}_0$.

\subsection{Defining and identifying large-cluster estimands}

We define an arbitrary sequence of CRs, which we denote by $\{G_1, G_2, G_3, \dots\}$. We denote this sequence by  $\mathbf{G}_{0}$, which we call an \textit{asymptotic regime}. We consider asymptotic regimes $\mathbf{G}_{0}$ that meets the following conditions:

\begin{itemize}
    \item [$E1.$] For each $n$, $G_n\in \Pi^{d}_n(\kappa_n, L_i, \Lambda)$.
    \item [$E2.$] For each $n$, $\kappa_n = \lfloor n \times \kappa^* \rfloor$, with $\kappa^* \in [0,1].$
\end{itemize}

An asymptotic regime following condition $E$ is a sequence of $L$-rank-preserving CRs with a fixed $\Lambda$, and resource constraints such that $\frac{\kappa_n}{n}$ converges to a fixed proportion $\kappa^*$. We discuss relaxations of condition $E1$ in Appendix \ref{appsec: gencdtrs}. A class of asymptotic regimes following these conditions is of interest because it is parsimoniously specified by two parameters, $\Lambda$ and $\kappa^*$, which in many settings will sufficiently characterize an investigator's regime of interest. All asymptotic regimes $\mathbf{G}_0$ following conditions $E1$ and $E2$ with respect to the same $\Lambda$ and $\kappa^*$ will be equivalent under an asymptotic law $\mathbb{P}_0$ following conditions $D$, in the sense that they share a common limit of an arbitrary individual's intervention density,  $q^{*}_{0}(a \mid l) \coloneqq \underset{n\to \infty}{\lim}  q^*_{i}(a, l)$ for all $i$.

We define the following \textit{large-cluster} g-formula density function $f^{\mathbf{G}_0}$, 
\begin{align}
      f^{\mathbf{G}_0}(o) \coloneqq &
        Q_{Y}(y\mid a,l) q^{*}_{0}(a \mid l)Q_L(l). \label{eq: largepopgformA}
\end{align}
The following theorem identifies the large-cluster expected average potential outcome $$\mathbb{E}[ \overline{Y}_0^{\mathbf{G}_0}] \coloneqq \underset{n\to \infty}{\lim}\mathbb{E}_{\mathbb{P}_n}[ \overline{Y}^{{G}_n}_n].$$

\begin{theorem} \label{theorem: YbarIDlarge}
    Consider an asymptotic law $\mathbb{P}_0$ and an asymptotic regime $\mathbf{G}_0$ following conditions $D$ and $E$, respectively. Then, $\mathbb{E}[\overline{Y}^{\mathbf{G}_0}_0] = \sum\limits_{o\in\mathcal{O}_i}yf^{\mathbf{G}_0}(o)$ whenever the right-hand side is well-defined. 
\end{theorem}

As in Section \ref{sec: ID} we provide a suitable positivity condition for the identifying functional in Theorem \ref{theorem: YbarIDlarge} to be well-defined. Define ${q}_0(a\mid l) \coloneqq \underset{n\to \infty}{\lim} \overline{q}_n(a \mid l)$ with respect to an asymptotic law $\mathbb{P}_0$. Then, this functional will be well-defined when the following condition holds:

\begin{itemize}
    \item [$C0^*.$] \textbf{Weak large-cluster positivity} ${q}_0$ is well-defined and if $q^*_{0}(a \mid l)Q_L(l) > 0$ then ${q}_0(a \mid l)>0$, for all $a,l$.
\end{itemize}

\subsection{Identifying intervention densities}

As with the finite-cluster parameters, the equality of Theorem \ref{theorem: YbarIDlarge} does not by itself constitute an identification result because $q_0^*$ is by definition a counterfactual density. To identify $q_0^*$, we leverage the rank-based formulations of CRs in $\Pi^{d}_n(\kappa_n, L_i, \Lambda)$, where we let $\mathbb{P}_0(f(L_i))$ denote $\sum\limits_{\mathcal{L}}I(f(l))Q(l)$, with $f(l)$ a logical proposition.

\begin{proposition} \label{lemma: largeint}
Consider an asymptotic law $\mathbb{P}_0$ and an asymptotic regime $\mathbf{G}_0$ following conditions $D$ and $E$, respectively. Then, 
\begin{align}
         q^*_{0}(1 \mid l)  
             & =  \begin{cases}
                     \frac{\kappa^* -  \mathbb{P}_0(\Lambda(L_i) > \omega_{0})}{\mathbb{P}_0(\Lambda(L_i) = \omega_{0})}   & : \Lambda(l)= \omega_{0}, \\
                     I\Big(\Lambda(l) > \omega_{0}\Big) & : \text{otherwise},
                \end{cases} \label{eq: gstar_large}  
\end{align}
where $\omega_{0} \coloneqq \textup{inf}\bigg\{ c: \mathbb{P}_0(\Lambda(L_i)>c) \leq \kappa^* \bigg\}.$
\end{proposition}

In Proposition \ref{lemma: largeint}, $\omega_{0}$ is interpretable as the almost-surely last rank group $\Lambda(l)$ treated in a large cluster under  $\mathbf{G}_0$, with respect to $\mathbb{P}_0$.
Theorem \ref{theorem: YbarIDlarge} and Proposition \ref{lemma: largeint} taken together thus identify $\mathbb{E}[\overline{Y}^{\mathbf{G}_0}_0]$ in terms of $Q_L$ and $Q_Y$. 

\subsection{Large-cluster optimal regimes}

Let $V_i=c(L_i)$ be some coarsening of $L_i$. Suppose we are interested in the optimal CR that is at most a function of $\mathbf{V}_n$, that is the regime $G_n: \{\mathcal{V}\}^n \times [0,1] \rightarrow \{0, 1\}^n$ that maximizes $\mathbb{E}[\overline{Y}_n^{G_n}]$. 
Then let $\Delta(l) \coloneqq \mathbb{E}[Y_i^{a=1} - Y_i^{a=0} \mid V_i=c(l)]$ be the conditional average effect of treatment for individual $i$. 
In a slight abuse of notation we also let $\Delta$ denote the equivalent function with domain $\mathcal{V}$. 

\begin{proposition}[Optimal CR] \label{lemma: optimallarge}
Consider a law $\mathbb{P}_n\in\mathcal{M}_n^{AB}$. The CR  $G^{\mathbf{opt}}_{n}\equiv G_n\in\Pi_n(\kappa_n)$ that maximizes $\mathbb{E}[\overline{Y}_n^{G_n}]$ is the $V$-rank-preserving CR $G_n\in\Pi_n^d(\kappa_n, V_i, \Lambda)$ characterized by $\Lambda \equiv \Delta$.

Consider an asymptotic law $\mathbb{P}_0$ following condition $D$, and an asymptotic regime $\mathbf{G}_0$ following condition $E$. If $G_{n^*}\in\mathbf{G}_0$ is such an optimal regime for $\mathbb{P}_{n^*}(\mathbb{P}_0)$ for some $n^*>1$ then $G_{n}\in\mathbf{G}_0$ is the optimal such regime for $\mathbb{P}_{n}(\mathbb{P}_0)$ for all $n$. Let $q^{\mathbf{opt}}_0$ denote the intervention density under the asymptotic regime $\mathbf{G}^{\mathbf{opt}}_{0} \equiv \Big(G^{\mathbf{opt}}_{1}, G^{\mathbf{opt}}_{2}, G^{\mathbf{opt}}_{3},\dots\Big)$. Then, $\mathbb{E}[\overline{Y}_0^{\mathbf{G}^{\mathbf{opt}}_{0}}]$ is identified as in Theorem \ref{theorem: YbarIDlarge} and its intervention density $q^{\mathbf{opt}}_0$ is identified as in Proposition \ref{lemma: largeint} where we take $\Lambda = \Delta_0\coloneqq \Delta$.

\end{proposition}

Proposition \ref{lemma: optimallarge} illustrates that the observed data parameter identifying $\mathbb{E}[\overline{Y}_0^{\mathbf{G}^{\mathbf{opt}}_{0}}]$ is isomorphic to that parameter studied by \citetMain{luedtke2016optimal}. We emphasize that Proposition \ref{lemma: optimallarge} considers the optimal regime in $\Pi_n(\kappa_n)$. In general, this CR will not necessarily be the optimal regime in the unrestricted class of $\kappa_n$-constrained CRs, $\Pi_n(\kappa_n)$, except in special cases, for example if we restrict model $\mathcal{M}^{AB}_n$ to only include laws for which $\Delta(l)$ is  positive for all $l$. 
We provide a straight-forward extension of Proposition \ref{lemma: optimallarge}, for the unrestricted class of regimes $\Pi^F_n(\kappa_n)$ in Appendix \ref{appsec: misc}.  
\section{Inference} \label{sec: estim}

\subsection{Asymptotics for inference on limited resource estimands}

In iid settings, estimators are defined as a function of $n$ independent observations from a common law $P$. Correspondingly, asymptotic properties of these estimators are defined with respect to sequences of such observations \citepMain{rosenblatt1956central}. In limited resource settings,  estimators are defined as a function of a \textit{single} draw from a some law $\mathbb{P}_n$. As in the empirical processe literature \citepMain{van2000asymptotic}, we consider asymptotic properties in the context of sequences of such laws $\mathbb{P}_n$, which we consolidate in terms of an asymptotic law $\mathbb{P}_0$.  
Let the sets of random variables $\{\tilde{Q}_{Y, n}, \tilde{q}_n, \tilde{Q}_{L,n}\}$ denote the (conditional) empirical mean analogues of $\{Q_Y, q_i, Q_L\}$, for example $\tilde{Q}_{Y, n}(y \mid a, l)  \coloneqq \frac{\mathbb{O}_n(y, a, l)}{\mathbb{B}_n(a, l)}$. 

Unusual challenges arise under non-iid models in the development of estimators with desirable asymptotic properties. For example, we argued in Theorem \ref{theorem: CDTRID} that many cluster causal parameters are identified  under $\mathcal{M}_n^{AB}$ and a suitable positivity condition. The identification functional is expressed entirely in terms of $Q_Y$ and $Q_L$, and investigator-known quantities. We show in Appendix \ref{appsec: proofs} that the sets $\{L_i\in\mathbf{L}_n\}$ and $\{Y_i\in\mathbf{Y}_n \mid  A_i=a, L_i=l\}$ are each correspondingly mutually independent and identically distributed under $\mathcal{M}_n^{AB}$. Thus, an apparently obvious estimator of $f^{G_n}_i(o)$, would consist of plug-in estimators of its component factors $\{Q_L, Q_Y\}$ using their single-cluster empirical mean analogues $\{\tilde{Q}_{Y, n}, \tilde{Q}_{L,n}\}$.  However, such an estimator will not in general be asymptotically consistent (in a sense subsequently defined). 

In Appendix \ref{appsec: estim} we give an example of an asymptotic law following conditions $D$ and weak cluster positivity $C0$, under which $\tilde{Q}_{Y, n}$ \textit{does not} converge to $Q_Y$. Heuristically, the observed asymptotic law  $\mathbb{P}_0$ may result in the convergence of $\overline{q}_n$, but, for each $n$, the causal system $\mathbb{S}_n$ generating $\mathbb{P}_n$ may be characterized by a regime $G_n$ equivalent to a distribution over deterministic regimes with mass concentrated at extreme allocation strategies. Under each such regime, $\tilde{q}_n(a\mid l)$ might be close to 0 or 1 for each $(a, l)$, and so, for any given realization at any cluster size $n$, $\Big|\overline{q}_n(a\mid l) - \tilde{q}_n(a\mid l)\Big|$ will remain large. A similar phenomenon is highlighted in \citetMain{lee2021network}, in their illustration of ``spurious associations'' in network data. Consider a toy data generating mechanism in which the vectors $\mathbf{L}_n$ and $\mathbf{Y}_n$ each take the values $(1,\dots, n)$ or $(n,\dots, 1)$ as determined by independent flips of fair coins. Thus $\mathbf{L}_n \independent \mathbf{Y}_n$, but in any given realization (i.e. an estimator computed on $\mathbb{O}_n$), it will appear as if $L_i$ and $Y_i$ are perfectly correlated (negatively or positively). 

\subsection{Asymptotic consistency for finite-cluster estimands}

Consider the following revised positivity condition for a given asymptotic law  $\mathbb{P}_0$ and regime $G_{n^*}$, where we let $\overline{q}_0\coloneqq \underset{n\to \infty}{\lim}  \tilde{q}_n$.

\begin{itemize}
    \item [C1.] \textbf{Strong finite-cluster positivity}  $\mathbb{P}_0$ is such that $\overline{q}_0$ is well-defined and, for all $a, l$, if $\overline{q}^*_{n^*}(a \mid l)Q_L(l) > 0$ then $\overline{q}_0(a \mid l)>0$, $\mathbb{P}_0-$almost surely.
\end{itemize}

To define non-parametric estimators of identifying parameters, we use that under $\mathcal{M}^{AB}_{n^*}$, $\{\mathbb{Q}_{L, n^*}, \mathbb{q}^*_{n^*}, \mathbb{Q}_{Y, n^*}\}$ may be entirely expressed in terms of $\{Q_L, Q_Y\}$ and known quantities. We emphasize that we have deliberately labeled these parameters with a fixed cluster size $n^*$, indicating the size defining a particular finite-cluster estimand. Likewise define the compositional random variables $\{\tilde{\mathbb{Q}}_{L, n^*}, \tilde{\mathbb{q}^*}_{n^*}, \tilde{\mathbb{Q}}_{Y, n^*}\}$ analogously by substituting $\{\tilde{Q}_L, \tilde{Q}_Y\}$ for $\{Q_L, Q_Y\}$ wherever they appear in the corresponding parameters. Accordingly, we construct empirical analogues of compositional and individual-level g-formula parameters, $\tilde{\mathbb{f}}^{G_{n^*}}$ and $\tilde{f}_i^{G_{n^*}}$. 
In Appendix \ref{appsec: estim} we also motivate estimators based on an inverse-probability weighted (IPW) formulation of the individual-level g-formula. In the following, we use the symbol $\boldsymbol{\lim}$ to denote limits in probability with respect to a measure $\mathbb{P}_0$, as $n\to\infty$.

\begin{theorem}[Asymptotic consistency] \label{theorem: asymptcons}
    Consider $\mathbb{P}_0$ following conditions $D$ and $G_{n^*}\in \Pi_{n^{\ast}}(\kappa_{n^{\ast}})$ following $C1$. Then,
    \begin{align}
        & \boldsymbol{\lim}\sum\limits_{\mathbb{o}_{n^{\ast}}}h(\mathbb{o}_{n^{\ast}})  \tilde{\mathbb{f}}^{G_{n^{\ast}}}(\mathbb{o}_{n^*}) = \mathbb{E}[ h(\mathbb{O}_{n^{\ast}}^{G_{n^{\ast}}+})], \\
        & \boldsymbol{\lim} \sum\limits_{o}h'(o)\tilde{f}_{i}^{G_{n^*}}(o) = \mathbb{E}[ h'(O_i^{G_{n^*}+})].
    \end{align}
\end{theorem}

\subsection{Remark on the strong positivity condition}

In Appendix \ref{appsec: estim}, we give a stronger condition $C2$ (Sub-cluster positivity) on the treatment assignment mechanisms $\mathbf{f}_{0,\mathbf{A}}$ in $\mathbb{S}_0$ for a given asymptotic law $\mathbb{P}_0$, which, in conjunction with weak positivity condition $C0$ implies the strong cluster positivity condition $C1$, and which may recognizably capture observed data generating mechanisms in many settings. Heuristically, this condition involves an asymptotic law wherein, as $n$ grows, clusters may be partitioned into an increasing number of weakly dependent sub-clusters of bounded size, each with their own finite resource supply, along with mild regularity conditions on the regimes implemented within each sub-cluster. As such, this condition bears structural resemblance to the partial interference assumptions \citepSM{sobel2006randomized} more frequently studied in the statistical causal inference literature. Under sub-cluster positivity ($C2$) we further have that these empirical estimators will be regular and asymptotically linear. Under correctly specified parametric models for $Q_Y$ and $\overline{q}_0$, parametric g-formula or IPW estimators can be used for their corresponding parameters and usual bootstrap procedures for iid data can be used to obtain (1-$\alpha$)\% confidence intervals. 

\subsection{Remark on inference for large-cluster estimands}

Appendix \ref{appsec: estim} gives extensive consideration to large-cluster estimands.  Therein, we provide a strong-positivity condition analogous to $C1$. However, this condition would not in general imply the existence of a consistent plug-in estimator of the large-cluster g-formula, even under the conditions of Theorem \ref{theorem: YbarIDlarge}. This follows because $\omega_0$ is not in general a pathwise-differentiable parameter.  Pathwise derivatives are important building blocks in the theory of regular estimation \citepMain{bickel1993efficient} and pathwise differentiability is a necessary condition for the existence of locally asymptotically unbiased or regular estimators \citepMain{hirano2012impossibility}.  We therefore provide minimal additional conditions to generalize the results of Theorem \ref{theorem: asymptcons} to large-cluster estimands. Importantly, we also leverage the results of Proposition \ref{lemma: optimallarge} and show that the procedures of \citetMain{luedtke2016optimal} will inherit analogous properties when applied to a single realization from $\mathbb{P}_n$: semiparametric efficiency and root-$n$ consistency when flexible data-adaptive algorithms are used to estimate nuisance parameters. We provide an extension of these results for the general asymptotic regime $\mathbf{G}_0$, which to our knowledge has not been considered, even heuristically, by the current literature. 

Nevertheless, the regularity of large-cluster estimands will in practice be violated whenever we consider CRs that are dynamic only with respect to a small number of discrete covariates. Such CRs will be of interest if lower dimensional coarsenings are more easily measured, or simpler rules are more easily implemented, for example in intensive care triage settings. In these cases we adapt the methods of \citetMain{luedtke2016statistical} to motivate inferential procedures, inspired by theory commonly used in online learning settings. We show that these procedures allow computation of root-$n$ consistent estimators and corresponding 95\% confidence interval.

\section{Application: Intensive Care Unit Admissions}
\label{sec: icu example}

Following \citetMain{harris2015delay}, we study the effect of prompt ICU admission on 90-day survival. We use resampled data from a cohort study of 13011 patients with deteriorating health referred for assessment for ICU admission in 48 UK National Health
Service (NHS) hospitals in 2010-2011. The data contain detailed covariate information including demographic variables (age, gender), patient clinical features, including the ICNARC physiology score, the NHS NEWS score,  SOFA score, septic diagnosis, and peri-arrest, as well as contextual factors including indicators for admission during a winter month, and admission out of hours (7 p.m.–7 a.m.). 

ICU bed availability is a limited resource; there was an available ICU bed for only ~27.3\% of the patients in the data upon their arrival to the hospital. Patients arriving concurrently in a given time window will thus necessarily compete for available beds. Specialist critical care outreach teams or single ICU physicians \citepMain{keele2020stronger} will decide on a given patient's prompt ICU admission on the basis of not only that patient's own clinical features, but of those for all such patients with indications for prompt admission. On the basis of these facts, an appropriate model for the data could be represented by the DAG in Figure \ref{fig: CFRLM} thus contradicting a conventional iid model, for example assumed by \citetMain{kennedy2019survivor} and \citetMain{mcclean2022nonparametric} who analyze these same data.

We consider the following regimes: 1) $\mathbf{G}_{0}^{\mathbf{opt}_1}$,  the optimal $V$-rank-preserving CR for a large cluster, where we take $V_i\equiv L_i$, the complete set of covariates measured for each patient in the observed data; 2) $\mathbf{G}_{0}^{\Lambda_1}$, a candidate $V$-rank-preserving CR for a large cluster, where we take $V_i$ equivalent to a patient's clinical risk prediction (i.e., a simple average of their ICNARC, NEWS, and SOFA scores), and a rank function $\Lambda_1$ that prioritizes patients with higher such scores; 3) $\mathbf{G}_{0}^{\mathbf{opt}_2}$,  the optimal regime in a large cluster, where we take $V_i$ equivalent to a 6-level categorization of a patient's clinical risk prediction; 4) $\mathbf{G}_{0}^{\Lambda_2}$, a candidate $V$-rank-preserving CR for a large cluster, inspired by a proposal to prioritize ventilator allocation during a hypothetical influenza pandemic \citepMain{nyc2015ventilator} using a rank function $\Lambda_2$ based on a patients trichotomized SOFA-score; and 5) ${G}_{n^*=20}^{\Lambda_2}$, that rule for a finite cluster with $n=20$ patients. Note that $\{\mathbf{G}_{0}^{\mathbf{opt}_1}, \mathbf{G}_{0}^{\Lambda_1}, \mathbf{G}_{0}^{\mathbf{opt}_2}, \mathbf{G}_{0}^{\Lambda_2}\}$ are asymptotic regimes implemented in large clusters, whereas ${G}_{n^*=20}^{\Lambda_2}$ is a regime implemented in a finite cluster;  $\{\mathbf{G}_{0}^{\mathbf{opt}_1}, \mathbf{G}_{0}^{\mathbf{opt}_2}\}$ are unknown optimal regimes, whereas $\{ \mathbf{G}_{0}^{\Lambda_1},  \mathbf{G}_{0}^{\Lambda_2}, {G}_{n^*=20}^{\Lambda_2}\}$ are known candidate regimes based on pre-specified rank functions, $\Lambda_1$ and $\Lambda_2$; and $\{\mathbf{G}_{0}^{\mathbf{opt}_1}, \mathbf{G}_{0}^{\Lambda_1}\}$ are $V$-rank-preserving regimes where $V$ contains at least one continuous covariate, whereas  $\{\mathbf{G}_{0}^{\mathbf{opt}_2}, \mathbf{G}_{0}^{\Lambda_2}, {G}_{n^*=20}^{\Lambda_2}\}$ are regimes where $V$ is a low dimensional covariate with few categories.

For each such regime, we target the expected cluster average outcome for the corresponding large- or finite-cluster under a range of resource limitations. We also target additional finite-cluster parameters under regime ${G}_{n^*}^{\Lambda_2}$, identified by the compositional $g$-formula \eqref{eq: compgform}, including $\mathbb{P}_{n^*}\Big(\mathbb{Y}_{n^*}^{G^{\Lambda_2}_{n^*}}(1)=x\Big)$, i.e. the probability that exactly $x$ of the 20 individuals in the cluster  are alive at 90-days, for each $x\in\{0,\dots,20\}$. 

For inference on large-cluster parameters under $\{\mathbf{G}_{0}^{\mathbf{opt}_1}, \mathbf{G}_{0}^{\Lambda_1}\}$, we implement a semiparametric efficient approach adapted from \citetMain{luedtke2016optimal}, and we compute asymptotically-valid 95\% confidence intervals based on the empirical variance of the estimated efficient influence function. Alternatively, when $V$ has few categories, $\{\mathbf{G}_{0}^{\mathbf{opt}_2}, \mathbf{G}_{0}^{\Lambda_2}\}$, we adopt the online-learning approach adapted from \citetMain{luedtke2016statistical}. For these procedures, asymptotic normality of the resulting estimators is motivated by a Martingale-based central limit theorem for triangular arrays \citepMain{ganssler1978central}. 
Finally, for inference on finite-cluster parameters under ${G}_{n^*}^{\Lambda_2}$, we compute plug-in estimators of the compositional g-formula, and an IPW estimator of the individual-level g-formula, where nuisance parameters are estimated using flexible parametric models with known parametric rates. For these parameters, 95\% confidence intervals are computed based on 500 nonparametric bootstrap samples. We provide additional details of analyses in Appendix \ref{appsec: misc}. Code is available at \href{https://github.com/AaronSarvet/OptimalLimited}{GitHub}.

\begin{figure}[ht]
    \centering
    \adjustbox{trim={.0 \width} {.01\height} {0\width} {.02\height},clip}%
{
    \begin{tikzpicture}[scale=1]
        \node[anchor=north, scale=0.5] at (0,0){%
            \pgfimage{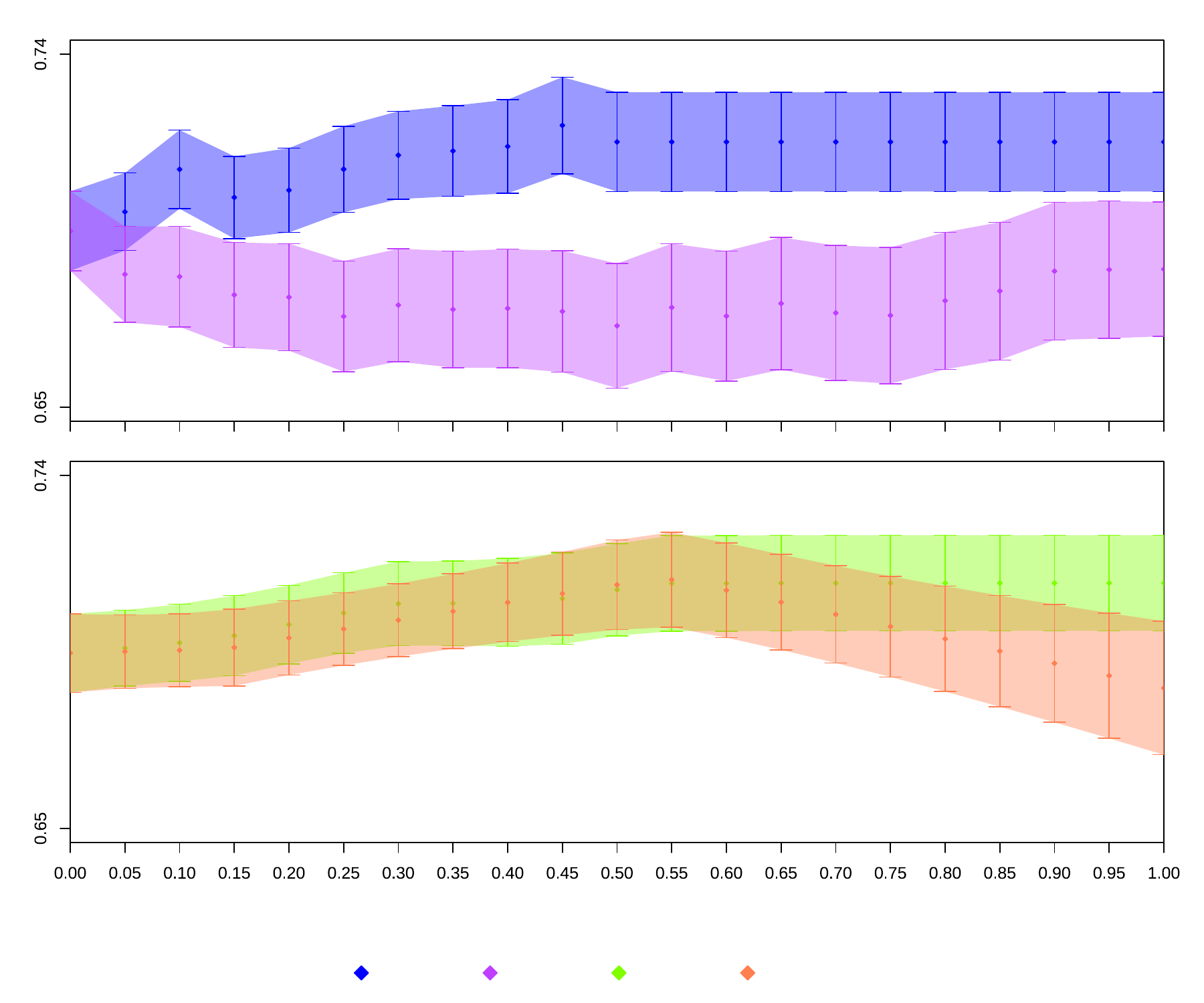}};
        \node[anchor=north] at (0.5,-11.3){$\kappa^*$};
        \node[anchor=north west] at (-3,-12.1){\tiny $\mathbf{G}_{0}^{\mathbf{opt}_1}$};
        \node[anchor=north west] at (-1.4,-12.1){\tiny $\mathbf{G}_{0}^{\Lambda_1}$};
        \node[anchor=north west] at (0.25,-12.1){\tiny $\mathbf{G}_{0}^{\mathbf{opt}_2}$};  
        \node[anchor=north west] at (1.85,-12.1){\tiny $\mathbf{G}_{0}^{\Lambda_2}$};

        \node[anchor=north west] at (-9,-5){\tiny $\widetilde{\mathbb{E}}\Big[\overline{Y}^{\mathbf{G}_0}_0\Big]$}; 
\end{tikzpicture}
}
\caption{Estimated values  (with 95\% CIs) of the expected proportion of individuals surviving to 90-days in a large cluster, under various prioritization regimes for prompt ICU admission.}
\label{fig: LargeComp}
    \end{figure}
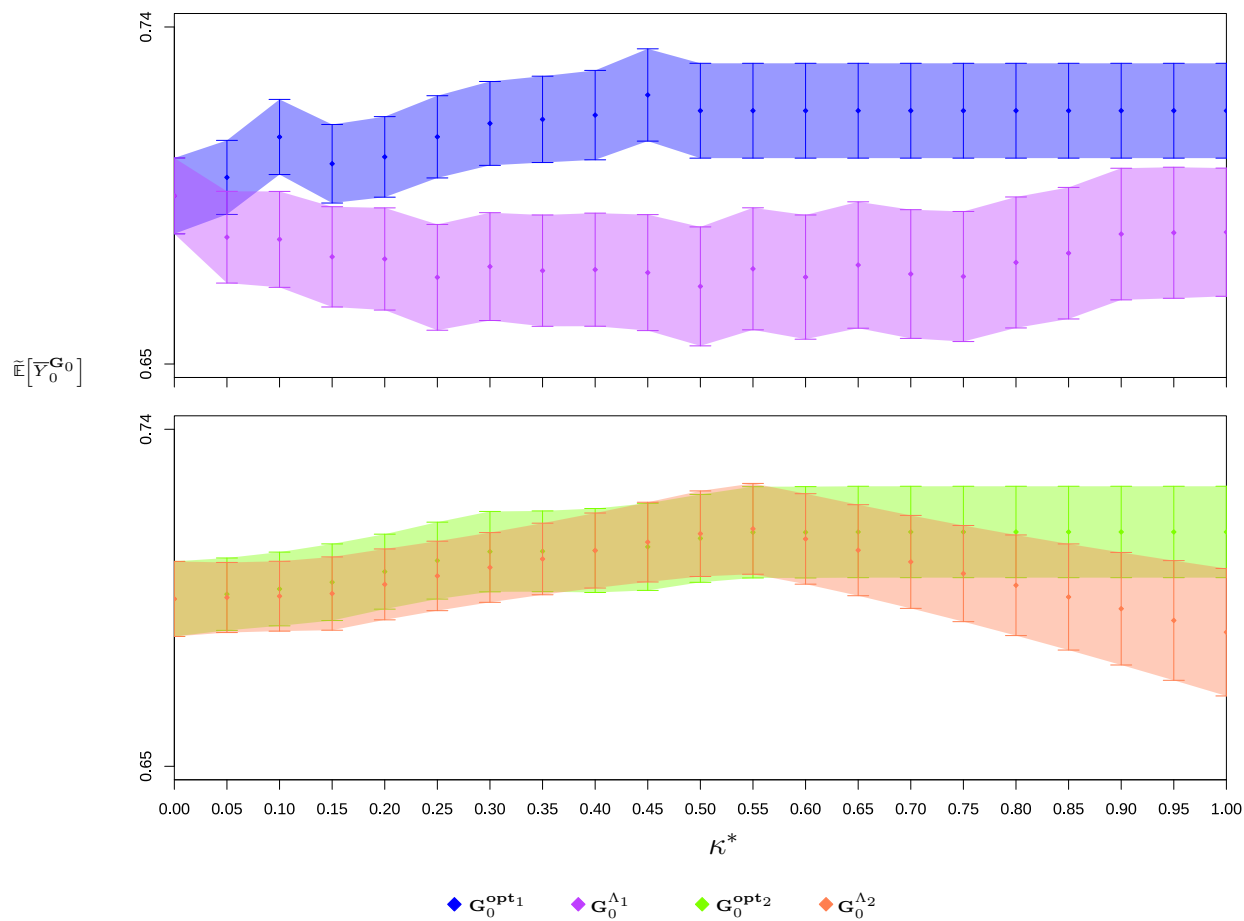

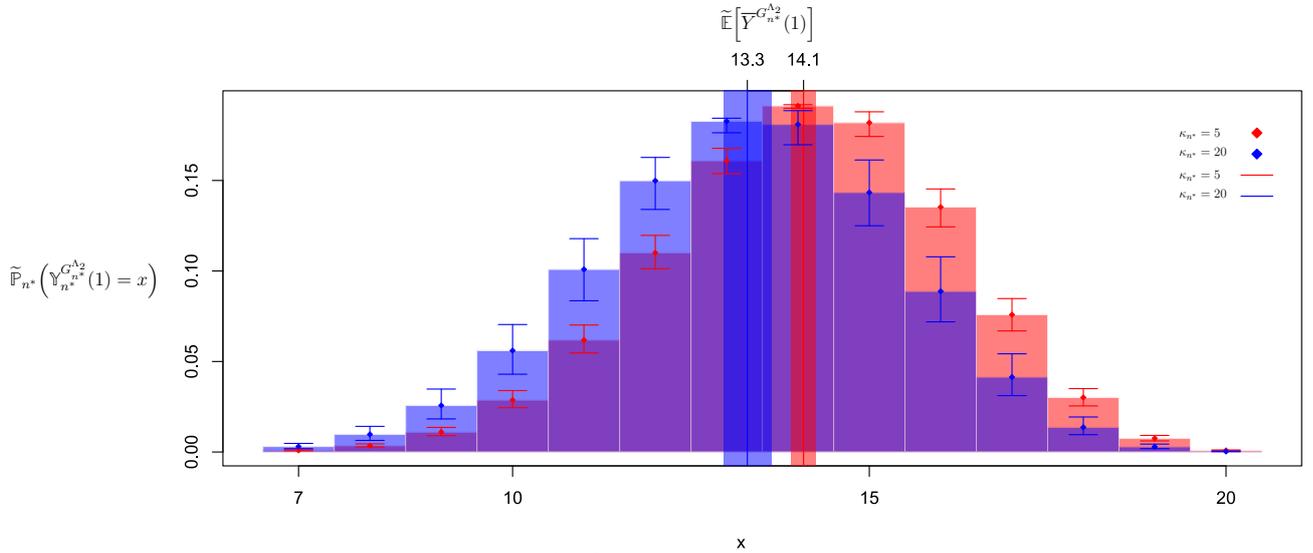
\begin{figure}[ht]
    \centering
    \adjustbox{trim={.0 \width} {.11\height} {0\width} {.0\height},clip}%
  {
    \begin{tikzpicture}[scale=1]
        \node[anchor=north, scale=.55] at (0,0){%
            \pgfimage{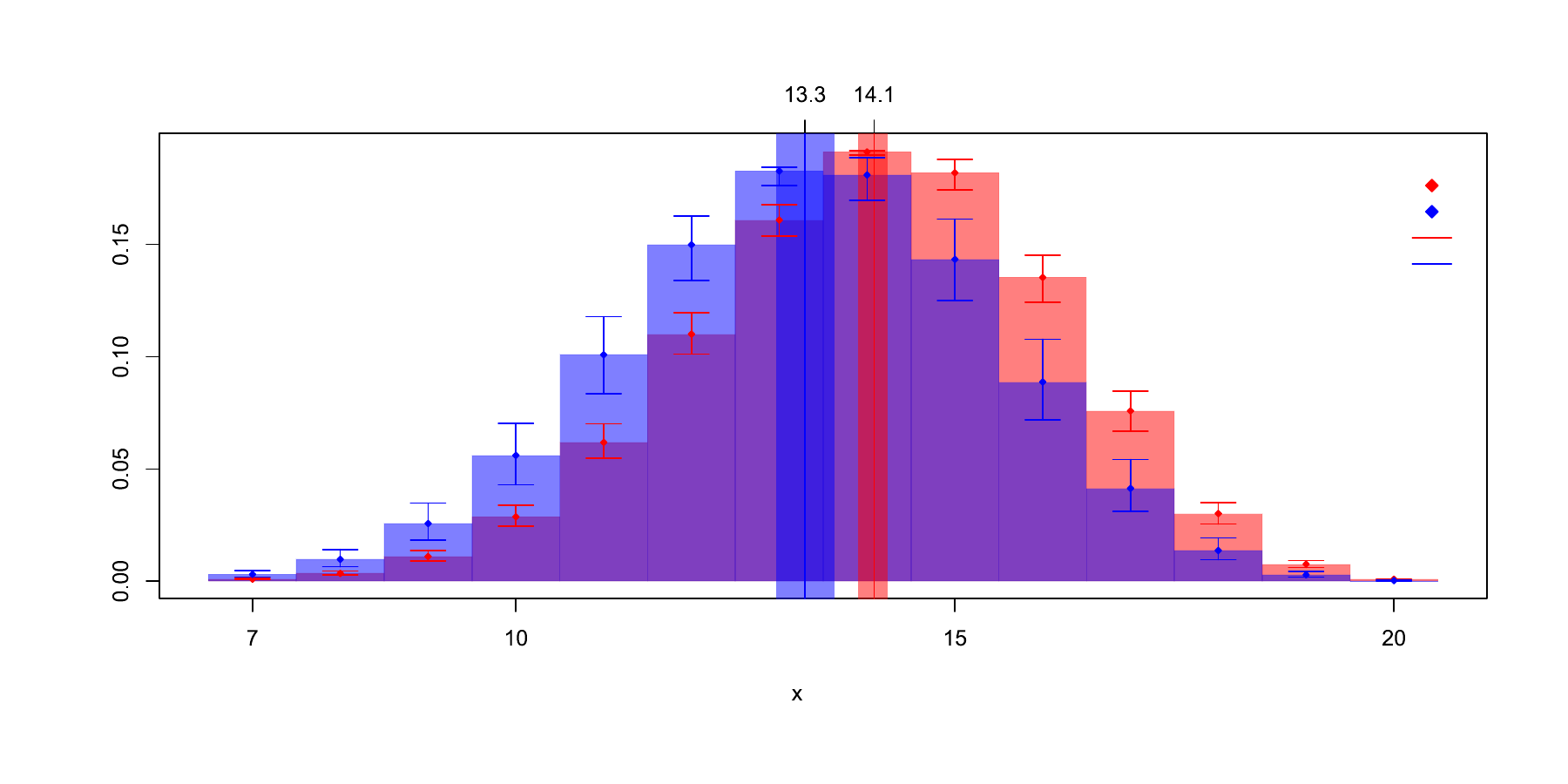}};
        \node[anchor=north west, scale=0.4] at (5.85,-1.95){$\kappa_{n^*}=5$};
        \node[anchor=north west, scale=0.4] at (5.85,-2.2){$\kappa_{n^*}=20$};
        
        \node[anchor=north west, scale=0.4] at (5.85,-2.5){$\kappa_{n^*}=5$};
        \node[anchor=north west, scale=0.4] at (5.85,-2.75){$\kappa_{n^*}=20$};

        \node [scale=0.6] at (0.5,-0.6){$\widetilde{\mathbb{E}}\Big[\overline{Y}^{G^{\Lambda_2}_{n^*}}(1)\Big]$};
        \node[scale=0.6] at (-8.5,-4){$\widetilde{\mathbb{P}}_{n^*}\Big(\mathbb{Y}_{n^*}^{G^{\Lambda_2}_{n^*}}(1)=x\Big)$};
\end{tikzpicture}}
\caption{Estimated probabilities (with 95\% CIs) that exactly $x$ individuals survive to 90-days
. Vertical lines provide the expected number of patients surviving to 90-days.}
\label{fig: finite}
    \end{figure} 

Figure \ref{fig: LargeComp} presents expected large-cluster-average 90-day survival under several large-cluster regimes. The top panel presents results using the semiparametric efficient estimators adapted from \citet{luedtke2016optimal}. Under $\mathbf{G}_{0}^{\mathbf{opt}_1}$, cluster survival is apparently maximized when at least approximately 50\% of patients can be treated, corresponding to the estimated proportion of patients with estimated positive CATEs (conditional on $V_i$) for prompt ICU admission. We illustrate the relationship between $\kappa^*$, CATEs, and the expected proportion of treated individuals in this large cluster under $\mathbf{G}_{0}^{\mathbf{opt}_1}$ in Figures \ref{fig: OptFullGridComp} and \ref{fig: OptTightGrid} in Appendix \ref{appsec: misc}. Under the candidate regime $\mathbf{G}_{0}^{\Lambda_1}$, which simply admits individuals to the ICU according to their predicted clinical risk scores, cluster survival is consistently lower under realistic resource limitations, suggesting that these risk scores may not optimally order individuals in terms of their expected benefit of prompt ICU admission, nor indicate when such admission may be harmful. The bottom panel presents results using the online-learning-based estimators adapted from \citet{luedtke2016statistical}. Results suggest similar survival between $\mathbf{G}_{0}^{\mathbf{opt}_2}$ and $\mathbf{G}_{0}^{\Lambda_2}$ up to a resource limitation of approximately 60\%. 

Figure \ref{fig: finite} presents $\mathbb{P}_{n^*}\Big(\mathbb{Y}_{n^*}^{G^{\Lambda_2}_{n^*}}(1)=x\Big)$ via histograms for settings with 5 ICU beds (red) vs. 20 ICU beds (blue), for a hypothetical cluster of size $20$. When only 5 ICU beds are available, 14.1 of the 20 patients (95\% CI: 13.9, 14.2) are expected to survive to 90 days, compared to 13.3 (13.0, 13.6) when all individuals are provided an ICU bed. Analogous results for all possible resource limitations,  $\kappa_n=0,\dots,20$, are presented in Table \ref{tab: CountTable} in Appendix \ref{appsec: misc}. Apparently worse survival among the all-treated group mirrors results for analogous parameters in a large cluster under the prioritization mechanism $\Lambda_2$, presented in the bottom panel of Figure \ref{fig: LargeComp}. Numerical results for Figure \ref{fig: LargeComp} are tabulated in Appendix \ref{appsec: misc}.

\section{Future directions}

The results in the main text and appendices are formulated generally. They serve as a foundation for future research on dynamic and optimal regimes in clustered data. The longitudinal setting will be an important extension, as many resource-limited decisions are necessarily formulated as time-varying dynamic regimes. For example, patients often have to wait for resources, like ventilators and organ transplants, to become available. Furthermore, in many healthcare settings with a fixed set of treatment units, rules for terminating the treatment (or discharge) are important. Resources are often limited in health-care crises, such as pandemics. For these settings, it will be important to clarify how to validly and efficiently fuse data from crisis and pre-crisis contexts, so as to rapidly and reliably provide statistical support for decisions during that crisis. Finally, the results should be extended for more complex interference and network structures.


\bibliographystyleMain{unsrtnat-init}
\bibliographyMain{Limited1.bib}

\clearpage

\setcounter{page}{1}
\begin{appendices}
\setcounter{table}{0}
\renewcommand{\thetable}{S\arabic{table}}
\setcounter{figure}{0}
\renewcommand{\thefigure}{S\arabic{figure}}
\setcounter{equation}{0}
\renewcommand{\theequation}{S\arabic{equation}}
\setcounter{definition}{0}
\renewcommand{\thedefinition}{S\arabic{definition}}
\setcounter{proposition}{0}
\renewcommand{\theproposition}{S\arabic{proposition}}
\setcounter{corollary}{0}
\renewcommand{\thecorollary}{S\arabic{corollary}}
\setcounter{lemma}{0}
\renewcommand{\thelemma}{S\arabic{lemma}}
\setcounter{theorem}{0}
\renewcommand{\thetheorem}{S\arabic{theorem}}
\appendixpage

\DoToC
\clearpage

\addtocontents{toc}{\protect\setcounter{tocdepth}{1}}
\section{Mapping the space of cluster dynamic regimes}
\label{appsec: gencdtrs}

The space of individualized dynamic regimes (IR) has expanded the last decades. First, average treatment effects (ATE) that were once restricted to contrasts of static interventions, for example, $\mathbb{E}_{P^F}[Y^a - Y^{a'}]$, were generalized to regimes that depended on covariates alone, $\mathbb{E}_{P^F}[Y^{g(L)} - Y^{g'(L)}]$, randomizing terms, $\mathbb{E}_{P^F}[Y^{g(L, \delta)} - Y^{g'(L,\delta)}]$, and even the natural value of treatment $\mathbb{E}_{P^F}[Y^{g(L, \delta, A)} - Y^{g'(L,\delta, A)}]$. These classification schemes were instrumental in developing a general theory of identification of ATEs under IRs, and clarified when conditions could be relaxed for sub-classes of regimes. For example, in time-varying settings sufficient and necessary exchangeability conditions for g-formula identification of expected potential outcomes under regimes $g(\emptyset)$, $g^*(\overline{L}_k^{g^*})$ and $g^{\dagger}(\overline{L}_k^{g^{\dagger}}, \overline{A}_k^{g^{\dagger}})$ are progressively stronger. See \citetSM{richardson2013single} for a review.  Elaborating the space of these regimes also equips investigators with an expanded set of tools to interrogate their substantive questions of interest. 

In contrast, very little work has been done to map the space of cluster dynamic regimes (CRs). A small body of work, from the perspective of design-based inference, elaborates on a space of stochastic regimes $G_n(\delta)$, subsuming concepts such as static cluster regimes, Bernoulli designs, complete randomization designs, paired-randomization designs, and arbitrary randomization designs, see \citetSM{savje2021average} for a review. \citetSM{van2014causal} and \citetSM{ogburn2020causal} consider dynamic treatment regimes in a network settings. However, their regimes are individualized in the sense that each individual's treatment is restricted to be a function of at most a randomizer, their own covariates, and the covariates of other individuals for whom they have ties, as encoded by a fixed adjacency matrix. In this sense, \citetSM{van2014causal} and \citetSM{ogburn2020causal} only consider regime $G_n$ for which each individual's treatment is restricted to be independent of any other individual's treatment conditional on their own covariates and those of their connections, and furthermore these individualized regimes are restricted to be the same function across all individuals in the cluster (``network'').

In this appendix we begin a general mapping of the space of CRs, with particular focus on constrained CRs, as these are the topic of this paper. 
In Section \ref{appsubsec: gencdtrs} we elaborate on the generalized class of constrained CRs in which at most (but possibly less than) $\kappa_n$ individuals can be assigned treatment. We provide a rank-and-treat formulation of this class of regimes and extend results relating stochastic and deterministic IRs to this class of CRs.
In Section \ref{appsubsec: lequit}, we introduce a restricted class of regimes defined by the property of $L$-equity. This property will be of interest to policy makers whenever non-discrimination on the basis of unmeasured features (for example, individual index $i$) is a pre-requisite for an equitable policy. We locate $L$-rank-preserving CRs (Definition \ref{def: LRPregimes}) as a subclass of these regimes, and importantly one in which the optimal $L$-equitable regime is contained for laws $\mathbb{P}_n^{F}$ under $\mathcal{M}^{AB}_n$. We also introduce subclasses of CRs intermediate between $L$-equitable regimes and $L$-rank-preserving regimes, which are useful for didactic purposes, but may be of substantive interest in their own right.
In Section \ref{appsubsec: intdens} we provide a general result for identifying the intervention densities $q^*_i$ and $\mathbb{q}^*_n$ comprising the individual-level and compositional g-formulae of expressions \eqref{eq: patlevgform} and \eqref{eq: compgform}, and illustrate simplified identification results for CRs in restricted classes.
In Section \ref{appsubsec: computation} we provide general identification results for compositional and individual level parameters defined by CRs that are a function of a subset of covariates. These results can be used to motivate estimators with significant computational advantages over those suggested by the functionals presented in the main text. 
In Section \ref{apsubsec: subclus} we posit the existence of an unmeasured variable indicating sub-cluster membership, and develop a class of sub-cluster regimes. In particular, we do so in service of articulating a class of asymptotic regimes that imply conditions sufficient for the asymptotic consistency result of Theorem \ref{theorem: asymptcons}, which we develop in Appendix \ref{appsec: estim}. Figure \ref{fig: classes} illustrates relations between the classes of CRs defined in this manuscript.

\subsection{Generalized constrained CRs} \label{appsubsec: gencdtrs}

In the main text, we focus results on regimes in the policy class $\Pi_n(\kappa_n)$ under which exactly $\kappa_n$ units are assigned treatment. In Section \ref{sec: elab} we illustrate alternative rank-and-treat formulations for these regimes, where a total rank ordering is assigned according to $G^*_{n, \mathbf{R}}$ and then treatment for a individual $i$ is assigned if their rank ordering is less than or equal to the number of treatment units $\kappa_n$. These regimes may be of interest particularly in settings where $\Delta_{P^F}(l)>0$ for all $l$, which will often be the case for target populations where treatment resources are scarce, for example a waiting list for transplants, and an ICU subpopulation with clinical indications for mechanical ventilation. However, the unknown treatment regime $\mathbf{f}_{n, \mathbf{A}}$ implemented in the observed data is unlikely to be equivalent to a regime $G_n \in \Pi_n(\kappa_n)$, because in many observed data settings, individuals or clinicians will possibly reject a treatment, even when available, and the number of treatment units actually used by the cluster may be less than the number of treatment units available. As examples: a triage system may stock-pile a random proportion of their treatment units; an individual clinician may subvert an official triaging policy in anticipation of other individuals that they would perceive to be in greater need; and individuals may reject an available treatment unit due to beliefs about inferiority to momentarily unavailable treatment alternatives, religious beliefs, or other concerns. 

The space of regimes $\Pi^F_n(\kappa_n)$ is especially important in developing realistic conditions on the components $\mathbf{f}_{0, \mathbf{A}}$ describing the observed treatment assignment mechanisms in an asymptotic law $\mathbb{P}_0$. We show in Appendix \ref{appsec: estim} that such conditions are generally necessary for establishing basic asymptotic properties of estimators (for example, asymptotic consistency), and we leverage conditions on such regimes to establish stronger asymptotic estimation results. 

\subsubsection{Rank-based formulations}\label{appsubsubsec: rankgen}

In contrast to regimes $G_n\in \Pi_n(\kappa_n)$, an arbitrary limited resource regime $G_n^*\in \Pi^F_n(\kappa_n)$ would have the following rank-and-treat formulation, where $\{\delta_{\mathbf{R}}, \delta_{A_1}, \dots, \delta_{A_n}\}$ are randomizer terms with some investigator-specified joint distribution. Let $\mathfrak{A}_i$ be the random variable denoting the subset of indices $\{1,\dots,n\}$ of those individuals who are offered treatment before individual $i$. 

\begin{align}
    R_i^{G_n^*+} \coloneqq & G^*_{n, R_i}(\mathbf{L}_n, \delta_{\mathbf{R}}) \label{eq: RplusF}\\
    A_{i}^{G^*_{n}+} \coloneqq & G^*_{n, A_i}(\mathbf{R}_n^{G_n^*+}, \mathbf{L}_n, \{A_{j}^{G^*_{n}+}: j\in\mathfrak{A}_i^{G^*_n+}\}, \delta_{A_i}) \nonumber \\ 
    = &  g_i\Big(\mathbf{L}_n,  \{A_{j}^{G^*_{n}+}: j\in\mathfrak{A}_i^{G^*_n+}\}, \delta_{\mathbf{R}}, \delta_{A_i}\Big) \times I\Big(\kappa_{n}>\sum\limits_{\mathfrak{A}_i^{G^*_n+}}A_{j}^{G^*_{n}+}\Big)
    .  \label{eq: AplusF}
\end{align}

A rank ordering, in which available treatments are offered, is determined by $G^*_{n, R_i}$, and subsequently in this order, each individual $i$ is assigned treatment according to the individual rule $g_i$ that may vary across individuals, until the resource supply is exhausted. $g_i$ represents each individual's \textit{individualized} dynamic regime, followed in the observed data, that determines their treatment \textit{if a unit is available} (i.e., $I\Big(\kappa_{n}>\sum\limits_{\mathfrak{A}_i^{G^*_n+}}A_{j}^{G^*_{n}+}\Big)=1$).

Notably, a generalized resource-limited regime admits the possibility of an assignment of strictly less than $\kappa_n$ treatment units, or even no treatment units at all, within a given realization of a cluster. Any IR $g$ in iid settings would have an analogue in this class, whereby $\kappa_n\geq n$, $g_i=g : (l, \delta_{A_i}) \mapsto g(l, \delta_{A_i})$ for all individuals, and $\{\delta_{A_1}, \dots, \delta_{A_n}\}$ are iid. In this sense, such regimes are \textit{rank-static}, but \textit{individually-dynamic}: a single set of individual ranks is assigned regardless of cluster configuration $\mathbf{L}_n$, but individuals make decisions when treatments are available according to a non-trivial function $g$. Alternatively, the general policy class $\Pi_n(\kappa_n)$ is the strict subset of $\Pi^F_n(\kappa_n)$ in which $g_i$ is the trivial function that maps to $1$ for all $i$. In this sense, regimes in $\Pi_n(\kappa_n)$ are \textit{rank-dynamic}, but \textit{individually-static}. 

An arbitrary regime in $\Pi^F_n(\kappa_n)$ is both \textit{rank-dynamic} and \textit{individually-dynamic}. Thus, we can understand this space as a product of the space of rank regimes $G^*_{n, \mathbf{R}}$ and the $n$-dimensional space of individualized treatment regimes $g_i$. More formally, let $\Pi$ represent the set of possible functions $g_i$, and $\Pi^n \equiv \{\Pi\}^n$. Then we see that there will exist a bijective mapping between $\Pi_n(\kappa_n)\times\Pi^n$ and $\Pi^F_n(\kappa_n)$ because elements $G_n\in\Pi_n(\kappa_n)$ are uniquely characterized by their rank function $G^*_{n, \mathbf{R}}$. 

These rank-based formulations of generalized constrained CRs in this manuscript will be important in future work for investigators interested in regimes that intervene on rank alone via  $G^*_{n, \mathbf{R}}$, or on two-stage regimes via $G^*_{n, \mathbf{R}}$ and then subsequently on individual's treatment acceptance (conditional on treatment availability) via $g_i$. 

\subsubsection{Stochastic and fully deterministic CRs}\label{appsubsubsec: stochanddet}

In this section we draw connections between generalized CRs in $\Pi^F_n(\kappa_n)$ and general stochastic IRs $g: (l, \delta) \mapsto g(l, \delta)$. General \textit{stochastic} IRs despite being deterministic functions of $L_i$ and $\delta$, are so-called because of their non-trivial dependence on the exogenous randomizer term $\delta$. Heuristically, they are often understood as two-stage regimes: first an individual will realize a value of $\delta\sim \text{Unif}(0,1)$, which determines a \textit{deterministic} IR $g^{D}: l \mapsto g^D(l)$ used to assign their treatment. This interpretation is apparent in the following formulation, where we let $\Pi^D$ be the set of all deterministic IRs $g^D$:

$$g(l, \delta) = \sum\limits_{\Pi^D}g^D(l)I\Big(\delta\in q^{\dagger}_{g}(g^D)\Big),$$ where $q^{\dagger}_{g}$ is a bijective function characteristic of $g$ that maps elements of $\Pi^D$ onto a partition of the unit interval. As such, we have the following simple identity for the intervention density given $l$:
$$P(A^{g+}=a \mid L=l) = \sum\limits_{\Pi^D}I\Big(a=g^D(l)\Big)q^*_g(g^D),$$
where $q^*_g$ is the function corresponding to $q^{\dagger}_{g}$ such that $q^*_g(g^D)$ evaluates to the width of the interval $q^{\dagger}_{g}(g^D)$. In the following, we extend this heuristic, and corresponding formulation, to the class of CRs $\Pi^F_n(\kappa_n)$. 

\begin{definition}[Fully-deterministic CRs]
The class of fully-deterministic CRs $\Pi^{D}_n(\kappa_n)$ is the subset of regimes $G_n$ in  $\Pi^F_n(\kappa_n)$ such that, for each $\mathbf{l}_n$,  $\mathbb{E}_{\delta}[I(G_n(\mathbf{l}_n, \delta) = \mathbf{a}_{\mathbf{l}_n})] = 1$  for some $\mathbf{a}_{\mathbf{l}_n}\in\{\mathbf{a}_n \mid \lVert \mathbf{a}_n \rVert \leq \kappa_n\}$.
\end{definition}

Fully-deterministic regimes thus are characterized by an assignment mechanism that maps a covariate vector $\mathbf{l}_n$ to a single treatment vector $\mathbf{a}_n$ and thus we omit the randomizer term $\delta$ as an argument.  

Analogous to IRs $g$, a given non-deterministic CR $G_n$, that is, any $G_n$ that depends non-trivially on $\delta$, can be expressed as a distribution over the set of deterministic CRs, that is 

$$G_n(\mathbf{l}_n, \delta) = \sum\limits_{\Pi^{D}_n(\kappa_n)}G_n^D(\mathbf{l}_n)I\big(\delta \in q^{\dagger}_{G_n}(G^D_n)\big),$$ where  $q^{\dagger}_{G_n}$ is a bijective function characterizing $G_n$ that maps elements of $\Pi^{D}_n(\kappa_n)$ onto increments forming a partition of the unit interval. As such, it follows that 

\begin{align*}
    \mathbb{P}^F_n(\mathbf{A}^{G_n+}_n=\mathbf{a}_n \mid \mathbf{L}_n=\mathbf{l}_n) 
    & = \sum\limits_{\Pi^{D}_n(\kappa_n)} I\big(\mathbf{a}_n = G^D_n(\mathbf{l}_n)\big) q^{*}_{G_n}(G^D_n),
\end{align*}
where $q^{*}_{G_n}$ is defined relative to $q^{\dagger}_{G_n}$ analogously to $q^{*}_{g}$.

\subsection{$L$-equitable CRs} \label{appsubsec: lequit}

In this section, we introduce a class of regimes under which a form of conditional stochastic exchangeability (in the de Finetti sense) is induced for individual treatment variables. We refer to such regimes as stochastic $L$-equitable CRs. We call such regimes $L$-equitable because of the impartiality with which they regard features of a cluster that depend on individual indices $i$, with respect to covariates $L_i$. Articulating such a class of regimes is important because it is defined by a feature that is arguably necessary for any regime understood as ``equitable'' or ``fair'' in the plain-language sense. Formally, such regimes are defined as follows:

\begin{definition}[Stochastic $L$-equitable CRs]

The class of stochastic $L$-equitable CRs $\Pi^*_n(\kappa_n, L_i)$ is the subset of regimes $G_n\in\Pi_n(\kappa_n)$ such that for each individual $i$, and for all laws $\mathbb{P}_n^F$, the following equality holds for each covariate composition $\mathbb{l}_n$ and covariate value $l$: $$\mathbb{P}_n^F(A^{G_n+}_i=1 \mid L_i=l, \mathbb{L}_n=\mathbb{l}_n) = p^{G_n}(l, \mathbb{l}_n)\in[0,1].$$ 
\end{definition}

In words, $\Pi^*_n(\kappa_n, L_i)$ is the class of CRs for which all individuals with the same covariate value $L_i=l$ in a cluster with composition $\mathbb{l}_n$ have the same probability of treatment assignment $p^{G_n}(l, \mathbb{l}_n)$, regardless of the specific instantiation of $\mathbf{L}_n$ or the law of the full data, $\mathbb{P}_n^F$. Under a stochastic $L$-equitable regime, the $n$ individuals' values of $L_i$ could be arbitrarily permuted without perturbing the probability that an arbitrary individual with covariate value $l$ would be assigned treatment. 

Whenever we consider an $L$-equitable regime, $q^*_{i, \mathbb{l}_n}$ will be fixed for all $i$. By consequence we have the following result generalizing Proposition \ref{theorem: YbarID} of Section \ref{sec: ID}.

 \begin{proposition} \label{theorem: YbarIDrelax}
    Consider a law $\mathbb{P}^F_n \in\mathcal{M}^{AB}_n$ and $G_n \in \Pi^*_n(\kappa_n, L_i)$. Then $\mathbb{E}_{\mathbb{P}^F_n}[ \overline{Y}^{{G}_n}_n] =  \sum\limits_{o\in\mathcal{O}_i}yf_{i}^{G_n}(o).$
\end{proposition}

\begin{proof}
\begin{align*}
    \mathbb{E}_{\mathbb{P}^F_n}[ \overline{Y}^{{G}_n}_n] = \frac{1}{n}\sum\limits_{j=1}^n\mathbb{E}_{\mathbb{P}^F_n}[ {Y}^{{G}_n}_j] = \frac{1}{n}\sum\limits_{j=1}^n\sum\limits_oyf_{j}^{G_n}(o) =\frac{1}{n}\sum\limits_{j=1}^n\sum\limits_oyf_{i}^{G_n}(o) = \sum\limits_oyf_{i}^{G_n}(o).
\end{align*}
The first equality follows by linearity of expectation function. The second follows from Theorem \ref{theorem: CDTRID}. The third follows because $q_j^*=q^*_i$ for all $j$ under regimes $G_n \in \Pi^*_n(\kappa_n, L_i)$. 
\end{proof}

Note that a complete randomization design is a special case of a stochastic $L$-equitable CRs in which $q^*_{i, \mathbb{l}_n}(1\mid l)=\frac{\kappa_n}{n}$ for all $l$, $\mathbb{l}_n$. In contrast, a Bernoulli design with probability  $q^*_{i, \mathbb{l}_n}(1 \mid l)=\frac{\kappa_n}{n}$ is not an $L$-equitable CR, because under such a design there would be a strictly positive probability that all individuals in the cluster receive treatment, $q^*_{i, \mathbb{l}_n}(1\mid l)$. This illustrates that the elements in $\mathbf{A}^{G_n+}_n$ are not in general independent under $L$-equitable CRs, as observed by \citetSM{savje2021average} in regards to complete randomization designs.

In the following subsection we locate the $L$-rank-preserving class of regimes $\Pi^d_n(\kappa_n, L_i, \Lambda)$ considered in the main text as a subclass of $\Pi^*_n(\kappa_n, L_i)$.  

\subsubsection{Deterministic, mixed-deterministic and $L$-rank-preserving CRs} \label{appsubsubsec: determ}

We first define a subclass of stochastic $L$-equitable regimes $\Pi^*_n(\kappa_n, L_i)$ closely-related to the  $L$-rank-preserving class of regimes $\Pi^d_n(\kappa_n, L_i, \Lambda)$ introduced in Section \ref{sec: elab}, which we refer to as deterministic CRs. 

\begin{definition}[Deterministic CRs]
The class of deterministic CRs $\Pi^{d*}_n(\kappa_n, L_i)$ is the subset of regimes $G_n$ in  $\Pi^{*}_n(\kappa_n, L_i)$ such that for each $\mathbb{l}_n$, there exists at most one $l'\in \mathcal{L}$ such that $p^{G_n}(l',\mathbb{l}_n)\not\in \{0,1\}$.
\end{definition}

Under such a regime, there will exist a set of functions $\Lambda_{\mathbb{l}_n}:  \mathcal{L}_{i} \rightarrow \mathbb{R}$ such that for each $\mathbb{l}_n$ $\Lambda_{\mathbb{l}_n}(L_{i}) > \Lambda_{\mathbb{l}_n}(L_{j}) \implies G^*_{n, R_i}(\mathbf{L}_n, \delta) > G^*_{n, R_j}(\mathbf{L}_n, \delta)$ for all  $i, j$, almost surely. Note these regimes  $\Pi^{d*}_n(\kappa_n, L_i)$ differ from those in $\Pi^{d}_n(\kappa_n, L_i, \Lambda)$ defined in Section \ref{sec: elab} in that coarsened rank functions $\Lambda_{\mathbb{l}_n}$ may vary with a cluster's composition $\mathbb{l}_n$, and are indexed as such. Thus $\Pi^{d}_n(\kappa_n, L_i, \Lambda) \subset \Pi^{d*}_n(\kappa_n, L_i)$. 

Conditional on $\mathbb{l}_n$, regimes in $\Pi^{d*}_n(\kappa_n, L_i)$ are static in the sense that the function $\Lambda_{\mathbb{l}_n}$ is fixed. We generalize $\Pi^{d*}_n(\kappa_n, L_i)$ to a class of regimes that is stochastic with respect to $\Lambda_{\mathbb{l}_n}$. Heuristically, conditional on $\mathbb{l}_n$, these regimes would be implemented by first selecting a coarsened rank function $\Lambda_{\mathbb{l}_n}$ using a pre-specified distribution over a suitable set of such functions, and then implementing the corresponding regime in $\Pi^{d*}_n(\kappa_n, L_i)$. First, we define some additional notation. Let $\mathcal{L}_{\Lambda}$ be the set of all coarsened rank functions, grouped by rank orderings of some partitioning of $\mathcal{L}$. For example, if $\mathcal{L} \equiv \{1,2,3\}$, $\mathcal{L}_{\Lambda}$ will have groups of coarsened rank functions defined by the following orderings:

\begin{align*}
    \begin{Bmatrix}
        (1,2,3),  & (2,1,3), & (3,2,1),  \\
        (1,3,2), & (2,3, 1), &(3,1,2), \\
        (\{1,2\}, 3), & (\{1,3\}, 2), & (\{2,3\}, 1), \\
        (3, \{1,2\}), & (2, \{1,3\}), & (1, \{2,3\}), \\
       &  (\{1, 2,3\}) &
    \end{Bmatrix}.
\end{align*}

The first six rank orderings correspond to rank orderings over the elements in $\mathcal{L}$. The second six rank orderings correspond to rank orderings over all possible coarsenings of $\mathcal{L}$ with only 2 elements. The final rank ordering is the trivial one in which all individuals are assigned the same rank group. Note that $\mathcal{L}_{\Lambda}$ will have a finite number of elements whenever $\mathcal{L}$ has a finite number of elements. 

Then, let $q^*_{\Lambda, \mathbb{l}_n}$ denote some distribution over the coarsened rank functions in $\mathcal{L}_{\Lambda}$, and define the following policy class:

\begin{definition}[Mixed-deterministic CRs]

The class of mixed-deterministic CRs $\Pi^{d}_n(\kappa_n, L_i, q^*_{\Lambda, \mathbb{l_n}})$ is the subset of regimes $G_n$ in $ \Pi_n(\kappa_n)$ for which $\delta_{\mathbf{R}} \equiv \{\delta_{\Lambda}, \delta'_{\mathbf{R}}\}$, and the following properties hold for some set of distributions $q^*_{\Lambda, \mathbb{l}_n}$:

\begin{itemize}
    \item[1.] For each $\mathbb{l}_n$, $G_n$ is in part characterized by $q^*_{\Lambda, \mathbb{l_n}}$ that puts mass on each $G^*_n\in\Pi^{d}_n(\kappa_n, L_i, \Lambda)$ proportional to $q^*_{\Lambda, \mathbb{l_n}}(\Lambda)$, as determined by $\delta_{\Lambda}$, and where the implemented regime $G^{D}_n\in\Pi^D_n(\kappa_n)$ is determined by $\delta'_{\mathbf{R}}$, and 
    \item[2. ] $\delta_{\mathbf{R}}'$ is independent of $\delta_{\Lambda}$.
\end{itemize}

\end{definition}

We thus have that  $\Pi^{d}_n(\kappa_n, L_i, \Lambda)$ is the subset of $\Pi^{d}_n(\kappa_n, L_i, q^*_{\Lambda, \mathbb{l_n}})$ in which $q^*_{\Lambda, \mathbb{l_n}}$ is degenerate for a single $\Lambda_{\mathbb{l}_n}\equiv\Lambda$, for all $\mathbb{l}_n$. In Appendix \ref{appsec: estim}, we consider a regime in the subclass of  $\Pi^{d}_n(\kappa_n, L_i, q^*_{\Lambda, \mathbb{l_n}})$ in which $q^*_{\Lambda, \mathbb{l_n}}=q^*_{\Lambda}$ for all $\mathbb{l}_n$, i.e., the distribution over coarsened rank functions is fixed for all compositions.

\subsubsection{$\epsilon-L$-rank-preserving CRs}\label{appsubsubsec: epsilonrankpres}

Results for large-cluster estimands in Theorem \ref{theorem: YbarIDlarge} are provided under a condition $E$ whereby the regimes in $\mathbf{G}_0$ are restricted to a class of $L$-rank-preserving regimes $\Pi_n^{d}(\kappa_n, L_i, \Lambda)$. Consideration of $L$-rank-preserving regimes simplifies our consideration of asymptotic regimes $\mathbf{G}_0$ to those characterized by a fixed coarsened rank function $\Lambda$. However, an investigator interested in a large-cluster parameter under a deterministic CR in $\Pi_n^{d*}(\kappa_n, L_i)$ may be concerned by the seeming lack of generality in the results of Theorem \ref{theorem: YbarIDlarge}. To highlight the practical weakness of restricting results to $L$-rank-preserving regimes, note that we may define a less-restricted policy class as follows:

\begin{definition}[$\epsilon-L$-rank-preserving deterministic CRs]
The class of $\epsilon-L$-rank-preserving CRs $\Pi^{d\epsilon}_n(\kappa_n, \Lambda, Q_L)$ is the subset of regimes $G_n$ in $ \Pi^{d*}_n(\kappa_n, L_i)$ for which their exists a rank-based formulation where, for each $\mathbb{l}_n$, the coarsened rank functions are invariant, $\Lambda_{\mathbb{l}_n} = \Lambda$ whenever $\sum\limits_{l\in\mathcal{L}}\Big |\frac{\mathbb{l}_n(l)}{n} - Q_L(l) \Big |  < \epsilon$.
\end{definition}

These policy classes strictly increase in size with decreasing $\epsilon$ and approach $\Pi_n^{d*}(\kappa_n, L_i)$ as $\epsilon$ approaches 0. Note that in large clusters under an asymptotic law $\mathbb{P}_0$ following condition $D$, we have that $$\lim\limits_{n\to \infty} \mathbb{P}_n\Bigg( \Big| \frac{\mathbb{L}_{n}(l)}{n} - Q_{L}(l) \Big| > \epsilon' \Bigg)=0,$$ which holds for a positive $\epsilon'$ arbitrarily close to 0. Thus, so long as we take $\epsilon$ to be greater than  $\epsilon'$ by an arbitrarily small amount, the probability of observing an $\mathbb{l}_{n}$ for which $\Lambda_{\mathbb{l}_{n}} \neq \Lambda$ will approach 0 with increasing $n$ for an arbitrarily small $\epsilon$. In other words, we may consider policy class $\Pi^{d\epsilon}_n(\kappa_n, \Lambda)$ that is only an arbitrarily smaller subset of $\Pi^{d*}_n(\kappa_n, L_i)$ in which the coarsened rank functions are locally invariant within an arbitrarily small neighborhood around $\mathbb{l}_{n}$.

\subsubsection{Practical considerations for implementing $L$-equitable regimes}

Many investigators will not in general specify interest in a specific stochastic regime. Alternatively, as has been done for stochastic IRs in the literature \citepSM{robins2004effects, munoz2012population, young2014identification, kennedy2019nonparametric, sarvet2023longitudinal}, stochastic $L$-equitable CRs may be more easily specified in terms of $q^*_{i,\mathbb{l}_{n}}$ directly via $p^{G_n}$, but we must carefully verify that such a $p^{G_n}$ is compatible with the resource limitation $\kappa_n$; it must be the case that $\sum\limits_{l\in\mathcal{L}}\mathbb{l}_n(l)p^{G_n}(l, \mathbb{l}_n) = \kappa_n$ for all $\mathbb{l}_n$. We show that for any such $p^{G_n}$, their will exist a regime $G_n$ with intervention density $q^*_{i,\mathbb{l}_{n}}(1\mid l) = p^{G_n}(l, \mathbb{l}_n)$. In the following proposition, we let $(\mathbb{b}_n)_{\mathbb{l}_n, \kappa_n}$ denote the set of compositions $\mathbb{b}_n$ compatible with $\mathbb{l}_n$ and $\kappa_n$, i.e., those for which $\sum\limits_{l}\mathbb{b}_n(1, l)=\kappa_n$ and  $\sum\limits_{a}\mathbb{b}_n(a, l)=\mathbb{l}_n$ for all $l$. We also let $(\mathbf{l}_n)_{\mathbb{b}_n}$ denote the set of configurations $\mathbf{l}_n$ compatible with $\mathbb{b}_n$, i.e. those for which $\sum\limits_{i=1}^nI\Big(\{l_i\in \mathbf{l}_n\} =l\Big)=\sum\limits_{a}\mathbb{b}_n(a, l)$ for all $l$, and we let $(G_n^D)_{\mathbb{b}_n}$ denote the set of fully deterministic regimes compatible with $\mathbb{b}_n$, i.e. those for which $\sum\limits_{i=1}^nI\Big(\{a_i\in G^D_n(\mathbf{l}_n)\}=a\Big)I\Big(\{l_i\in \mathbf{l}_n\}=l\Big) = \mathbb{b}_n(a, l)$ for all $a, l$, for each $\mathbf{l}_n\in(\mathbf{l}_n)_{\mathbb{b}_n}$.

\begin{proposition}
    
\label{lemma: Lequitprobs}
    Given a cluster with composition $\mathbb{L}_n=\mathbb{l}_n$, consider an investigator-specified function $p^{G_n}$ such that $\sum\limits_{l\in\mathcal{L}}\mathbb{l}_n(l)p^{G_n}(l, \mathbb{l}_n) = \kappa_n$. Then, there exists a stochastic $L$-equitable CR $G_n\in\Pi^*_n(\kappa_n, L_i)$ defined by a distribution function $q^{*}_{G_n, \mathbb{l}_n}$ over $\Pi^D(\kappa_n)$ such that the following conditions hold:

\begin{itemize}
    \item[\textbf{P\ref{lemma: Lequitprobs}.1}] For each $l$, and all $i$, 
        $\mathbb{P}_n^F(A^{G_n+}_i=1 \mid \mathbb{L}_n=\mathbb{l}_n, L_i=l) = p^{G_n}(l, \mathbb{l}_n).$
    \item[\textbf{P\ref{lemma: Lequitprobs}.2}] For each $\mathbb{b}_n \in (\mathbb{b}_n)_{\mathbb{l}_n, \kappa_n}$ there exists a constant $p_{\mathbb{b}_n}\in[0,1]$ such that 
            $q^{*}_{G_n, \mathbb{l}_n}(G_n^D) = p_{\mathbb{b}_n}$
        for all $G_n^D\in(G_n^D)_{\mathbb{b}_n}$. 
    
\end{itemize}

\end{proposition}

Condition $\textbf{P\ref{lemma: Lequitprobs}.1}$ asserts the existence of a weight function characterizing a stochastic $L$-equitable CR with $q_{i, \mathbb{l}_n}$ equivalent to a valid investigator-specified $p^{G_n}$. Condition $\textbf{P\ref{lemma: Lequitprobs}.2}$ states that one such weight function $q^{*}_{G_n, \mathbb{l}_n}$ is that for which all deterministic regimes resulting in a given composition of assigned treatment and covariates $\mathbb{b}_n$ will be assigned the same weight. 

The existence of such a distribution (Condition $\textbf{P\ref{lemma: Lequitprobs}.1}$) for any valid $q^*_{i,\mathbb{l}_{n}}$ follows from an application of the Birkhoff von-Neumann Theorem \citepSM{birkhoff1946three}, which states that any bistochastic matrix is decomposable as a linear combination of permutation matrices with non-negative weights. A bistochastic matrix is an $n\times n$ matrix where each row and each column sums to 1, and a permutation matrix is a special bistochastic matrix where each row and column are strictly comprised of 0's and a single element 1. In this case, each column in the bistochastic matrix would correspond to a single individual and each row would correspond to a single treatment option - with the first $\kappa_n$ rows corresponding to receiving treatment and the final $n-\kappa_n$ rows corresponding to \textit{not} receiving treatment. With these interpretations, an investigator's valid specification of $q^*_{i,\mathbb{l}_{n}}$ would correspond to the bistochastic matrix for which the $i,j^{th}$ entry would equal $\frac{n\times q^*_{i,\mathbb{l}_{n}}(1 \mid L_i)}{\kappa_n}$ for $j\leq \kappa_n$ and $\frac{n\times q^*_{i,\mathbb{l}_{n}}(0 \mid L_i)}{\kappa_n}$ otherwise. Each of the permutation matrices then would correspond to a specific rank ordering in $S(\{1,\dots, n\})$, and its positive weight in the Birkhoff decomposition would correspond to their probability mass in the distribution used by the investigator to implement the CR. 

A given distribution over the fully-deterministic regimes satisfying an investigator-specified $q^*_{i,\mathbb{l}_{n}}$ will not in general be unique, and furthermore, many parameters of the joint distribution of $\mathbb{O}_n^{G_n}$ will in general vary with $G_n$ on the subset of $\Pi^*_n(\kappa_n, L_i)$ with fixed $q^*_{i,\mathbb{l}_{n}}$. For these reasons, and other pragmatic concerns related to implementation, investigators may be interested in a particular $L$-equitable regime $G_n$. Specifically, there will exist a convex combination of rank orderings whereby: 1) the \textit{only} rank orderings with non-zero weights will result in $\mathbb{b}_n(1, l) \in \{ \lfloor n\times q^*_{i,\mathbb{l}_{n}}(1 \mid l) \rfloor, \lceil n\times q^*_{i,\mathbb{l}_{n}}(1 \mid l) \rceil \}$ for each $l$; and 2) if rank orderings are grouped by a common $\mathbb{b}_n$ under a limitation $\kappa_n$, the number of such groups with positive weights will be at most $|\mathcal{L}|$, the cardinality of $L_i$. These properties follow from application of results in \citetSM{budish2013designing}, who also provided an algorithm for finding such groupings and corresponding weights in computational time that is polynomial with respect to $|\mathcal{L}|$.

\subsection{Identifying intervention densities for CRs} \label{appsubsec: intdens}

In this section we provide a series of results identifying intervention densities $q^*_i$ and $\mathbb{q}^*_n$ for the class of CRs reviewed in this Appendix. To obtain the interventional density $q^*_i$ of the individual-level g-formula for an arbitrary regime $G_n\in\Pi_n(\kappa_n)$, we leverage the following identity which holds simply by laws of probability:
$$q^*_{i}(1,  l)  = \sum\limits_{\mathbb{l}_n} q^*_{i, \mathbb{l}_n}(1\mid  l)\mathbb{Q}_{L, n\mid l}(\mathbb{l}_n),$$ 
where $\mathbb{Q}_{L,n\mid l}(\mathbb{l}_n) = \mathbb{P}_n(\mathbb{L}_n=\mathbb{l}_n \mid L_i=l)$. Therefore, whenever $\mathbb{Q}_{L, n\mid l}$ and $q^*_{i, \mathbb{l}_n}$ are identified in terms of $Q_L$ then so will $q^*_{i}(1\mid  l)$. 

The following proposition provides identification results for $q^*_{i, \mathbb{l}_n}$ for several classes of CRs. In the following,  we let $(\mathbf{l}_n)_{i,l, \mathbb{l}_n}$ denote the set of configurations $\mathbf{l}_n$ with the common entry $l_i=l$ and common composition $\mathbb{l}_n$, and we let $(\mathbf{b}_n)_{\mathbb{b_n}}$ denote the set of configurations $\mathbf{b}_n$ with the common composition $\mathbb{b_n}$. The identities therein follow simply by definition of the regimes and laws of probability.

\begin{proposition} \label{prop: intdens_general}
Consider a $\mathbb{P}_n^F\in\mathcal{M}_n^{U}$.  We have for any $G_n\in \Pi_n(\kappa_n)$
 $$q^*_{i, \mathbb{l}_n}(a \mid l) = \sum\limits_{(\mathbf{l}_n)_{i,l, \mathbb{l}_n}} .   
    \Bigg[\sum\limits_{\Pi^{D}_n(\kappa_n)} I\Big(a = \big\{a_i \in G^D_n(\mathbf{l}_n)\big\}\Big) q^{*}_{G_n}(G^D_n)\Bigg]\mathbb{P}_n(\mathbf{L}_n=\mathbf{l}_n \mid L_i=l, \mathbb{L}_n=\mathbb{l}_n),$$
and
\begin{align*}
 \mathbb{q}^*_n(\mathbb{b}_n \mid \mathbb{l}_n) = &\sum\limits_{(\mathbf{b}_n)_{\mathbb{b_n}}} \Bigg[\sum\limits_{\Pi^{D}_n(\kappa_n)}I\Big(\mathbf{a}_n  = G^D_n(\mathbf{l}_n)\Big)q^{*}_{G_n}(G^D_n)\Bigg] 
      \mathbb{P}_n(\mathbf{L}_n=\mathbf{l}_n \mid \mathbb{L}_n=\mathbb{l}_n).   
\end{align*}
\end{proposition}

Note that an investigator interested in identifying an individual-level parameter $\mathbb{E}_{\mathbb{P}^F_n}[ h'(O_i^{G_n+})]$, or linear combination thereof, for example, $\mathbb{E}_{\mathbb{P}^F_n}[ \overline{Y}_n^{{G}_n}]$, may sufficiently specify their regime $G_n$ via $p^{G_n}$, as opposed to the specific functional form of $G_n$. While there is indeed a large class of regimes $G'_n\in\{\Pi_n^*(\kappa_n, L_i): \mathbb{q}_{i, \mathbb{l}_n}^*(1 \mid l) = p^{G_n}(l, \mathbb{l}_n) \ \forall \ \mathbb{l}_n\}$ with common $p^{G_n}$, i.e., stochastic $L$-equitable regimes $G^*_{n}$ that would result in a particular set of assigned treatment probabilities described by $p^{G_n}$, the intervention densities of the individual-level g-formula $f_i^{G'_n}$ will be invariant in $i$ under $\mathcal{M}_n^{AB}$ and identical across all such regimes $G'_n$, $$q^*_{i}(1 \mid  l)  = \sum\limits_{\mathbb{l}_n} p^{G_n}(l, \mathbb{l}_n)\mathbb{Q}_{L, n\mid l}(\mathbb{l}_n),$$ which is not a function of $G'_n$.  

However, when interest is in a compositional parameter $\mathbb{E}_{\mathbb{P}^F_n}[ h(\mathbb{O}_n^{G_n+})]$, such a specification will not be sufficient: the intervention densities $\mathbb{q}_n^*$ in $\mathbb{f}^{G_n}$ will vary across regimes with common $p^{G_n}$, and computation of $\mathbb{q}_n^*$ must proceed analogously to an arbitrary regime $G_n\in\Pi_n(\kappa_n)$, as in Proposition \ref{prop: intdens_general}.

Now, we provide closed-form expressions for intervention densities $q^*_{i,\mathbb{l}_{n}}$ under the set of deterministic regimes, $\Pi^{d*}_n(\kappa_n, L_i)$, i.e., those characterized by a set of rank functions $\Lambda_{\mathbb{l}_n}$. These derivations also apply to the $L$-rank-preserving regimes, $\Pi^{d}_n(\kappa_n, L_i, \Lambda)$, by replacing $\Lambda_{\mathbb{l}_n}$  with $\Lambda$ wherever the latter appears in subsequent expressions. Define the random variable $\Lambda_{\mathbb{l}_n}(L_i)$ as the individual's coarsened rank under regime $G_n$ in a composition with $\mathbb{l}_n$. Let $S_{\Lambda_{\mathbb{l}_n}}(m) \coloneqq \sum\limits_{\mathcal{L}}\mathbb{l}_n(l)I(\Lambda_{\mathbb{l}_n}(l) \geq m)$ denote the number of individuals with coarsened rank group  greater than or equal to $m$ in a cluster with composition $\mathbb{l}_n$, and let $S_{\Lambda_{\mathbb{l}_n}}^-$ be a corresponding function denoting the number of such individuals with coarsened rank group strictly greater than $m$ in such a cluster. In a cluster with composition $\mathbb{l}_n$, there will exist a fixed coarsened rank representing the last group to receive at least one treatment unit under $G_n$. We denote the coarsened rank for this group by $\omega_{\mathbb{l}_n}$:

  \begin{align}
    \omega_{\mathbb{l}_n} \coloneqq \textup{inf}\bigg\{ c \in \mathbb{R}: S_{\Lambda_{\mathbb{l}_n}}(c)  < \kappa_{n}\bigg\}. \label{eq: omega_finite}
 \end{align}

\begin{proposition} \label{prop: intdens_Lrankpres}
Consider a $\mathbb{P}_n^F\in\mathcal{M}_n^{U}$.  We have for any $G_n\in \Pi^{d*}_n(\kappa_n, L_i)$

  \begin{align}
         q^*_{i,\mathbb{l}_{n}}(1 \mid l)  
             & =  \begin{cases}
                     \frac{\kappa_{n} - S_{\Lambda_{\mathbb{l}_n}}^-(\omega_{\mathbb{l}_n})}{S_{\Lambda_{\mathbb{l}_n}}(\omega_{\mathbb{l}_n})-S_{\Lambda_{\mathbb{l}_n}}^-(\omega_{\mathbb{l}_n})}   & : \Lambda_{\mathbb{l}_n}(l)= \omega_{\mathbb{l}_n}, \\
                     I\Big(\Lambda_{\mathbb{l}_n}(l) > \omega_{\mathbb{l}_n}\Big) & : \text{otherwise},
                \end{cases} \label{eq: gstar_tv_finite_cond}
\end{align}

and for a $\mathbb{P}_n^F\in\mathcal{M}^{AB}_n$, we have:

\begin{align}
 \mathbb{q}^*_n(\mathbb{b}_n \mid \mathbb{l}_n) = \begin{cases}
    0  &  \text{if there exists an } l\in\mathcal{L} \text{ such that:} \\
       & \ \ \ \  \Lambda_{\mathbb{l}_n}(l) > \omega_{\mathbb{l}_n} \text{ and } \mathbb{b}_n(1, l) \neq \mathbb{l}_n(l); \text{ or } \\
       & \ \ \ \  \Lambda_{\mathbb{l}_n}(l) < \omega_{\mathbb{l}_n} \text{ and } \mathbb{b}_n(1, l) \neq 0; \text{ or }  \\  
       & \ \ \ \  \sum\limits_{l': \Lambda_{\mathbb{l}_n}(l')=\omega_{\mathbb{l}_n}}\mathbb{b}_n(1, l') \neq \kappa_{n} - S^-_{\Lambda_{\mathbb{l}_n}}(\omega_{\mathbb{l}_n}),\\ 
    \frac{
            \prod\limits_{l: \Lambda_{\mathbb{l}_n}(l)=\omega_{\mathbb{l}_n}}
                \begin{pmatrix*}[c]
                \mathbb{b}_{n}(1, l) \\
                \mathbb{l}_{n}(l)
                \end{pmatrix*}
        }{
                \begin{pmatrix*}[c]\kappa_{n} - S^-_{\Lambda_{\mathbb{l}_n}}(\omega_{\mathbb{l}_n})\\
               S_{\Lambda_{\mathbb{l}_n}}(\omega_{\mathbb{l}_n}) - S^-_{\Lambda_{\mathbb{l}_n}}(\omega_{\mathbb{l}_n})\end{pmatrix*}
        } &  \text{ otherwise}. 
 \end{cases} \label{eq: gstar_tv_finite_cond2}
\end{align}

\end{proposition}

The identity in \eqref{eq: gstar_tv_finite_cond} follows by recognizing that $ \omega_{\mathbb{l}_n}$ may be interpreted as the last rank-group in which individuals could receive treatment with positive probability. By construction of the regime, individuals in rank-groups after $\omega_{\mathbb{l}_n}$ certainly will not receive a treatment unit and individuals in rank-groups before $\omega_{\mathbb{l}_n}$ certainly will receive a treatment unit. Individuals in rank-group $\omega_{\mathbb{l}_n}$ have the same probability of receiving treatment, by definition of the regime. As there are $\kappa_{n} - S_{\Lambda_{\mathbb{l}_n}}^-(\omega_{\mathbb{l}_n})$ treatment units remaining for this group, and $S_{\Lambda_{\mathbb{l}_n}}(\omega_{\mathbb{l}_n})-S_{\Lambda_{\mathbb{l}_n}}^-(\omega_{\mathbb{l}_n})$ individuals in the group, then their quotient gives the probability.

The identity in \eqref{eq: gstar_tv_finite_cond2} follows by recognizing that any composition that does not have the appropriate counts of treatment within each rank group (per the preceding arguments) will have 0 probability. When there are multiple covariate patterns that map to the rank group $\omega_{\mathbb{l}_n}$, then each of the valid compositions will be equally likely with probability equal to their relative frequency, given by a quotient derived using simple combinatoric identities.

\begin{corollary}
    Consider a $\mathbb{P}_n^F\in\mathcal{M}_n^{AB}$. We have for any $G_n\in \Pi^{d*}_n(\kappa_n, L_i)$ where $\Lambda_{\mathbb{l}_n}$ is bijective that $q^*_n(\mathbb{b}_n \mid \mathbb{l}_n) = I\Big(\mathbb{b}_n = \mathbb{b}_{\mathbb{l}_n}^{G_n+}\Big).$
\end{corollary}

When $\Lambda$ is bijective, i.e. $L_i\neq L_j \implies \Lambda(L_i)\neq \Lambda(L_j)$, then there is only a single value $l$ such that $\Lambda(l)=\omega_{\mathbb{l_n}}$ and thus there is only a single composition $\mathbb{b}_{\mathbb{l}_n}^{G_n+}$ compatible with $G_n$ given a composition $\mathbb{l}_n$.

Under $\mathcal{M}_n^{AB}$, the terms $\mathbb{Q}_{L, n\mid l}$ and $\mathbb{P}_n(\mathbf{L}_n=\mathbf{l}_n \mid \mathbb{L}_n=\mathbb{l}_n)$ are both identified in terms of $Q_L$ (see Appendix \ref{appsec: proofs}). Therefore,  $q^*_i$ and  $\mathbb{q}^*_n$ are both both identified in terms of $Q_L$ under $\mathcal{M}_n^{AB}$ in all the expressions in Propositions \ref{prop: intdens_general} and \ref{prop: intdens_Lrankpres}, since the remaining terms are investigator-specified or known functions. Although identified, computations of $q_i^*$ and $\mathbb{q}^*_n$ for generalized regimes in Proposition \ref{prop: intdens_general} involve evaluating very large sums over $(\mathbf{b}_n)_{\mathbb{b_n}}$ or $(\mathbf{l}_n)_{i,l}$ and $\Pi^{D}_n(\kappa_n)$ for a general non-deterministic regime.  

Evidently, $\mathbb{q}^*_n$ may be easier to compute than $q^*_i$ for an $L$-equitable deterministic CR. However, evaluating the g-formulae $f_i^{G_n}$ and $\mathbb{f}^{G_n}$ will involve summing over the set of possible compositions $\mathbb{l}_n$ in either case. Nevertheless, the resulting computations will involve summations of orders of magnitude fewer terms than those that require summing over $(\mathbf{b}_n)_{\mathbb{b_n}}$ or $(\mathbf{l}_n)_{i,l}$ and $\Pi^{D}_n(\kappa_n)$. Furthermore, additional computational advantages are possible under $V$-equitable regimes.

\subsection{Computational considerations for regimes that are at most a function of $\mathbf{V}_n$} \label{appsubsec: computation}

When parameters for finite clusters of size $n^*$ are of interest, practical concerns may arise about plug-in estimation of the compositional g-formula, or the intervention densities for the individual-level g-formula. We re-state these expressions here, for convenience:

\begin{align}
    \mathbb{E}_{\mathbb{P}^F_{n^*}}[ h(\mathbb{O}_{n^*}^{G_n+})] = & \sum\limits_{\mathbb{o}_{n^{\ast}}}h(\mathbb{o}_{n^{\ast}})  {\mathbb{f}}^{G_{n^{\ast}}}(\mathbb{o}_{n^*}) \label{eq: compgformsum}\\ 
    q^*_{i}(1 \mid  l)  = & \sum\limits_{\mathbb{l}_{n^*}} q^*_{i, \mathbb{l}_{n^*}}(1\mid  l)\mathbb{Q}_{L, n^*\mid l}(\mathbb{l}_{n^*}) \label{eq: igformsum}.
\end{align}

If we assume a binary outcome, then the number of terms in the summand of \eqref{eq: compgformsum} (i.e., the cardinality of $\mathbb{o}_{n^{\ast}}$) may be formulated as the number of weak compositions of $n^*$ into exactly $4|\mathcal{L}|$ parts, and thus is obtained by the binomial coefficient $\Big(^{n^* + 4|\mathcal{L}| - 1}_{n^*}\Big)$. Similarly, the number of terms in in the summand of \eqref{eq: igformsum} is obtained by $\Big(^{n^* + |\mathcal{L}| - 1}_{n^*}\Big)$. 

We emphasize that the cluster size defining a parameter of interest, $n^*$, is specified separately from $n$, the size of the cluster representing the observed data on hand. For example, $n$ will often be moderately large, whereas when interest is in a finite-cluster parameter (vs. a parameter defined by an arbitrarily large cluster), $n^*$ will be relatively small, i.e. $n^*<<n$. For example, our data analysis (presented in Section \ref{sec: icu example} and Appendix \ref{appsec: misc})  had $n=13,011$ and considered finite-cluster parameters corresponding to $n^*=20$, representing a reasonable number of beds in a moderately-sized intensive care unit. Nevertheless, $|\mathcal{L}|$ will often be large, reflecting an investigators concerns around adequate confounder control. In our data analysis, we had $|\mathcal{L}|=37632$, and thus $|\mathbb{O}_{n^*}|$ is prohibitively large. 

Nevertheless, investigators are often interested in regimes defined by $V_i = c(L_i)$, a coarsening of the covariates used for confounder control, such that $|\mathcal{V}| << |\mathcal{L}|$. In other words, investigators consider a regime $G_{n^*}: \{\mathcal{V}\}^{n^*} \times [0,1] \rightarrow \{0, 1\}^{n^*}.$ In our data analysis, we considered a $V$-rank-preserving regime ${G}_{n^*=20}^{\Lambda_2}$ adapted from a proposal of \citetSM{nyc2015ventilator}, where $V_i$ is the value of an individual's trichotomized SOFA score. In the following, we present a re-formulation of \eqref{eq: compgformsum} and \eqref{eq: igformsum}, which leverages distinctions between $|\mathcal{V}|$ and $|\mathcal{L}|$. Let $U_i \equiv \{A_i, V_i\}$ and $W_i \equiv \{Y_i, A_i, V_i\}$. Furthemore, let $Q_V$ denote the marginal densities of $V_i$ and $Q_{L\mid v}$ be the conditional density of $L_i$ given $V_i=v$ under $\{\mathbb{P}_n, \mathbb{P}_{n^*}\}\in\mathbb{P}_0^F$ following condition $D$. Then we consider the following definitions and subsequent proposition:

\begin{align}
    \mathbb{Q}_{V, n^*}(\mathbb{v}_{n^*}) & \coloneqq  \frac{n^*!}{\prod\limits_{l\in\mathcal{V}}\mathbb{v}_{n^*}(v)!}\prod\limits_{v\in\mathcal{V}}Q_{V}(v)^{\mathbb{v}_{n^*}(v)}\nonumber\\
    \mathbb{q}^{**}_{{n^*}}(\mathbb{u}_{n^*} \mid \mathbb{v}_{n^*}) & \coloneqq \mathbb{P}_{n^*}(\mathbb{U}^{G_n^*+}_{n^*}=\mathbb{u}_{n^*}\mid \mathbb{V}_{n^*}=\mathbb{v}_{n^*}) \nonumber\\
    Q^*_Y(y \mid a, v) & \coloneqq \sum\limits_{l:c(l)=v}Q_{Y}(y \mid a, l)Q_{L\mid v}(l)\nonumber\\
    \mathbb{Q}^*_{Y, n^*}(\mathbb{w}_{n^*} \mid \mathbb{u}_{n^*}) & \coloneqq  \prod\limits_{ a, v} \frac{\mathbb{u}_{n^*}(a,v)!}{\prod\limits_{y\in\mathcal{Y}_i}\mathbb{w}_{{n^*}}(y,a,v)!}\prod\limits_{y\in\mathcal{Y}_i}Q^*_Y(y \mid a, v) ^{\mathbb{w}_{{n^*}}(y,a,v)} \nonumber
\end{align}
\begin{align}
     \mathbb{f}_{red}^{G_{n^*}}(\mathbb{w}_{n^*}) & \coloneqq \mathbb{Q}^*_{Y,n^*}(\mathbb{w}_{n^*} \mid \mathbb{u}_{n^*})
                \mathbb{q}^{**}_{n^*}(\mathbb{u}_{n^*} \mid  \mathbb{v}_{n^*})\mathbb{Q}_{V,n^*}(\mathbb{v}_{n^*}) \label{eq: compgformred} \\
    q^*_{i, \mathbb{v}_{n^*}}(1\mid  v) & \coloneqq \mathbb{P}_{n^*}({A}^{G_n^*+}_i=a \mid V_i=v, \mathbb{V}_{n^*}=\mathbb{v}_{n^*}) \label{eq: indintdensV}
\end{align}

We refer to expression \eqref{eq: compgformred} as the \textit{reduced} compositional g-formula. Expression \eqref{eq: indintdensV} is defined analogously to $q^*_{i, \mathbb{l}_{n^*}}$ except with $v$ and $\mathbb{v}_{n^*}$ replacing  $l$ and $\mathbb{l}_{n^*}$. The following proposition links this g-formula to a counterfactual parameter of interest, where we define $\mathbb{Q}_{V, {n^*}\mid v}$ analogously to $\mathbb{Q}_{L, {n^*}\mid l}$.
\begin{proposition}\label{prop: compgformred}
Consider a law $\mathbb{P}_n^F\in\mathcal{M}^{AB}_n$, a real-valued function $h''$ 
and a regime $G_n \in \Pi_n(\kappa_n)$, with $G_n$ a function of at most $\{\mathbf{V}_n, \delta\}$, i.e., $G_n: \{\mathcal{V}\}^n \times [0,1] \rightarrow \{0, 1\}^n$, and with $V_i\coloneqq c(L_i)$ for all $i$. Then $\mathbb{E}_{\mathbb{P}^F_n}[ h(\mathbb{W}_n^{G_n+})] =  \sum\limits_{\mathbb{w}_{n}}h''(\mathbb{w}_n)  \mathbb{f}^{G_n}_{red}(\mathbb{w}_{n})$
whenever the right hand side of the equation is well-defined, and also that $q^*_{i}(1\mid  l)  = \sum\limits_{\mathbb{v}_n} q^*_{i, \mathbb{v}_n}(1\mid  c(l))\mathbb{Q}_{V, n\mid c(l)}(\mathbb{v}_n).$
\end{proposition}

For deterministic regimes in $\Pi_{n^*}^{d}(\kappa_{n^*}, V_i, \Lambda)$, we obtain $q^*_{i, \mathbb{v}_{n^*}}$ and $\mathbb{q}^{**}_{n^*}$ analogously to Proposition \ref{prop: intdens_Lrankpres}. Arguments are isomorphic, so we do not repeat them here. Proposition \ref{prop: compgformred} shows us that plug-in estimation of compositional and individual-level g-formula parameters poses far fewer computational demands in practice, than is suggested upon inspection of their unreduced functionals presented in Section \ref{subsec: IDres}. Additional computational gains may be obtained for $ \sum\limits_{\mathbb{w}_{n^*}}h''(\mathbb{w}_{n^*})  \mathbb{f}^{G_{n^*}}_{red}(\mathbb{w}_{n^*})$ if one only sums over elements $\mathbb{w}_{n^*}$ with positive probability under $G_{n^*}$. When $\Lambda$ is bijective for $V_i$, then there is only a single composition $\mathbb{u}_{n^*}$ with positive probability, given $\mathbb{v}_{n^*}$. Thus, when $Y_i$ is binary, the number of terms with positive probability is equal to $\Big(^{n^* + 2|\mathcal{V}| - 1}_{n^*}\Big)$. For example, in the real data analysis of Section \ref{sec: icu example}, inference for $\mathbb{P}_{n^*=20}^F\Big(\mathbb{Y}_{n^*}^{G^{\Lambda_2}_{n^*=20}}(1)=x\Big)$ via parametric compositional g-formula estimation involved summing over $\Big(^{n^* + 2|\mathcal{V}| - 1}_{n^*}\Big) = \Big(^{20 + 2\times3 - 1}_{20}\Big)=53130$ terms. We have provided code for these computations at \href{https://github.com/AaronSarvet/OptimalLimited}{GitHub}.

\subsection{Sub-cluster regimes} \label{apsubsec: subclus}

Let $\mathbf{W}_n$ be an unmeasured random vector denoting the membership of each individual $i$ to one of $m_n<n$ sub-clusters that partition the (super) cluster of size $n$. Let $\boldsymbol{\kappa}_{m_n} \equiv \{\kappa_{n, 1}, \dots, \kappa_{n, m_n}\}$ denote a partition of the $\kappa_n$ treatment units across the $m_n$ sub-clusters. In a slight abuse of notation, we let $\boldsymbol{\kappa}_{m_n}$ both denote a random variable, and an arbitrary realization.

Now we define a general class of sub-cluster regimes:

\begin{definition}[Sub-cluster regimes]
Let $\delta$ be partitioned into 
$\delta_{\mathbf{R}} \equiv \{\delta_{\mathbf{R}, 1}, \dots, \delta_{\mathbf{R}, m_n}\}$, and $\delta'_{\mathbf{A}} \equiv \{\delta'_{\mathbf{A}, 1}, \dots, \delta'_{\mathbf{A}, m_n}\}$. Let $\mathbf{G}_m$ denote an $m\times m$ matrix of generalized resource-limited regimes, whose entry in its $r^{th}$ row and $p^{th}$ column, denoted by $[\mathbf{G}_m]_{r,p}$,  is some element in $\Pi_r^F(p)$, i.e. a generalized regime for a cluster of size $r$ with $p$ treatment units. Then, define the class of sub-cluster regimes $\Pi^F_{n, m_n}(\mathbf{G}_m)$ for which the following properties hold:

\begin{itemize}
        \item[1. ] Within each sub-cluster characterized by $W_i=c$, each  individual index is renumbered from $1,\dots, \mathbb{W}_n(c)$, according to the rank ordering of the original indices among such individuals.
        \item[2. ] Within each such sub-cluster characterized by $W_i=c$, individuals are assigned treatment according to $[\mathbf{G}_m]_{\mathbb{W}_n(c),\kappa_{n,c}}$ using the set of randomizers  $\{ \delta_{\mathbf{R}, c}, \delta'_{\mathbf{A}, c}\}$.
        \item[3. ] The sets of randomizers $\Big\{\{\delta_{\mathbf{R}, c}, \delta'_{\mathbf{A}, c}\} \mid c\in \{1,\dots, m_n\}\Big\}$ are mutually independent. 
    \end{itemize}

\end{definition}

In words, sub-cluster dynamic regimes determine treatment assignment independently by sub-cluster, where the regime implemented in each sub-cluster $c$ is some arbitrary CR in the general class defined by the number of individuals in that sub-cluster, $\mathbb{w}_n(c)$ and its corresponding treatment supply $\kappa_{n, c}$. Note that when their exists only a single sub-cluster, $m_n=1$, then $\Pi^F_{n, m_n}(\mathbf{G}_m)\equiv\Pi^F_{n}(\kappa_{n}).$ 

We leverage sub-cluster regimes in defining an alternative positivity condition in Appendix \ref{appsec: estim}, which, together with weak positivity condition $C0$, implies the strong positivity condition $C1$ articulated in Section \ref{sec: estim}.

\begin{figure}[h]
    \centering
\def\sizeA{1648800}
\def\sizeB{1251900}
\def\sizeC{750300}
\centering
\begin{tikzpicture}[x=5cm, y=5cm, scale=1.5,
  textDiscLabel/.style={%
    decorate,
    decoration={%
        text along path,%
        text align=center,
        reverse path=true,%
        raise=-2ex,%
        text={#1}}
  }]



    \draw  
    (0,0) arc (90:450:1) -- cycle;
\draw  
    (0,0) arc (90:450:.925) -- cycle;
\draw [very thick]
    (0,0) arc (90:450:.85) -- cycle;
\draw 
    (0,0) arc (90:450:.775) -- cycle;
\draw 
    (0,0) arc (90:450:.7) -- cycle;
\draw [dashed] 
    (0,0) arc (90:450:.625) -- cycle;
\draw [ultra thick]
    (0,0) arc (90:450:.55) -- cycle;
\draw  
    ({sin(325)},{-1 + cos(325)} )  arc (125:485:.4) -- cycle;
\draw  
    ({sin(35)},{-1 + cos(35)} )  arc (55:415:.4) -- cycle;

  \node[anchor=south] (NP) at (0,-2)  { $\Pi_n^F(\kappa_n)$};
  \node[anchor=south] (NP) at (0,-2*0.925)  { $\Pi_n(\kappa_n)$};
  \node[anchor=south] (NP) at (0,-2*0.85)  { $\Pi^*_n(\kappa_n, L_i)$};
  \node[anchor=south] (NP) at (0,-2*0.773)  { $\Pi^d_n(\kappa_n, L_i, q^*_{\Lambda_{\mathbb{l}_n}})$};
  \node[anchor=south] (NP) at (0,-2*0.695)  { $\Pi^{d*}_n(\kappa_n, L_i)$};
  \node[anchor=south] (NP) at (0,-2*0.61)  { $\Pi^{d\epsilon}_n(\kappa_n, \Lambda, Q_L)$};
 \node[anchor=south] (NP) at (0,-2*0.52)  { $\Pi^{d}_n(\kappa_n, L_i, \Lambda)$};
 \node[anchor=south] (NP) at (-.3,-.85)  { $\Pi^{D}_n(\kappa_n)$};
 \node[anchor=south] (NP) at (.3,-.85)  { $\Pi^{F}_{n, m_n}(\mathbf{G}_m)$};

\end{tikzpicture}
    \caption{Diagram depicting relations between classes of cluster dynamic treatment regimes (C-DTR). $\Pi_n^F(\kappa_n)$: generalized constrained C-DTRs; $\Pi_n(\kappa_n)$: constrained C-DTRs with exactly $\kappa_n$ treated units; $\Pi^*_n(\kappa_n, L_i)$: stochastic $L$-equitable C-DTRs; $\Pi^d_n(\kappa_n, L_i, q^*_{\Lambda_{\mathbb{l}_n}})$: mixed-deterministic C-DTRs; $\Pi^{d*}_n(\kappa_n, L_i)$: deterministic C-DTRs; $\Pi^{d\epsilon}_n(\kappa_n, \Lambda, Q_L)$: $\epsilon-L$-rank-preserving C-DTRs; $\Pi^{d}_n(\kappa_n, L_i, \Lambda)$: $L$-rank-preserving DTRs; $\Pi^{D}_n(\kappa_n)$: fully-deterministic C-DTRs; $\Pi^{F}_{n, m_n}(\mathbf{G}_m)$  generalized sub-cluster C-DTRs.}
    \label{fig: classes}
\end{figure}

\clearpage
\section{Inference} \label{appsec: estim}

This appendix elaborates on several inference-related topics deferred from the main text.  Section \ref{appsec: IPW} motivates a set of inverse-probability weighted (IPW) estimators of individual-level g-formula parameters, and notes important distinctions from conventional IPW estimators.
Section \ref{appsec: counterexamp_cons} provides an example demonstrating the insufficiency of weak positivity condition $C0$ for the general existence of a consistent estimator defined with respect to an asymptotic law $\mathbb{P}_0$ following condition $D$, and thus motivating the strong positivity condition $C1$ presented in Section \ref{sec: estim}.
Section \ref{appsec: altpos} provides an alternative to positivity condition $C1$ that is stronger but that may be easier to justify by subject matter experts. 
Section \ref{appsec: conditions} discusses inference for large-cluster parameters, which faces challenges due to the irregularity of these parameters when they are defined by a general CR. Therein, we develop conditions and results for parametric, semiparametric, and semiparametric efficient estimation of large-cluster parameters defined by candidate and optimal CRs, adapting results from \citetSM{luedtke2016optimal}. Notably, the conditions considered in this section ensure the regularity of the large-cluster parameters, in part through restrictions on the class of regimes considered.
Section \ref{appsec: online} presents strategies for consistent estimation of large-cluster parameters under conditions that do not ensure their regularity, adapted from online-learning-based estimators developed by \citetSM{luedtke2016statistical}. These strategies are important whenever investigators are interested in clinically implementing regimes defined by low-dimensional covariates, which often enjoy practical advantages over their high-dimensional counterparts.

\subsection{IPW} \label{appsec: IPW}
To motive a class of inverse-probability weighted estimators for cluster parameters, define the weight variable $$\tilde{W}_i^{G_{n^*}, j}  \coloneqq \frac{\tilde{q}^*_j(A_i \mid L_i)}{\tilde{q}_n(A_i \mid L_i)}.$$ Then, using simple algebra, we can show that $$\frac{1}{n}\sum\limits_{i=1}^nY_i\tilde{W}_i^{G_{n^*},j}  = \sum\limits_{o}y\tilde{f}_{j}^{G_{n}}(o).$$ As such, this appropriately-weighted empirical mean analogue of $Y_i$ will also be asymptotically consistent for its corresponding causal cluster parameter, $\mathbb{E}_{\mathbb{P}_{n^*}^F}[Y_j^{G_{n^*}}]$ under the conditions of Theorem \ref{theorem: asymptcons}. These relations hold despite the fact that a classical IPW representation of the individual-level g-formula $f_j^{G_{n^*}}(o)$ would suggest a different expression. Motivated by the density ratio $\frac{f_j^{G_{n^*}}(o)}{\mathbb{P}_{n^*}(O_i=o)} = \frac{q^{*}_{j}(a \mid l)}{q_{i}(a \mid l)}$ (under the conditions of Theorem \ref{theorem: asymptcons}), define the weight variable $$W^{G_{n^*}, j}_i\coloneqq \frac{q^{*}_{j}(A_i \mid L_i)}{q_{i}(A_i \mid L_i)}.$$ We thus have an IPW representation of the individual-level g-formula of Theorem \ref{theorem: CDTRID}:  by the Radon-Nikodym Theorem, as in \citetSM{richardson2013single}, we have that $$\mathbb{E}_{\mathbb{P}^F_{n^*}}[Y_j^{G_{n^*}}] =\mathbb{E}_{\mathbb{P}_{n^*}}[Y_iW^{G_{n^*}, j}_i]$$ under the conditions of Theorem \ref{theorem: CDTRID}. We emphasize that the weights $W^{G_{n^*}, j}_i$ are constructed with the density functions $q_i$ whereas $\tilde{W}_i^{G_{n^*},j}$ are constructed with a density function $\tilde{q}_n$. While we have that $\tilde{q}_n$ converges to $\overline{q}_0$ (under strong finite-cluster positivity $C1$), we do not have any guarantees that  $\overline{q}_0=q_i$. The following corollary may be illustrative, where we note the algebraic relation 
        $$\frac{1}{n}\sum\limits_{i=1}^nY_i  = \sum\limits_{o\in\mathcal{O}_i}y\tilde{Q}_{Y}(y \mid a, l) \tilde{q}_n(a \mid l) \tilde{Q}_L(l).$$

\begin{corollary}
    Consider an asymptotic law $\mathbb{P}_0^F$ following conditions $D$ and $\mathbf{f}_0$ such that $\lim\tilde{q}_n(a \mid l) = \overline{q}_0(a \mid l)$ for all $a,l$. Then, $\lim\frac{1}{n}\sum\limits_{i=1}^nY_i  = \mathbb{E}_{\mathbb{P}_0}[\overline{Y}_0] = \mathbb{E}_{\mathbb{P}_0}[Y_j\frac{\overline{q}_0(A_j \mid L_j)}{q_{j}(A_j \mid L_j)}].$
\end{corollary}

Thus, we see that the probability limit of the empirical mean analogue of $Y_i$ is equivalent (under stated conditions) to the expectation of a weighted outcome $Y_j$ whose likelihood-ratio weight has a denominator and numerator defined by the measures $q_j$ and $\overline{q}_0$, respectively. Thus, we see that the probability limit of $\frac{1}{n}\sum\limits_{i=1}^nY_i\tilde{W}_i^{G_n, j}$ will result in the appropriate cancellations of measures, and that it agrees with the classical IPW representation of the individual-level g-formula $\mathbb{E}_{\mathbb{P}_0}[Y_i\frac{\overline{q}_0(A_i \mid L_i)}{q_{i}(A_i \mid L_i)}\frac{q^*_{j}(A_i \mid L_i)}{\overline{q}_0(A_i \mid L_i)}]$.

The preceding results focused on IPW representation of finite cluster parameter $f_j^{G_{n^*}}(o)$. When interest is in the large cluster parameter $f^{\mathbf{G}_0}(o)$, analogous estimators may be obtained using weights $$\tilde{W}_i^{\mathbf{G}_0}  \coloneqq \frac{\tilde{q}^*_0(A_i \mid L_i)}{\tilde{q}_n(A_i \mid L_i)},$$
provided that $f^{\mathbf{G}_0}(o)$ is indeed a regular parameter, for example under the conditions of Lemma \ref{lemma: omegacons} in Appendix Section \ref{appsec: conditions}.

\subsection{Counterexample to the existence of a consistent plug-in estimator of $f_i^{G_n}$ under conditions of Theorem \ref{theorem: CDTRID}.} \label{appsec: counterexamp_cons}

Theorem \ref{theorem: CDTRID} provides a functional of individual-level observed data parameters $Q_L$ and $Q_Y$ that equals a finite cluster counterfactual parameter, provided that functional is well-defined. In line with conventional approaches \citepSM{richardson2013single}, weak cluster positivity condition $C0$ guarantees that functional is well-defined. In section \ref{sec: estim}, we articulated an a notion of asymptotics defined by an asymptotic law $\mathbb{P}_0^F$. Here we demonstrate by example that the asymptotic analogues of the conditions of Theorem \ref{theorem: CDTRID} (i.e. condition $D$) and an asymptotic analogue of condition $C0$ are not sufficient for the consistency result of Theorem \ref{theorem: asymptcons}. 

Let condition $D^*$ denote the intersection of condition $D$ and the following condition, where we remind the reader that we have defined $q_0(a\mid l) \coloneqq \underset{n\to\infty}{\lim}\overline{q}_n(a \mid l)$:

\begin{itemize}
    \item[$D3.$] ${q}_0$ is well-defined and if $\overline{q}^*_n(a \mid l)Q_L(l) > 0 $ then ${q}_0(a \mid l)>0$, for all $a,l$.
\end{itemize}

In other words, we will show by counterexample that the following proposition is false.

\begin{proposition*}
    Consider $\mathbb{P}_0^F$ and $G_{n^*}\in \Pi_{n^{\ast}}(\kappa_{n^{\ast}})$ following condition $D^*$. Then,
    \begin{align}
        & \boldsymbol{\lim} \sum\limits_{o}h'(o)\tilde{f}_{i}^{G_{n^*}}(o) = \mathbb{E}_{\mathbb{P}^F_{n^*}}[ h'(O_i^{G_{n^*}+})].
    \end{align}
\end{proposition*}

We do so by finding a law $\mathbb{P}_0^F$ following condition $D^*$ in which $\tilde{Q}_{Y,n}$ does not converge and thus $\tilde{f}_{i}^{G_{n^*}}(o)$ does not converge. 

Consider a setting with binary covariate $L\in\mathcal{L}\equiv\{0,1\}$, and consider an asymptotic law $\mathbb{P}_0^F$ following conditions $D$, characterized by $Q_L(1) =0.5$ and $\lim\frac{\kappa_n}{n} = 0.5$: asymptotically, half of the cluster can receive treatment. Furthermore, suppose that for each $n$, $\mathbb{P}_n$ is induced in part by $\mathbf{f}_{n, \mathbf{A}}\coloneqq \{f_{A_i}, \dots, f_{A_n}\}$ that is equivalent to some mixed-deterministic CR $G_n \in \Pi^{d}_n(\kappa_n, L_i, q^*_{\Lambda_{\mathbb{l}_n}})$ such that for each $n$, for all $\mathbb{l}_n$ $q^*_{\Lambda_{\mathbb{l}_n}}$ puts equal mass on a $\Lambda$ such that $\Lambda(1)>\Lambda(0)$ and on a $\Lambda^*$ such that $\Lambda^*(0)>\Lambda^*(1)$. Thus, $q_i(1 \mid l) = 0.5$ for each $l$ and each $i$ and all $n$. It follows then that $\overline{q}_0(1 \mid l) = 0.5$, and so $\mathbb{P}_0^F$ will follow condition $D^*$ fro any $G_{n^*}$. 

Under such an asymptotic law $\mathbb{P}_0^F$ we will have that $\lim\limits_{n\to \infty}\mathbb{P}_n\Big( \mathbb{B}_{n}(1, l) = 0 \Big)=0.5$ for each $l$. In words, the probability that a given cluster will have exactly 0 treated individuals with covariate history $l$ is exactly 0.5, for each $l$.  We thus have that $$\lim\limits_{n\to \infty} \mathbb{P}_n\Bigg( \Big| \frac{\mathbb{O}_{n}(y,a, l)}{\mathbb{B}_{n}(a, l)} - Q_{Y}(y \mid a, l) \Big| > \epsilon' \Bigg)=0.5,$$ for all $\epsilon'>0$, which violates the convergence condition for consistent estimation of the parameters described by $Q_Y$ under the asymoptotic regime in this case.

\subsection{Alternative conditions for asymptotic consistency: sub-cluster positivity} \label{appsec: altpos}

In the preceding Appendix Section \ref{appsec: counterexamp_cons}, we illustrated that an asymptotic law following condition $D$ could not provide any guarantees for the convergence of relevant estimators of finite-cluster parameters (e.g. $\tilde{f}^{G_{n^*}}_i(o)$), because no such guarantees can be made about the convergence of $\tilde{{q}}_n$, and by consequence the convergence of $\tilde{Q}_{Y,n}$. Strong cluster positivity condition $C1$ provides direct guarantees for such convergence, by asserting that the probability limit of $\tilde{{q}}_n$ is well-defined. However, this limiting condition is arguably abstract; subject matter experts may find it difficult to reason about its truth-value for a given setting. Therefore, in this Appendix Section we provide a stronger condition on the treatment assignment mechanisms $\mathbf{f}_{0,\mathbf{A}}$ in $\mathbb{S}_0$ (for a given asymptotic law $\mathbb{P}_0$), which, together with $C0$ implies the strong cluster positivity condition $C1$, and which may recognizably capture observed data generating mechanisms in many settings.

Let $P^{\kappa}_0 \equiv \{P^{\kappa}_{1}, P^{\kappa}_{2}, P^{\kappa}_{3}, \dots, P^{\kappa}_{k},\dots, P^{\kappa}_{m},\dots\}$ denote some sequence of distributions, whereby, for each $k$,  $P^{\kappa}_{k}$ is some probability mass function over the space $\{0,\dots, k\}$. Furthermore define the alpha-mixing coefficient introduced by \citetSM{rosenblatt1956central}, commonly used in interference settings (see for example, \citetSM{savje2021average}), $$\alpha_{\mathbb{P}_n}(V_1,V_2) = \sup\limits_{v_1\in\sigma(V_1), v_2\in\sigma(V_2)}\Big| \mathbb{P}_n(V_1\in v_1, V_2\in v_2) - \mathbb{P}_n(V_1\in v_1)\mathbb{P}_n(V_2\in v_2)\Big|,$$ where $\sigma(X)$ is the sub-sigma algebra generated by $X$.  

Now we consider the following alternative conditions for an asymptotic law $\mathbb{P}_0^F$: 

\begin{itemize}
    \item [$C2.$] \textbf{Sub-cluster positivity} $\mathbb{P}_0^F$, is such that for all $a, l$, the following conditions hold:

\begin{itemize}
    \item [1. ] For all $n$, $\mathbb{P}_n(\mathbb{W}_n(c) \leq m ) = 1$ for some fixed $m\in\mathbb{N}^+$ for all $c$.
    \item[2. ] For some fixed random variable $V^{\dagger}$ with support on $\{0, \dots, m\}$, we have that $(\mathbb{W}_n(c) - V^{\dagger}) = o_{\mathbb{P}_0}(1)$ for all $c\in\mathcal{C}_n \subset \{1,\dots, m_n\}$ and $\frac{m_n-|\mathcal{C}_n|}{m_n} = o(1)$.
    \item[3. ] For $c_1 \neq c_2$ two sub-cluster indices, $\alpha_{\mathbb{P}_n^F}(\mathbb{W}_n(c_1), \mathbb{W}_n(c_2)) = o(1)$.
    \item [4. ] For all $n$, $\mathbb{P}_n\in\mathbb{P}_0$ is characterized by an $\mathbf{f}_{n, \mathbf{A}}$ equivalent to a regime $G^*_n\in \Pi^{F}_{n, m_n}(\mathbf{G}_m)$, with $\delta\equiv\epsilon_{\mathbf{A}_n}$.
    \item [5. ]  $\kappa_{n, c}$ is distributed according  $P^{\kappa}_{\mathbb{W}_n(c)}\in P^{\kappa}_0$, with the elements of $\boldsymbol{\kappa}_{m_n}$ mutually independent and independent of $\boldsymbol{\epsilon}_n$.
    \item [6. ] $\epsilon_{\mathbf{W}_n} \independent \boldsymbol{\epsilon}_n$.
\end{itemize}
\end{itemize}

We direct the reader to Appendix \ref{apsubsec: subclus} for explicit definitions of the class of sub-cluster regimes  $\Pi^{F}_{n, m_n}(\mathbf{G}_m)$. Conditions $C2.1$ to $C2.3$ place mild but important conditions on the sub-cluster asymptotic regime. Considered jointly, they ensure that the number of clusters grows sufficiently fast, that the number of individuals per sub-cluster converges to some common distribution for almost all sub-clusters, and that jointly, the sub-cluster sizes are merely weakly dependent. Condition $C2.4$ ensures that any sub-cluster with fixed size $n_c$ will have a common distribution for the number of treatment units available to it, and condition $C2.5$ ensures that, once the number of treatment units available for that sub cluster is additionally fixed to $\kappa_{n,c}$, any such cluster will follow a common regime $[\mathbf{G}_m]_{n_c,\kappa_{n,c}}$. A consequence of $C2.1$ to $C2.5$ considered jointly is that $\tilde{q}_n$ converges sufficiently fast to an empirical average of weakly dependent and identically distributed random variables with a number of terms that grows proportional to $n$.  Condition $C2.6$ ensures that sub-cluster membership may be excluded as a covariate from identification functionals. 

Sub-cluster positivity condition $C2$ describes a data-generating mechanism that may be justifiable via subject matter reasoning in many cases. For example, a setting in which investigators receive data in trenches from an increasing number of intensive care units of bounded size, i.e., that of our observed data setting, would arguably be characterized by condition $C2$. In the following results, we assert that condition $C2$ does indeed imply a stronger version of convergence of $\tilde{q}_n$, and indeed when considered in conjunction with weak positivity condition $C0$, does imply strong positivity condition $C1$. 

\begin{lemma} \label{lemma: subclus}
    Consider an asymptotic law $\mathbb{P}_0^F$ following conditions $D$ and sub-cluster positivity condition $C2$. Then $\tilde{q}_n$ is $L^2(\mathbb{P}_0)$ consistent for $\overline{q}_0$, i.e. $\mathbb{E}[\lVert \tilde{q}_n - \overline{q}_0\rVert_2^2] = o(1).$
\end{lemma}

\begin{corollary}
    Consider a finite cluster regime of interest $G_{n^*}$.  Under the conditions of Lemma \ref{lemma: subclus}, Weak finite cluster positivity condition $C0$ implies strong finite cluster positivity condition $C1$. 
\end{corollary}

\subsection{Estimation of regular large cluster parameters} \label{appsec: conditions}

In this appendix, we study estimators for the expected average potential outcome in a large-cluster under an arbitrary asymptotic regime $\mathbf{G}_0$, and also under the optimal such regime $\mathbf{G}_{0}^{\mathbf{opt}}$. We do so by deriving their efficient influence functions, and establish the regularity, asymptotic linearity (RAL), and efficiency of their corresponding estimators among this class (of RAL estimators) under a nonparametric model and appropriate regularity conditions. Our results build on those of   \citetSM{luedtke2016optimal} and related work \citepSM{van2014targeted, luedtke2016statistical, qiu2020optimal}. Interestingly, the conditions, estimators, and their resulting properties are largely isomorphic to those in this literature, despite the substantially different nature of the causal parameters we consider, the joint distributions in which they are defined, and the model restrictions assumed for identification. 

The crux of our arguments resolves around the following proposition: that our parameter of interest -- the expected \textit{average} outcome over \textit{all} individuals in a \textit{large-cluster} under a \textit{cluster-} dynamic regime -- is equivalent at all laws in our model, to the expected outcome for a \textit{single} \textit{solitary} individual under a particular  \textit{individualized} dynamic regime. Insodoing, we establish that the statistical (observed data) parameter of interest is identical to that studied by \citetSM{luedtke2016optimal}. As such, under isomorphic conditions to those in \citetSM{luedtke2016optimal}, we inherit the properties of these estimators while benefiting from the interpretive properties of our causal estimands and the realistic models under which inference is done.

In Section \ref{appsec: CDTRreexpress} we establish the linkages between our large-cluster parameters and previously studied individual level parameters; In Section \ref{appsec: OneStepConstruction}, we re-express the efficient influence functions (EIF) and resulting EIF-based estimators of \citetSM{luedtke2016optimal}; and in Section \ref{appsec: OneStepResults}, we adapt conditions of \citetSM{luedtke2016optimal} in terms of our notation and extend their results. 

\subsubsection{Re-expressing large-cluster CRs as stochastic IRs} \label{appsec: CDTRreexpress}

Note that $Y_i^{\mathbf{a}_n}$, $i=1,..n$ are marginally mutually independent under $\mathcal{M}^{AB}_n$, and that $Y_i^{\mathbf{a}_n} = Y_i^{a_i}$ for the $a_i \subset \mathbf{a}_n$ (See Proposition \ref{cor: ciidcount} in Appendix \ref{appsec: proofs}). By consequence, the expectation of a given individual's outcome under any CR in $\Pi^F(\kappa_n)$ can be characterized by the probability with which that individual is assigned treatment under that regime, conditional on their covariates, i.e. it is characterized by $q^*_i$. This is easily appreciated upon inspection of the individual-level g-formula function $f_i^{G_n}$, given in \eqref{eq: patlevgform}: by Theorem \ref{theorem: CDTRID}, $\mathbb{E}_{\mathbb{P}^F_n}[ h'(O_i^{G_n+})]$ is only a function of $i$ via $q^*_i$. Therefore, any two regimes $G_n$ and $G_n^{\dagger}$ that result in the same $q^*_i$ will be equivalent in the sense that $\mathbb{E}_{\mathbb{P}^F_n}[ h'(O_i^{G_n+})] = \mathbb{E}_{\mathbb{P}^F_n}[ h'(O_i^{G^{\dagger}_n+})]$. As a special case, we might consider a particular pair of regimes: 1) an arbitrary stochastic CR $G_n$; and 2) a stochastic IR $g_i$ in which treatment is assigned as a random draw from the treatment assignment density $q^*_n$ under $G_n$.

\begin{proposition} \label{lemma: CDTRID}
    Consider an IR $g_i$ that only intervenes on treatment for individual $i$, assigning $A^{g_i+}_i$ as a random draw from $q^{*}_{i}$, the conditional intervention density for an arbitrary $G_n\in\Pi_n(\kappa_n)$. Under $\mathcal{M}^{AB}_n$,
\begin{align}
     \mathbb{E}_{\mathbb{P}^F_n}[ h'(O_i^{g_i})] = & \sum\limits_{\mathcal{O}_i}h'(o)f_{i}^{G_n}(o). \label{eq: igformstoch}
\end{align}
\end{proposition}

\begin{proof}
    The result follows by showing that $\mathbb{P}_n^F(Y^{g_i}_i=y \mid A^{g_i+}_i=a, L_i=l) = Q_Y(y \mid a, l)$ via identical arguments to the proof of Lemma \ref{eq: IDkernel} and applying Proposition \ref{theorem: YbarID} in Appendix \ref{appsec: proofs}.
\end{proof}

The following corollary extends Proposition \ref{lemma: CDTRID} to large-cluster regimes.

\begin{corollary}
 Consider an asymptotic regime $\mathbf{G}_0$ following  conditions $E$ under a law $\mathbb{P}_0^F$ following conditions $D$, and a law $\mathbb{P}_n^F\in\mathbb{P}_0^F$. Furthermore, consider an IR $g^{*}_i$ that only intervenes on treatment for individual $i$, assigning $A^{g^{*}_i+}_i$ as a random draw from $q^{*}_{0}$, the conditional intervention density for the asymptotic regime $\mathbf{G}_0$  under $\mathbb{P}_0^F$.  Then the following identities hold, for all $i$ and all $n$:

    $$\mathbb{E}_{\mathbb{P}_n^F}[Y_i^{g^{*}_i}] = \mathbb{E}_{\mathbb{P}_0^F}[Y_i^{g^{*}_i}]= \mathbb{E}_{\mathbb{P}_0^F}[\overline{Y}_0^{\mathbf{G}_0}].$$

    Let $g^{\mathbf{opt}}_i$ be such a regime assigning $A^{g^{\mathbf{opt}}_i+}_i$ as a random draw from $q^{\mathbf{opt}}_{0}$, the conditional intervention density under the optimal asymptotic regime $\mathbf{G}^{\mathbf{opt}}_{0}$. Likewise, the following identities hold, for all $i$ and all $n$:

    $$\mathbb{E}_{\mathbb{P}_n^F}[Y_i^{g^{\mathbf{opt}}_i}] = \mathbb{E}_{\mathbb{P}_0^F}[Y_i^{g^{\mathbf{opt}}_i}]= \mathbb{E}_{\mathbb{P}_0^F}[\overline{Y}_0^{\mathbf{G}^{\mathbf{opt}}_{0}}].$$
\end{corollary}

In the following derivations, it will be useful to express the single parameter $\mathbb{E}_{\mathbb{P}_n^F}[Y_i^{g^{\mathbf{opt}}_i}]$ -- the value under a stochastic regime dependent on $\delta$ and $L_i$ -- as a particular average over a set of parameters that each correspond to the value under a \textit{deterministic} regime dependent on $L_i$ alone. To that end, note that any stochastic regime $g^{*}_i$ corresponding to an asymptotic regime $\mathbf{G}_0$ can be formulated as a deterministic function with respect to $L_i$ and a randomizer term $\delta \sim \text{Unif}(0,1)$, specifically:

\begin{align}
    g^{*}_i : \{\delta, l\} \mapsto I\Big(\delta < q^*_0(1 \mid l)\Big).
\end{align}

Thus conditional on a specific value of $\delta$, $g^{*}_i$ is equivalent to a specific deterministic regime that is a function of $L_i$ alone,  $g^{*}_{i, \delta}: l \mapsto I\Big(\delta < q^*_0(1 \mid l)\Big)$. Let $P^F \equiv \mathbb{P}^F_1\in\mathbb{P}_0^F$ following conditions $D$, and let $P\equiv P(P^F)$. Then, we have the following equivalence:

\begin{align}
    \mathbb{E}_{P^F}[Y_i^{g^{*}_i}] = \mathbb{E}_{\delta}\Big[\mathbb{E}_{P^F}[Y_i^{g^{*}_{i, \delta}}]\Big].
\end{align}


\subsubsection{Constructing a one-step estimator for the value of $G^{\mathbf{opt}}_{n}$} \label{appsec: OneStepConstruction}

In the preceding Section, we've argued that when we are interested in the value of an asymptotic regime $\mathbf{G}_0$ for a large cluster following an asymptotic law $\mathbb{P}_0^F$ (under conditions $D$ and $E$), we can do so by instead targeting the value under an appropriately-equivalent class of individual-level regimes $g_{i, \delta}$, deterministic with respect to $L_i$, applied to a single individual in a cluster entirely of their own, $\mathbb{E}_{P^F}[Y_i^{g^{*}_{i, \delta}}]$. Thus, such parameters are identical to those classically studied by the causal inference literature, under iid settings, where investigators target the value under some deterministic regime $g$, i.e., $\mathbb{E}_{P^F}[Y^{g}]$, where $Y\equiv Y_i$ when the observed data $\mathbf{O}_n\equiv \{O_1,\dots, O_n\}$ are viewed as $n$ independent draws from the common law $P$.

We leverage this formulation to motivate articulation of the efficient influence function for $\mathbb{E}_{\mathbb{P}_0^F}[\overline{Y}_0^{\mathbf{G}_{0}}]$, via the efficient influence function of $\mathbb{E}_{P^F}[Y_i^{g^{*}_{i, \delta}}]$. 

Let $\Psi^{\mathbf{G}_0}\coloneqq \sum\limits_{o\in\mathcal{O}_i}yf^{\mathbf{G}_0}(o)$ denote the indentifying functional for the expected average outcome under the asymptotic regime $\mathbf{G}_0$, and let $\Psi^{\mathbf{opt}}\coloneqq \sum\limits_{o\in\mathcal{O}_i}yf^{\mathbf{G}_{0}^{\mathbf{opt}}}(o)$ denote the identifying functional for that parameter under the asymptotic optimal regime $\mathbf{G}_{0}^{\mathbf{opt}}$. Note that $\Psi^{\mathbf{G}_0}$ is simply the well-studied individual-level g-formula with respect to $P$, an arbitrary margin of the more elaborated law $\mathbb{P}_n$. 

Let $\overline{P}$ be the observed margins of a distribution in $\mathcal{M}_1$, which factorizes according to $Q_Y$, $\overline{q}_0$, and $Q_L$, with respect to the asymptotic law $\mathbb{P}_0$. We consider this law $\overline{P}$ because, under condition set $E$ and the strong cluster positivity conditions $C1$, it represents the limiting distribution of the empirical law $\tilde{P}_n \equiv \{\tilde{Q}_{Y,n}, \tilde{q}_n, \tilde{Q}_{L,n}\}$. Then we define the following two random variables $\Phi_1(\overline{P}, g)(O_i)$ and $\Phi_2(g, \eta)(O_i)$, with respect to a deterministic regime $g$: 

\begin{align}
    \Phi_1(\overline{P}, g)(o) \coloneqq & 
        \frac{I\big(a=g(l)\big)}{\overline{q}_0(a\mid l)}\Big(y-\mathbb{E}_{\overline{P}}[Y_i \mid a, l]\Big) \nonumber \\ & + \mathbb{E}_{\overline{P}}[Y_i \mid g(l), l] - \mathbb{E}_{\overline{P}}\big[\mathbb{E}_{\overline{P}}[Y_i \mid g(L_i), L_i]\big]. \\
    \Phi_2(g, \eta)(o) \coloneqq & 
        \eta(g(l) - \kappa^*).
\end{align}

$ \Phi_1(\overline{P}, g)$ is recognizable as the EIF for the g-formula parameter under a fixed deterministic regime $g$ and $\Phi_2(g, \eta)$ has been shown by \citetSM{luedtke2016optimal} to be an additional additive term for the EIF when the regime $g$ is the $\kappa^*$ resource-constrained optimal regime that assigns treatment to all individuals with a CATE greater than $\eta$.

Let $g^*$ be an arbitrary stochastic regime that depends on $L_i$ and $\delta$, and $g^*_{\delta}$ a corresponding deterministic regime with $\delta$ fixed. We further define $\Phi_0(\overline{P}, g^*, \eta)(O_i)$ to be the expectation of the sum of $\Phi_1(\overline{P}, g^*_{\delta})(O_i)$ and $\Phi_2(g^*_{\delta}, \eta)(O_i)$, taken with respect to the distribution of $\delta$, and we let $\sigma^{2}(\overline{P}, g^*, \eta)$ denote the variance of $\Phi_0(\overline{P}, g^*, \eta)(O_i)$ under $\overline{P}$:

\begin{align}
    \Phi_0(\overline{P}, g^*, \eta)(o) & \coloneqq \mathbb{E}_{\delta}\Big[ \Phi_1(\overline{P}, g^*_{\delta})(o) + \Phi_2(g^*_{\delta}, \eta)(o)\Big].\\
    \sigma^{2}(\overline{P}, g^*, \eta) & \coloneqq \text{Var}_{\overline{P}}\Big(\Phi_0(\overline{P}, g^*, \eta)(O_i)\Big) \label{eq: IF}.
\end{align}

\citetSM{luedtke2016optimal} showed that, under conditions, $\Phi_0(\overline{P}, g^{*}, \eta)(O_i)$ represents the EIF when the regime $g^{*}$ is the $\kappa^*$ resource-constrained optimal regime that assigns treatment to all individuals with a CATE greater than $\eta$, and randomly assigns treatment to individuals with CATE equal to $\eta$ with precisely the probability for which the marginal probability of assigned treatment is exactly $\kappa^*$. 

In a slight abuse of notation, let $\tilde{Q}_{Y, n}$,  $\tilde{q}_n$, $\tilde{Q}_{L, n}$ denote functionals of $\mathbb{O}_n$ designed to approximate $Q_Y$,  $\overline{q}_0$, $Q_L$, respectively, possibly via flexible regression or machine learning procedures. 
Then define $\tilde{g}^{*}_i$ and $\tilde{g}^{\mathbf{opt}}_i$ to be the appropriately equivalent individual regimes, learned from the data $\tilde{P}_n$, corresponding to $\mathbf{G}_0$ and $\mathbf{G}_{0}^{\mathbf{opt}}$, respectively:

\begin{align}
    \tilde{g}^{*}_i & : \delta \times l \mapsto I\Big(\delta < \tilde{q}^*_0(1 \mid l)\Big),\\
    \tilde{g}^{\mathbf{opt}}_i & : \delta \times l \mapsto I\Big(\delta < \tilde{q}^{\mathbf{opt}}_0(1 \mid l)\Big).
\end{align}

With these objects defined, let $\tilde{\Psi}^{\mathbf{G}_0}_{n}$ and $\tilde{\Psi}^{\mathbf{opt}}_{n}$ denote one-step estimators for $\Psi^{\mathbf{G}_0}$ and $\Psi^{\mathbf{opt}}$, and let $\tilde{\sigma}^{2, \mathbf{G}_0}_n$ and $\tilde{\sigma}^{2, \mathbf{opt}}_n$ denote estimators for their asymptotic variances, respectively:

\begin{align}
 \tilde{\Psi}^{\mathbf{G}_0}_{n} & \coloneqq    
    \sum\limits_{o\in\mathcal{O}_i}y\tilde{f}^{\mathbf{G}_0}(o) +  \frac{1}{n}\sum\limits_{i=1}^n\Phi_0(\tilde{P}_n, \tilde{g}^{*}_i, 0)(O_i),  \\  
    \tilde{\Psi}^{\mathbf{opt}}_{n} & \coloneqq    
    \sum\limits_{o\in\mathcal{O}_i}y\tilde{f}^{\mathbf{G}^{\mathbf{opt}}_{0}}(o) +  \frac{1}{n}\sum\limits_{i=1}^n\Phi_0(\tilde{P}_n, \tilde{g}^{\mathbf{opt}}_i, \tilde{\eta}_0)(O_i),\\
    \tilde{\sigma}^{2, \mathbf{G}_0}_n & \coloneqq \frac{1}{n}\sum\limits_{i=1}^n\Phi_0(\tilde{P}_n, \tilde{g}^{*}_i, 0)(O_i)^2, \\
    \tilde{\sigma}^{2, \mathbf{opt}}_n & \coloneqq \frac{1}{n}\sum\limits_{i=1}^n\Phi_0(\tilde{P}_n, \tilde{g}^{\mathbf{opt}}_i, \tilde{\eta}_0)(O_i)^2.
\end{align}

We emphasize that the argument for $\eta$ with respect to the function $\Phi_0$ in the definition of $\tilde{\Psi}^{\mathbf{G}_0}_{n}$ is 0, indicating that the empirical analogue of the term deriving from $\Phi_2$ (the additional term in the influence function for the resource-constrained optimal regime) is omitted from the estimator for the value under a fixed asymptotic regime $\mathbf{G}_0$.

\subsubsection{Results for a one-step estimator} \label{appsec: OneStepResults}

In the following, we re-state regularity conditions from \citetSM{luedtke2016optimal}. To do so, we define for convenience some additional objects. First, let $S_{\Lambda, P}$ be the survival function of $\Lambda(L_i)$, i.e., $S_{\Lambda, P}: c \mapsto P(\Lambda(L_i)<c)$. Furthermore define the following second-order remainder terms:

\begin{align}
    R_{1,0}(g^*, \tilde{P}_n) & = \mathbb{E}_{\delta}\Bigg[\mathbb{E}_{P}\Big[\Big(1-\frac{\overline{q}_0(g^*_{\delta}(L_i) \mid L_i)}{\tilde{q}_n(g^*_{\delta}(L_i) \mid L_i)}\Big)\Big(\mathbb{E}_{\tilde{P}_n}[Y_i\mid g^*_{\delta}(L_i), L_i] - \mathbb{E}_{P}[Y_i\mid g^*_{\delta}(L_i), L_i]\Big)\Big]\Bigg],\\
    R_{2,0}(g^*, g^{\dagger}, \Lambda, \omega) & = \mathbb{E}_{\delta}\Bigg[\mathbb{E}_{P}\Big[\big\{g^*_{\delta}(L_i) - g^{\dagger}_{\delta}(L_i)\big\}\big\{\Lambda(L_i) - \omega\big\}\Big]\Bigg].  
\end{align}

The following conditions are articulated with respect to a general large-cluster parameter under asymptotic regime $\mathbf{G}_0$, characterized by a rank-function $\Lambda$, with equivalent IR $g^*_i$ under $P\in\mathbb{P}_0.$
\begin{itemize}
    \item[F1.] $\Lambda(L_i)$ has continuous density $f_0$ at $\omega_0$ under $P$ and $0< f_0(\omega_0)<\infty$.
    \item[F2.]  $S_{\Lambda, P}$ is continuous in a neighborhood of $\omega_0$.
    \item[F3.] $P\big(\Lambda(L_i) = \omega\big)=0$ for all $\omega$ in a neighborhood of $\omega_0$.
    \item[F4.] $P\big(\tau < \overline{q}_0(1 \mid L_i) <1-\tau\big)=1$ for some $\tau>0$.
    \item[F5.] $P\big(\tau < \tilde{q}_n(1 \mid L_i) <1-\tau\big)=1$ with probability approaching 1 for some $\tau>0$.
    \item[F6.] $R_{1,0}(\tilde{g}^{*}_{i}, \tilde{P}_n) = o_P(n^{1/2})$.
    \item[F7.] $\mathbb{E}_{\overline{P}}\Big[\Big(\Phi_0(\tilde{P}_n, \tilde{g}^{*}_i, 0)(O_i) - \Phi_0(P, g^{*}_i,0)(O_i)\Big)^2\Big]= o_P(1)$, and $\mathbb{E}_{\Tilde{P}_n}[Y_i \mid A_i = a, L_i=l]$ belongs to a $P$-Donsker class and is an $L^2(P)$-consistent estimator of $\mathbb{E}_{P}[Y_i \mid A_i = a, L_i=l]$.
    \item[F8.] $\Phi_0(\tilde{P}_n, \tilde{g}^{*}_{i}, 0)$ belongs to a $P-$Donsker class $\mathcal{D}$ with probability approaching 1.
    \item[F9.] $\frac{1}{n}\sum\limits_{i=1}^n\Phi_0(\tilde{P}_n, \tilde{g}^{*}_i, 0)(O_i)= o_P(n^{1/2})$.
\end{itemize}

Let Condition $F$ denote the intersection of conditions $F1$-$F9$. For inference on the optimal asymptotic regime $\mathbf{G}^{\mathbf{opt}}_{0}$, consider the following modified conditions:

\begin{itemize}
    \item[$F1^*.$] $\Delta_0(L_i)$ has continuous density $f^*_0$ at $\eta_0$ under $P$ and $0< f^*_0(\eta_0)<\infty$.
    \item[$F2^*.$]  $S_{\Delta_0, P}$ is continuous in a neighborhood of $\eta_0$.
    \item[$F3^*.$] $P\big(\Delta_0(L_i) = \eta\big)=0$ for all $\eta$ in a neighborhood of $\eta_0$.
    \item[$F4^*.$] $P\big(\tau^* < \overline{q}_0(1 \mid L_i) <1-\tau^*\big)=1$ for some $\tau^*>0$.
    \item[$F5^*.$] $P\big(\tau^* < \tilde{q}_n(1 \mid L_i) <1-\tau^*\big)=1$ with probability approaching 1 for some $\tau^*>0$.
    \item[$F6^*.$] $R_{1,0}(\tilde{g}^{\mathbf{opt}}_{i}, \tilde{P}_n) = o_P(n^{1/2})$ and  $R_{2,0}(\tilde{g}^{\mathbf{opt}}_{i}, {g}^{\mathbf{opt}}_{i}, \Delta_0, \eta_0) = o_P(n^{1/2})$.
    \item[$F7^*.$] $\mathbb{E}_{\overline{P}}\Big[\Big(\Phi_0(\tilde{P}_n, \tilde{g}^{\mathbf{opt}}_i, {\eta}_0)(O_i) - \Phi_0(\overline{P}, g^{\mathbf{opt}}_i,\eta_0)(O_i)\Big)^2\Big]= o_P(1)$, and $\mathbb{E}_{\Tilde{P}_n}[Y_i \mid A_i = a, L_i=l]$ belongs to a $P$-Donsker class and is an $L^2(P)$-consistent estimator of $\mathbb{E}_{P}[Y_i \mid A_i = a, L_i=l]$.
    \item[$F8^*.$] $\Phi_0(\tilde{P}_n, \tilde{g}^{\mathbf{opt}}_{i}, \eta_0)$ belongs to a $P-$Donsker class $\mathcal{D}$ with probability approaching 1.
    \item[$F9^*.$] $\frac{1}{n}\sum\limits_{i=1}^n\Phi_0(\tilde{P}_n, \tilde{g}^{\mathbf{opt}}_i, \eta_0)(O_i)= o_P(n^{1/2})$.
    \item[$F10^*$] $\tilde{\Delta}_n$ is consistent for $\Delta_0$ in $L^1(P)$ and  $\tilde{\Delta}_n$ belongs to a Glivenko-Cantelli class with probability approaching 1.
\end{itemize}

Let Condition $F^*$ denote the intersection of these adapted conditions. The following Theorem consolidates results from \citetSM{luedtke2016optimal}.

\begin{theorem}\label{thm: asympeffGopt}
    Consider an asymptotic law $\mathbb{P}_0$ following condition $D$ and an asymptotic optimal regime $\mathbf{G}_{0}^{\mathbf{opt}}$ following condition $E$. Suppose further that condition $F^*$ holds. Then the following properties also hold:

    \begin{itemize}
        \item[(1.)] $\Psi^{\mathbf{opt}}$ is a pathwise differentiable parameter with efficient non-parametric influence function $\Phi_0(\overline{P}_0, g^{\mathbf{opt}}_{i}, \eta_0)$.
        \item[(2.)] $\tilde{\Psi}^{\mathbf{opt}}_{n}$ is the asymptotically-efficient, regular and asymptotically linear estimator of $\Psi^{\mathbf{opt}}$.
        \item[(3.)] An asymptotically valid two-sided $1-\alpha$ confidence interval is given by $$\tilde{\Psi}^{\mathbf{opt}}_{n} \pm z_{1-\alpha/2}\frac{\tilde{\sigma}^{2, \mathbf{opt}}_n}{\sqrt{n}}.$$
    \end{itemize}
\end{theorem}

The subsequent proposition extends these results to estimators for $\Psi^{\mathbf{G}_0}$.

\begin{proposition}
    
 \label{prop: asympeffG0}
    Consider an asymptotic law $\mathbb{P}_0$ following condition $D$ and an asymptotic  regime $\mathbf{G}_{0}$ following condition $E$. Suppose further that condition $F$ holds. Then the following properties also hold:

    \begin{itemize}
        \item[(1.)] $\Psi^{\mathbf{G}_0}$ is a pathwise differentiable parameter with efficient non-parametric influence function $\Phi_0(\overline{P}_0, g^{*}_{i}, 0)$.
        \item[(2.)] $\tilde{\Psi}^{\mathbf{G}_0}_{n}$ is the asymptotically-efficient RAL estimator of $\Psi^{\mathbf{G}_0}$.
        \item[(3.)] An asymptotically valid two-sided $1-\alpha$ confidence interval is given by $$\tilde{\Psi}^{\mathbf{G}_0}_{n} \pm z_{1-\alpha/2}\frac{\tilde{\sigma}^{2, \mathbf{G}_0}_n}{\sqrt{n}}.$$
    \end{itemize}
\end{proposition}

The following results establish consistency of $\tilde{\omega}_n$, thereby providing a sufficient condition for the regularity of $\Psi^{\mathbf{G}_0}$.

\begin{lemma} \label{lemma: omegacons}
    Consider an asymptotic law $\mathbb{P}_0$ following condition $D$ and an asymptotic regime $\mathbf{G}_{0}$ following condition $E$. Suppose further that Conditions $F2$ and $F3$ hold. Then, $\tilde{\omega}_n$ converges to $\omega_0$ in probability.
\end{lemma}
Lemma \ref{lemma: omegacons} can be proved in the same manner as Lemma 5 in \citetSM{luedtke2016optimal}.

Now, consider a large-cluster positivity condition analogous to strong finite-cluster positivity condition $C1$.

\begin{itemize}
    \item [C1$^*$.] \textbf{Strong large-cluster positivity} Consider an asymptotic regime $\mathbf{G}_0$ following condition $E$.  $\mathbb{P}_0^F$ is such that $\overline{q}_0$ is well-defined and, for all $a, l$, if ${q}^*_{0}(a \mid l)Q_L(l) > 0$ then $\overline{q}_0(a \mid l)>0.$
\end{itemize}

The following corollary extends the results of Theorem \ref{theorem: asymptcons} to large cluster parameters:

\begin{corollary} \label{cor: asymconslarge}
   Consider an asymptotic law $\mathbb{P}_0$ following conditions $D$ and an asymptotic regime $\mathbf{G}_{0}$ following conditions $E$. Suppose that $\tilde{\omega}_n$ is consistent for $\omega_0$, possibly via the conditions of Lemma \ref{lemma: omegacons}, and that the strong large-cluster positivity condition $C1^*$ holds (possibly via the conjunction of sub-cluster positivity condition $C2$ and weak large-cluster positivity condition $C0^*)$. Then we have that $\sum\limits_{o\in\mathcal{O}_i}y\tilde{f}^{\mathbf{G}_0}(o)$ is an asymptotically consistent estimator for $\Psi^{\mathbf{G}_0}$. Under sub-cluster positivity condition $C2$ we additional have that $\sum\limits_{o\in\mathcal{O}_i}y\tilde{f}^{\mathbf{G}_0}(o)$ is a RAL estimator for $\Psi^{\mathbf{G}_0}$.    
\end{corollary}

\subsection{Estimation of irregular large-cluster parameters} \label{appsec: online}

Following \citetSM{luedtke2016optimal}, let $\{l_n\}_{n \in \mathbb{N}}$ be the sequence of sizes of the smallest subsample of $\mathbf{O}_n$ on which the optimal rule is learned. For $j = l_n, \ldots,n$, let $\Tilde{P}_{n,j}$ be the empirical distribution of $\{O_1, \ldots, O_j\}$, and let $\Tilde{q}_{n,j}$, $\Tilde{g}_{i, j}^*$ and $\Tilde{\eta}_{n,j}$ be the estimates of $\Bar{q}_0$, $g_i^*$, and ${\eta}_0$ based on $\{O_1, \ldots, O_{j-1}\}$.

Conditions $F^*$ are reasonable if $L_i$ is continuous, but if the covariates are discrete then will not hold be verified: they would contradict conditions $F1^*$ and $F2^*$. The following results will therefore consider a set of conditions that exclude $F1^*$ and $F2^*$. In particular, we let $F^{\dagger}$ denote the intersection of $F4^*$, $F7^*$, and the following boundedness condition for the outcomes $Y_i$:

\begin{itemize}
    \item[$F3^{\dagger}$.] $\overline{P}(|Y_i| < M) = 1$ for some $M < \infty$.
\end{itemize}

\citetSM{luedtke2016statistical} consider a slightly weaker version of condition $F7^*$, that we do not restate here to focus on relations between conditions for the estimators of this section compared to estimators considered in Appendix \ref{appsec: OneStepResults}.
 
Then, \citetSM{luedtke2016statistical} develop an online estimator. Let $\ell_n$ be the size of the smallest sample on which we train our models and learn the optimal rule.

For $j = 1, \ldots, n$, let $(\Tilde{q}_{n,j}, \Tilde{Q}_{n,j}, \Tilde{g}_{i,j})$ be estimators of $(\Bar{q}_0, \mathbb{E}_{\overline{P}}, g^{\textbf{opt}})$ respectively, based on $(O_1, \ldots, O_{j-1})$.

Let 
\begin{equation*}
    \Tilde{\sigma}^2_{0,n,j} := \text{Var}_{\overline{P}}(\Phi_0(\Tilde{P}_{n,j},\Tilde{g}_{i,j}, \Tilde{\eta}_{n,j}) \mid O_1, \ldots, O_{j-1}),
\end{equation*}
and $\Tilde{\sigma}^2_{n,j}$ be an estimate of $\Tilde{\sigma}^2_{0,n,j}$ based on $(O_1, \ldots, O_{j-1})$.

Furthermore, let 
\begin{equation*}
    \Gamma_n := \frac{1}{n - \ell_n} \sum\limits_{j = \ell_n + 1}^n \Tilde{\sigma}_{n,j}^{-1}.
\end{equation*}

Let $\Tilde{\Phi}_{n,j} := \Phi_0(\Tilde{P}_{n,j},\Tilde{g}_{i,j}, \Tilde{\eta}_{n,j})$. Then, \citetSM{luedtke2016statistical} define the following estimator.
\begin{equation*}
    \hat{\Psi}_n^{\textbf{opt}} := \Gamma_n^{-1} \frac{1}{n - \ell_n} \sum\limits_{j = \ell_n + 1}^n \Tilde{\sigma}_{n,j}^{-1} \Tilde{\Phi}_{n,j} (O_j).
\end{equation*}

We further consider the following remainder terms:

    \begin{align*}
        &R_{1n}:= \frac{1}{n - \ell_n}
        \sum\limits_{j = \ell_n + 1}^n \frac{ \mathbb{E}_{\overline{P}} \left[\left(1 - \frac{\Bar{q}_0(\Tilde{g}_{i,j}(L_i) \mid L_i)}{\Tilde{q}_{n,j}(\Tilde{g}_{i,j}(L_i) \mid L_i)} \right)(\Tilde{Q}_{n,j}(\Tilde{g}_{i,j}(L_i), L_i) - \mathbb{E}_{\overline{P}}[Y_i \mid \Tilde{g}_{i,j}(L_i), L_i]) \right]}{\Tilde{\sigma}_{n,j}},\\
       & R_{2n} := \frac{1}{n - \ell_n} \sum\limits_{j = \ell_n + 1}^n \frac{ \mathbb{E}_{\Bar{P}}[\mathbb{E}_{\overline{P}}[Y_i \mid \Tilde{g}_{i,j}(L_i), L_i]] - \Psi^{\textbf{opt}}}{\Tilde{\sigma}_{n,j}}.
    \end{align*}

The following assumptions from \citetSM{luedtke2016statistical} are required for $\hat{\Psi}_n^{\textbf{opt}}$ to be asymptotically normal.
\begin{itemize}
    \item[G1.] $n - \ell_n \to \infty$ as $n \to \infty$.
    \item[G2.](Lindeberg-like condition) $\forall \epsilon > 0$,
    \begin{align*}
        &\frac{1}{n - \ell_n} \sum\limits_{j = \ell_n + 1}^n \mathbb{E}_{\overline{P}}\left[\left(\frac{\Tilde{\Phi}_{n,j}(O_j)}{\Tilde{\sigma}_{n,j}} \right)^2 I\left[\frac{|\Tilde{\Phi}_{n,j}(O_j)|}{\Tilde{\sigma}_{n,j}} > \epsilon \sqrt{n - \ell_n} \right] \Bigg| O_1, \ldots, O_{j-1} \right]= o_{\overline{P}}(1).
    \end{align*}
    \item[G3.] $\frac{1}{n - \ell_n} \sum\limits_{j = \ell_n + 1}^n \frac{\Tilde{\sigma}^2_{0,n,j}}{\Tilde{\sigma}^2_{n,j}} \to 1$ in probability.
    \item[G4.] $R_{1n} = o_{\overline{P}}(n^{-1/2})$.

    \item[G5.] $R_{2n}  = o_{\overline{P}}(n^{-1/2})$.
\end{itemize}
The following theorem is an immediate consequence of Theorem 2 in \citetSM{luedtke2016statistical}.
\begin{theorem}[\citepSM{luedtke2016statistical}] \label{thm: onlineNormality}
    Under Conditions $F^{\dagger}$ and G,  
    \begin{equation*}
        \Gamma_n \sqrt{n - \ell_n}(\hat{\Psi}_n^{\textbf{\textup{opt}}} - \Psi^{\textbf{\textup{opt}}}) \overset{d}{\to} \mathcal{N}(0,1),
    \end{equation*}
    as $n \to \infty$.
\end{theorem}
Theorem \ref{thm: onlineNormality} implies the following interval is a valid $1-\alpha$ CI for $\Psi(P_0)$,
\begin{equation*}
    \hat{\Psi}_n^{\textbf{opt}} \pm z_{1 - \alpha/2} \frac{\Gamma_n^{-1}}{\sqrt{n - \ell_n}},
\end{equation*}
where $z_{1-\alpha/2}$ denotes the $1-\alpha/2$ quantile of a standard normal random variable. 

As the asymptotic normality of $\hat{\Psi}_n^{\textbf{opt}}$ is proved in \citetSM{luedtke2016statistical} using a martingale central limit theorem, the proof of Theorem 2 in \citetSM{luedtke2016statistical} is also valid in our interference setting, as it does not rely on the i.i.d. assumption.

\clearpage
\section{Proofs of main results} \label{appsec: proofs}

Here, we re-state and prove formal results, which we organize in two main sections. Appendix  \ref{appsec: proofsID} provides proofs for identification-related results presented inSections \ref{sec: ID} and \ref{sec: largecluster} of the main text, and Appendices \ref{appsec: gencdtrs} and \ref{appsec: misc}. Appendix  \ref{appsec: proofsEst} provides proofs for estimation-related results in the main text Section \ref{sec: estim}, and Appendix \ref{appsec: estim}.

\subsection{Identification results} \label{appsec: proofsID}

\subsubsection{Finite-cluster parameters}

We begin by establishing the following lemma.

\begin{lemma}
 \label{lemma: ciid}
    Under $\mathcal{M}_n^{AB}$, the following conditions hold for all $l, a$, and for all $i \neq j$:
\begin{align}
    & \{L_i\in \mathbf{L}_n\}  \text{ are mutually independent.} \label{eq: condid1}\\
   & \{Y_i \in \mathbf{Y}_n \mid L_i = l, A_i=a\}  \text{ are mutually independent.} \label{eq: condid2} \\
   & \mathbb{P}_n(L_i =l) =  \mathbb{P}_n(L_j = l). \label{eq: ID1}\\
    & \mathbb{P}_n(Y_i=y \mid L_i=l, A_i=a) =  \mathbb{P}_n(Y_j=y \mid L_j=l, A_j=a) \label{eq: ID2}.
\end{align}
\end{lemma}

\begin{proof}

We prove the proposition in \eqref{eq: condid1} by showing that $\mathbb{P}_n(\mathbf{L}_n=\mathbf{l}_n) = \prod\limits_{i=1}^n\mathbb{P}_n(L_i=l_i)$.
\begin{align*}
    \mathbb{P}_n(\mathbf{L}_n=\mathbf{l}_n) = & \prod\limits_{i=1}^n\mathbb{P}_n(L_i=l_i\mid L_1=l_1,\dots, L_{i-1}=l_{i-1}) \\
    = & \prod\limits_{i=1}^n\mathbb{P}_n(f_{L_i}(\epsilon_{L_i})=l_i\mid f_{L_1}(\epsilon_{L_1})=l_1,\dots, f_{L_{i-1}}(\epsilon_{L_{i-1}})=l_{i-1}) \\
    = & \prod\limits_{i=1}^n\mathbb{P}_n(f_{L_i}(\epsilon_{L_i})=l_i) \\
    = & \prod\limits_{i=1}^n\mathbb{P}_n(L_i=l_i).
\end{align*}

The first equality follows by laws of probability. The second equality follows by definition of the structural equation model and the restrictions on the parents of $L_i$ in Condition A1 (Conditional noninterference). The third equality follows by the mutual independence of the error terms for each $L_i$. The final equality is by the definition of the structural equation model. We prove the proposition in \eqref{eq: condid2} by showing that $\mathbb{P}_n(\mathbf{Y}_n=\mathbf{y}_n \mid \mathbf{B}_n=\mathbf{b}_n) = \prod\limits_{i=1}^n\mathbb{P}_n(Y_i=y_i \mid \mathbf{B}_n=\mathbf{b}_n)$, where we let $\mathbf{V}_{-i}\coloneqq \mathbf{V}_{n}\setminus V_i$ and let $\overline{L}_i \coloneqq \{L_1,\dots, L_i\}$.

\begin{align*}
   & \mathbb{P}_n(\mathbf{Y}_n=\mathbf{y}_n \mid \mathbf{B}_n=\mathbf{b}_n) \\ 
   = & \prod\limits_{i=1}^n\mathbb{P}_n(Y_i=y_i \mid B_i=b_i, \mathbf{B}_{-i}=\mathbf{b}_{-i}, \overline{Y}_{i-1}=\overline{y}_{i-1}) \\
    = & \prod\limits_{i=1}^n\mathbb{P}_n(f'_{Y_i}(\epsilon_{Y_i})=y_i \mid B_i=b_i, f'_{\mathbf{B}_{-i}}(\epsilon_{\mathbf{A}_{-i}},\epsilon_{\mathbf{L}_{-i}})=\mathbf{b}_{-i}, \overline{f}'_{Y_{i-1}}(\overline{\epsilon}_{Y_{i-1}})=\overline{y}_{i-1}). \\
    = & \prod\limits_{i=1}^n\mathbb{P}_n(f'_{Y_i}(\epsilon_{Y_i})=y_i \mid B_i=b_i) \\
    = & \prod\limits_{i=1}^n\mathbb{P}_n(Y_i=y_i \mid \mathbf{B}_n=\mathbf{b}_n).
\end{align*}

The first equality follows by laws of probability. The second equality follows by definition of the structural equation model and restrictions on the parent sets of $Y_i$ in Condition $A1$, where we have let  $f'_{Y_i}(\cdot) \equiv f_{Y_i}(\cdot, B_i=b_i)$ denote the structural equation for $Y_i$ with $B_i$ fixed to $b_i$,  $f'_{\mathbf{B}_{-i}}(\cdot)$ denote the vector valued structural equation for $\mathbf{B}_{-i}$ with $B_i$ fixed to $b_i$, and $\overline{f}'_{\mathbf{Y}_{i-1}}(\cdot)$ denote the vector valued structural equation for $\overline{Y}_{i-1}$ conditional on $\mathbf{B}_{-i} = \mathbf{b}_{-i}$. The third equality follows by the error independence structure for $\epsilon_{Y_i}$ in Condition A1 (Conditional noninterference). The final equality is again by the definition of the structural equation model, and by reversing the step between the second and third equality.

We prove the propositions in \eqref{eq: ID1} and \eqref{eq: ID2}  by showing that the elements of $\mathbf{L}_n$ are identically distributed variables and likewise that $\{Y_i : B_i=b\}$ are identically distributed. First, by condition $A3$ (conditional identically distributed errors) we have the error terms for each $L_i$ have some common distribution, say that of $V_L$.  Then we have that each $L_i = f_{L_i}(\epsilon_{L_i}) = f_{L}(\epsilon_{L_i}) \sim f_{L}(V_L)$, that is, each $L_i$ is some fixed function $f_{L}$ of some variables with a fixed distribution, where $f_{L_i}=f_L$ follows by condition $A2$ (structural invariance). Likewise, let $f'_{Y_i}$ be defined as above. By condition $A3$ and $A1$ we have the error terms for each $Y_i$ have some common distribution conditional on $B_i=b_i$, say that of $V_Y$. Then we have that, given $B_i=b_i$, each $Y_i = f'_{Y_i}(\epsilon_{Y_i}) = f'_{Y}(\epsilon_{Y_i}) \sim f'_{Y}(V_Y)$, that is, each $Y_i$ is some fixed function $f'_{Y}$ of some variables with a fixed distribution, where $f'_{Y_i}=f_Y$ follows by condition $A2$ (structural invariance). Thus, $\mathbb{P}_n(L_i=l)$ and $\mathbb{P}_n(Y_i=y \mid B_i=b)$ are thus described by common measures $Q_L$ and $Q_Y$, respectively.
\end{proof}

\begin{corollary} \label{cor: ciidcount}
    Under $\mathcal{M}_n^{AB}$, the following conditions hold for all $l, a$:
\begin{align}
   & \{Y_i^{G_n} \in \mathbf{Y}_n ^{G_n} \mid L_i = l, A^{G_n+}_i=a\}  \text{ are mutually independent.} \label{eq: condid2count} \\
   & \mathbb{P}_n^F(Y_i^{\mathbf{a}_n} =  Y_i^{a_i}) = 1 \text{ for each } i, a_i\in\mathbf{a}_n. \label{eq: condid4count} \\
   & \{Y_i^{\mathbf{a}_n} \in \mathbf{Y}_n^{\mathbf{a}_n} \}  \text{ are mutually independent.} \label{eq: condid3count} 
\end{align}
\end{corollary}
Property \eqref{eq: condid2count} of Corollary \ref{cor: ciidcount} follows immediately by recognizing that $Y_i^{G_n}$ given $L_i = l, A^{G_n+}_i=a$ is equal to $f^*_Y(\epsilon_{Y_i})$ with $f^*_Y$ equivalent to $f_Y$ with $A_i$ and $L_i$ fixed to $a$ and $l$. Property \eqref{eq: condid4count} follows directly by inspection of the structural equation for each $Y_i$ and its arguments when $A_i$ is fixed to $a_i$. Property \eqref{eq: condid3count} follows if $\mathbb{P}_n(\mathbf{Y}_n^{\mathbf{a}_n}=\mathbf{y}_n ) = \prod\limits_{i=1}^n\mathbb{P}_n(Y^{\mathbf{a}_n}_i=y_i)$.

\begin{align*}
   & \mathbb{P}_n(\mathbf{Y}_n^{\mathbf{a}_n}=\mathbf{y}_n ) \\ 
   = & \prod\limits_{i=1}^n\mathbb{P}_n(Y_i^{\mathbf{a}_n}=y_i \mid \overline{Y}^{\mathbf{a}_n}_{i-1}=\overline{y}_{i-1}) \\
   = & \prod\limits_{i=1}^n\mathbb{P}_n(Y_i^{a_i}=y_i \mid Y_1^{a_1}=y_1, \dots, Y_{i-1}^{a_{i-1}}=y_{i-1}) \\
   = & \prod\limits_{i=1}^n\mathbb{P}_n(Y_i^{a_i}=y_i) \\
   = & \prod\limits_{i=1}^n\mathbb{P}_n(Y_i^{\mathbf{a}_n}=y_i).
\end{align*}

The first equality follows by laws of probability. The second follows by Property \eqref{eq: condid4count}. The third follows by independencies evident on the DAG $\mathcal{G}(\mathbf{a}_n)$ in Figure \ref{fig: proofDAG1}b (See the Proof of Proposition \ref{prop: compgformred}). The fourth follows again by Property \eqref{eq: condid4count}. 

The subsequent corollary follows simply by recognizing that, under the conditions of the previous Lemma, $\mathbb{L}_{n}$, $\mathbb{O}_{n}$ given $\mathbb{B}_{n}=\mathbb{b}_{n}$, and $\mathbb{O}_{n}^{G_n+}$ given $\mathbb{B}^{G_n+}_{n}=\mathbb{b}_{n}$ may be formulated as grouped Poisson-multinomial random variables. 

\begin{corollary} \label{cor: compid}
    Consider a law $\mathbb{P}_n$ under $\mathcal{M}^{AB}_n$. Then,
    \begin{align}
        \mathbb{Q}_{L, n}(\mathbb{l}_{n})&  =  \frac{n!}{\prod\limits_{l\in\mathcal{L}}\mathbb{l}_{n}(l)!}\prod\limits_{l\in\mathcal{L}}Q_{L}(l)^{\mathbb{l}_{n}(l)}, \\
          \mathbb{Q}_{L, n\mid l}(\mathbb{l}_n)& =  \frac{(n-1)!}{(\mathbb{l}_{n}(l)-1)!\prod\limits_{l'\in\mathcal{L}\setminus l}\mathbb{l}_{n}(l')!}Q_{L}(l)^{\mathbb{l}_{n}(l)-1}\prod\limits_{l'\in\mathcal{L}\setminus l}Q_{L}(l')^{\mathbb{l}_{n}(l')}, \\
        \mathbb{Q}_{Y, n}(\mathbb{o}_{n} \mid \mathbb{b}_n) & = \prod\limits_{ a, l} \frac{\mathbb{b}_n(a,l)!}{\prod\limits_{y\in\mathcal{Y}}\mathbb{o}_{n}(y,a,l)!}\prod\limits_{y\in\mathcal{Y}}Q_{Y}(y \mid a, l)^{\mathbb{o}_{n}(y,a,l)}, \label{eq: compidfact} \\
        \mathbb{P}_n^F(\mathbb{O}_n^{G_n+}=\mathbb{o}_n \mid \mathbb{B}_n^{G_n+}=\mathbb{b}_n) & = \prod\limits_{ a, l} \frac{\mathbb{b}_n(a,l)!}{\prod\limits_{y\in\mathcal{Y}}\mathbb{o}_{n}(y,a,l)!}\prod\limits_{y\in\mathcal{Y}}\mathbb{P}_n^F(Y^{G_n} =y \mid A^{G_n+}_i=a, L_i=l)^{\mathbb{o}_{n}(y,a,l)}. \label{eq: compidPO} 
    \end{align}
\end{corollary}

Note that for such laws, $\{\mathbb{Q}_{L, n}, \mathbb{Q}_{L, n \mid l}, \mathbb{Q}_{Y, n}\}$ are defined entirely in terms of $\{Q_L, Q_Y\}$ and known quantities, as is $\mathbb{q}^*_n$ and $q^*_i$ (see Proposition \ref{prop: intdens_general} in Appendix \ref{appsec: gencdtrs}). 

\begin{lemma} \label{eq: IDkernel}
Consider a law $\mathbb{P}_n^F\in\mathcal{M}^{AB}_n$. Then:

\begin{align}
    \mathbb{P}_n^F(Y^{G_n}_i=y \mid A^{G_n+}_i=a, L_i=l) = Q_Y(y \mid a, l), \label{eq: indIDkernel}
\end{align}
and

\begin{align}
    \mathbb{P}_n^F(\mathbb{O}_n^{G_n+}=\mathbb{o}_n \mid \mathbb{B}_n^{G_n+}=\mathbb{b}_n) = \mathbb{Q}_{Y, n}(\mathbb{o}_{n} \mid \mathbb{b}_n). \label{eq: compIDkernel}
\end{align}
    
\end{lemma}

\begin{proof}

First we prove the proposition of \eqref{eq: indIDkernel}.
    \begin{align*}
       & \mathbb{P}_n^F(Y^{G_n}_i=y \mid A^{G_n+}_i=a, L_i=l) \\
    =& \mathbb{P}_n^F(f_Y(\epsilon_{Y_i}, a, l)=y \mid A^{G_n+}_i=a, L_i=l) \\
    =& \mathbb{P}_n^F(f_Y(\epsilon_{Y_i}, a, l)=y \mid G^*_n(\epsilon_{\mathbf{L}}, \delta)=a, L_i=l) \\
    =& \mathbb{P}_n^F(f_Y(\epsilon_{Y_i}, a, l)=y \mid G'_n(\epsilon_{\mathbf{L}_{-i}}, \delta)=a, L_i=l) \\
    =& \mathbb{P}_n^F(f_Y(\epsilon_{Y_i}, a, l)=y \mid L_i=l)\\
    =& \mathbb{P}_n^F(Y^{a}=y \mid L_i=l)\\
    =& \mathbb{P}_n^F(Y^{a}=y \mid A_i, =a, L_i=l)\\
    =& \mathbb{P}_n^F(Y=y \mid A_i, =a, L_i=l) = Q_Y(y \mid a, l).\\
    \end{align*}
The first equality follows by definition of the structural equation model, and by conditions $A1$ (conditional noninterference) and $A2$ (structural invariance). The second equality follows by definition of the CR $G_n$, and by defining a function $G_n^*$ that is a composition of $G_n$ and $f_{\mathbf{L}_n}$ . The third follows by defining a function $G_n'$ that is equivalent to $G^*_n$ with $L_i$ fixed to $l$. The fourth equality follows by condition $A1$ and the exogeneity of the randomizer term $\delta$. The fifth equality follows by definition of the structural equation model. The sixth follows by condition $B1$ (no individual-level confounding of outcomes). The final equalities follow again by the structural equation model and then by Lemma \ref{lemma: ciid}. 

The proposition in \eqref{eq: compIDkernel} follows immediately by applying \eqref{eq: indIDkernel} to \eqref{eq: compidPO} and noting \eqref{eq: compidfact} in Lemma \ref{cor: compid}.

\end{proof}

\begin{theorem*}[Re-statement of Theorem \ref{theorem: CDTRID} of Section \ref{sec: ID}] 
Consider a law $\mathbb{P}_n\in\mathcal{M}^{AB}_n$  
and a regime $G_n \in \Pi_n(\kappa_n)$. 
Then   $$\mathbb{E}[ h(\mathbb{O}_n^{G_n+})] =  \sum\limits_{\mathbb{o}_{n}}h(\mathbb{o}_n)  \mathbb{f}^{G_n}(\mathbb{o}_{n}) \ 
 \text{and} \ \mathbb{E}[ h'(O_i^{G_n+})] =  \sum\limits_{o}h'(o)f_{i}^{G_n}(o),$$ 
whenever the right hand sides of the equations are well-defined. 
\end{theorem*}

\begin{proof}

We first show that $\mathbb{E}_{\mathbb{P}^F_n}[ h'(O_i^{G_n+})] =  \sum\limits_{o}h'(o)f_{i}^{G_n}(o)$. 

\begin{align*}
    & \mathbb{P}_n^F(Y^{G_n}_i=y, A^{G_n+}_i=a, L_i=l) \\
   =& \mathbb{P}_n^F(Y^{G_n}_i=y \mid A^{G_n+}_i=a, L_i=l) \mathbb{P}_n^F(A^{G_n+}_i=a \mid L_i=l) \mathbb{P}_n^F(L_i=l) \\
   =& \mathbb{P}_n(Y_i=y \mid A_i=a, L_i=l) \mathbb{P}_n^F(A^{G_n+}_i=a \mid L_i=l) \mathbb{P}_n(L_i=l) \\
   =& Q_y(y \mid a, l) q^*_i(a \mid l) Q_L(l) = f_{i}^{G_n}(o).
\end{align*}

The first equality follows by laws of probability. The second follows by Lemma \ref{eq: IDkernel}. The third follows by Lemma \ref{lemma: ciid} and definitions, and the final follows by definition. Then the result follows immediately by definition of the expectation. 

Now we prove  $\mathbb{E}_{\mathbb{P}^F_n}[ h(\mathbb{O}_n^{G_n+})] =  \sum\limits_{\mathbb{o}_{n}}h(\mathbb{o}_n)  \mathbb{f}^{G_n}(\mathbb{o}_{n})$.

\begin{align*}
    \mathbb{P}_n^F(\mathbb{O}_n^{G_n+}=\mathbb{o}_n) &  =  \mathbb{P}_n^F(\mathbb{O}_n^{G_n+}=\mathbb{o}_n, \mathbb{B}_n^{G_n+}=\mathbb{b}_n,  \mathbb{L}_n=\mathbb{l}_n) \\
    & = \mathbb{P}_n^F(\mathbb{O}_n^{G_n+}=\mathbb{o}_n \mid \mathbb{B}_n^{G_n+}=\mathbb{b}_n) 
    \mathbb{P}_n^F(\mathbb{B}_n^{G_n+}=\mathbb{b}_n \mid \mathbb{L}_n=\mathbb{l}_n)
     \mathbb{P}_n^F(\mathbb{L}_n=\mathbb{l}_n) \\
    & = \mathbb{P}_n(\mathbb{O}_n=\mathbb{o}_n \mid \mathbb{B}_n=\mathbb{b}_n) 
    \mathbb{P}_n^F(\mathbb{B}_n^{G_n+}=\mathbb{b}_n \mid \mathbb{L}_n=\mathbb{l}_n)
     \mathbb{P}_n(\mathbb{L}_n=\mathbb{l}_n) = \mathbb{f}^{G_n}(\mathbb{o}_{n})
\end{align*}

The first equality follows because $\mathbb{B}_n^{G_n+}$ and $\mathbb{L}_n$ are coarsenings of $\mathbb{O}_n^{G_n+}$. The second follows by laws of probability. The third follows by Lemma \ref{eq: IDkernel}, and the final follows by definition. Then the result follows immediately. 
\end{proof}

In Appendix \ref{appsec: gencdtrs}, we prove Proposition \ref{theorem: YbarIDrelax}. Proposition \ref{theorem: YbarIDrelax} is a stronger version of Proposition \ref{theorem: YbarID} in Section \ref{sec: ID} of the main text except with respect to the larger class of stochastic $L_i$-equitable regimes $\Pi_n^*(\kappa_n, L_i)$ defined in Appendix \ref{appsec: gencdtrs} for which the $L_i$-rank-preserving regimes $\Pi_n^{d}(\kappa_n, L_i, \Lambda)$ considered in Proposition \ref{theorem: YbarID} are a special sub-class.

\begin{proposition*} [Re-statement of Proposition \ref{prop: compgformred} of Appendix \ref{appsec: gencdtrs}] 
Consider a law $\mathbb{P}_n^F\in\mathcal{M}^{AB}_n$, a real-valued function $h''$ 
and a regime $G_n \in \Pi_n(\kappa_n)$, with $G_n$ a function of at most $\{\mathbf{V}_n, \delta\}$, i.e., $G_n: \{\mathcal{V}\}^n \times [0,1] \rightarrow \{0, 1\}^n$, and with $V_i\coloneqq c(L_i)$ for all $i$. Then $\mathbb{E}_{\mathbb{P}^F_n}[ h(\mathbb{W}_n^{G_n+})] =  \sum\limits_{\mathbb{w}_{n}}h''(\mathbb{w}_n)  \mathbb{f}^{G_n}_{red}(\mathbb{w}_{n})$
whenever the right hand side of the equation is well-defined, and also that $q^*_{i}(1\mid  l)  = \sum\limits_{\mathbb{v}_n} q^*_{i, \mathbb{v}_n}(1\mid  c(l))\mathbb{Q}_{V, n\mid c(l)}(\mathbb{v}_n).$
\end{proposition*}

\begin{proof}

Our proof for the reduced compositional g-formula result hinges on the proposition that $\mathbb{P}_n^F(\mathbb{W}_n^{G_n+}=\mathbb{o}_n \mid \mathbb{B}_n^{G_n+}=\mathbb{b}_n) = \mathbb{Q}^*_{Y, n}(\mathbb{w}_n \mid \mathbb{u}_n)$. We first show that, under a  $G_n: \{\mathcal{V}\}^n \times [0,1] \rightarrow \{0, 1\}^n$, we have that $\mathbb{P}_n^F(\mathbb{W}_n^{G_n+}=\mathbb{w}_n \mid \mathbb{U}_n^{G_n+}=\mathbb{u}_n)$ decomposes isomorphically to $\mathbb{P}_n^F(\mathbb{O}_n^{G_n+}=\mathbb{o}_n \mid \mathbb{B}_n^{G_n+}=\mathbb{b}_n)$, as in expression \eqref{eq: compidPO} of Corollary \ref{cor: compid}. In the following, we use graphical arguments for convenience. In particular we posit a particular DAG model for the joint distribution of a subset of the variables in $\mathbf{O}^{G_n+}_n$, where graphical rules of d-separation can be used to infer conditional independencies at all distributions in the DAG model. When $\mathcal{M}_n^{AB}$ (appropriately marginalized) is a submodel of the DAG model, then a d-separation in the DAG can be used to infer a conditional independence in a law $\mathbb{P}_n^F\in\mathcal{M}_n^{AB}$.

We present such a DAG $\mathcal{G}(G_n)$ in Figure \ref{fig: proofDAG1}a, representing a model for the joint distribution of $\{O_i^{G+}, \mathbf{L}_{-i}\}$, where we use bi-directed arrows between nodes to denote lack of independence between error terms for those nodes. To show that this DAG model contains all distributions for $\{O_i^{G+}, \mathbf{L}_{-i}\}$ in $\mathcal{M}_n^{AB}$, consider a supermodel of $\mathcal{M}_n^{AB}$ that is defined only by assumption $A1$ (Conditional noninterference), which we denote by $\mathcal{M}_n^{A1}$. By $A1$, only $\{A_i^{G_n+}, L_i\}$ can have arrows into $Y_i^{G_n}$ . By definition of $A_i^{G_n+}$, only $\{V_i, \mathbf{V}_{-i}\}$ may have directed arrows into $A_i^{G_n+}$ in the DAG model. By definition of $V_i\coloneqq c(L_i)$, only $L_i$ may have an arrow into $V_i$ and likewise only $\mathbf{L}_{-i}$ may have an arrow into $\mathbf{V}_{-i}$. The assumed error independencies between $\{\epsilon_{L_1}, \dots, \epsilon_{L_n}\}$ preclude any bidirected arrows between  $L_i$ and $\mathbf{L}_{-i}\}$. The assumed error independencies between $\epsilon_{Y_i}$ and $\{\boldsymbol{\epsilon}_{n}\setminus \epsilon_{i}\}$ preclude any bidirected arrows between $Y_i$ and $\mathbf{L}_{-i}$. A second DAG $\mathcal{G}(\mathbf{a}_n)$ in Figure \ref{fig: proofDAG1}b represents a model for the joint distribution of $\{\mathbf{O}_n^{\mathbf{a}_n}\setminus \mathbf{a}_n\}$, which is a supermodel for $\mathcal{M}_n^{A1}$ by identical arguments. 

Now we proceed with the proof. To show that $\mathbb{P}_n^F(\mathbb{W}_n^{G_n+}=\mathbb{o}_n \mid \mathbb{B}_n^{G_n+}=\mathbb{b}_n)$ decomposes analogous to \eqref{eq: compidPO} of Corollary \ref{cor: compid} , we argue that $\{Y^{G_n}_i \in \mathbf{Y}^{G_n}_n \mid V_i = v, A^{G_n+}_i=a\}$ are mutually independent. We establish the latter by showing that $\mathbb{P}^F_n(\mathbf{Y}^{G_n}_n=\mathbf{y}_n \mid \mathbf{U}^{G_n+}_n=\mathbf{u}_n) = \prod\limits_{i=1}^n\mathbb{P}^F_n(Y^{G_n}_i=y_i \mid \mathbf{U}^{G_n+}_n=\mathbf{u}_n)$.

\newpage
\begin{align*}
   & \mathbb{P}^F_n(\mathbf{Y}^{G_n}_n=\mathbf{y}_n \mid \mathbf{U}^{G_n+}_n=\mathbf{u}_n) \\ 
=& \prod\limits_{i=1}^n\mathbb{P}^F_n(Y_i^{G_n}=y_i \mid  \mathbf{A}^{G_n+}_{n}=\mathbf{a}_{n},  \mathbf{V}_{n}=\mathbf{v}_{n}, \overline{Y}^{G_n}_{i-1}=\overline{y}_{i-1}) \\
=& \prod\limits_{i=1}^n\mathbb{P}^F_n(Y_i^{G_n}=y_i \mid \mathbf{A}^{G_n+}_{n}=\mathbf{a}_{n},  \mathbf{V}_{n}=\mathbf{v}_{n}, \overline{Y}^{G_n}_{i-1}=\overline{y}_{i-1}, \delta) \\
=& \prod\limits_{i=1}^n\mathbb{P}^F_n(Y_i^{a_i}=y_i \mid \mathbf{A}^{G_n+}_{n}=\mathbf{a}_{n},  \mathbf{V}_{n}=\mathbf{v}_{n}, \overline{Y}^{\overline{a}_{i-1}}_{i-1}=\overline{y}_{i-1}, \delta) \\
=& \prod\limits_{i=1}^n\mathbb{P}^F_n(Y_i^{a_i}=y_i \mid  \mathbf{V}_{n}=\mathbf{v}_{n}, \overline{Y}^{\overline{a}_{i-1}}_{i-1}=\overline{y}_{i-1}) \\
=& \prod\limits_{i=1}^n\mathbb{P}^F_n(Y_i^{a_i}=y_i \mid  \mathbf{V}_{n}=\mathbf{v}_{n}) \\
=& \prod\limits_{i=1}^n\mathbb{P}^F_n(Y_i^{a_i}=y_i \mid A^{G_n+}_i=a_i, \mathbf{V}_{n}=\mathbf{v}_{n}) \\
=& \prod\limits_{i=1}^n\mathbb{P}^F_n(Y_i^{G_n}=y_i \mid \mathbf{U}^{G_n+}_n=\mathbf{u}_n) .
\end{align*}

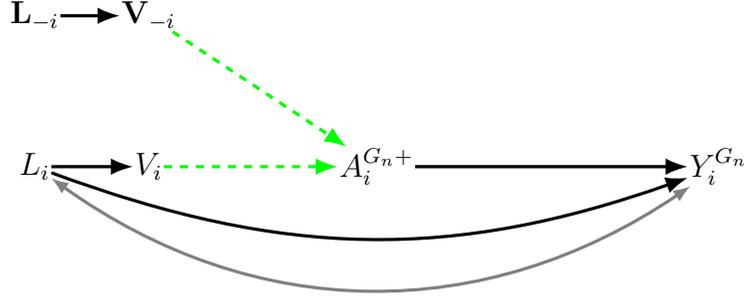
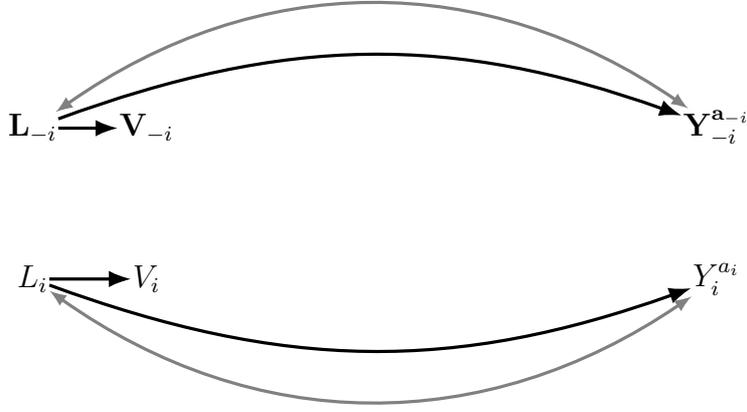
\begin{figure}[h] 
\centering
\subfloat[]{
\begin{tikzpicture}[inner sep=0.3, outer sep=0.3, scale=.5]
        \node (A2)    at  (4   *3, 0   ) {$A^{G_n+}_{i}$};
        
        \node (L1)    at  (1   *3, 4  )  {$\mathbf{L}_{-i}$};
        \node (L2)    at  (1   *3, 0   ) {${L}_{i}$};
        \node (V1)    at  (2   *3, 4  )  {$\mathbf{V}_{-i}$};
        \node (V2)    at  (2   *3, 0   ) {${V}_{i}$};

        \node (Y2)    at  (7   *3, 0   ) {$Y^{G_n}_{i}$};
        
\begin{scope}

 
        \path (L1)          edge[very thick]                   (V1);
        \path (V1)          edge[very thick, green, dashed]    (A2);
        
        \path (L2)          edge[very thick]                   (V2);
        \path (V2)          edge[very thick, green, dashed]                   (A2);
        \path (A2)          edge[very thick]                   (Y2);
        \path (L2)          edge[bend right=20,  very thick]                   (Y2);

        \path[{latex}-{latex}] (L2)    edge[bend right=35,  very thick, gray]                  (Y2);

 \end{scope}       
        
\end{tikzpicture}
}\\ \par\bigskip 
\subfloat[]{
\begin{tikzpicture}[inner sep=0.3, outer sep=0.3, scale=.5]
        
        \node (L1)    at  (1   *3, 4  )  {$\mathbf{L}_{-i}$};
        \node (L2)    at  (1   *3, 0   ) {${L}_{i}$};
        \node (V1)    at  (2   *3, 4  )  {$\mathbf{V}_{-i}$};
        \node (V2)    at  (2   *3, 0   ) {${V}_{i}$};

        \node (Y1)    at  (7   *3, 4   ) {$\mathbf{Y}^{\mathbf{a}_{-i}}_{-i}$};
        \node (Y2)    at  (7   *3, 0   ) {$Y^{a_i}_{i}$};
        
\begin{scope}

 
        \path (L1)          edge[very thick]                   (V1);
        
        \path (L2)          edge[very thick]                   (V2);
        \path (L2)          edge[bend right=20,  very thick]                   (Y2);
        \path (L1)          edge[bend left=20,  very thick]                   (Y1);
        \path[{latex}-{latex}] (L2)    edge[bend right=35,  very thick, gray]                  (Y2);
        \path[{latex}-{latex}] (L1)    edge[bend left=35,  very thick, gray]                  (Y1);

 \end{scope}       
        
\end{tikzpicture}
}
\caption{(a) DAG $\mathcal{G}({G_n})$ representing a model for the joint distribution of $\{O_i^{G_n+}, \mathbf{L}_{-i}\}$ under condition $A1$. (b) DAG $\mathcal{G}(\mathbf{a}_n)$ representing a model for the joint distribution of $\{\mathbf{O}_n^{\mathbf{a}_n}\setminus \mathbf{a}_n\}$ under condition $A1$.}
\label{fig: proofDAG1}
\end{figure}

The first equality follows by laws of probability. The second equality follows by the strict exogeneity of the randomizer, $\delta$. The third follows by definition of the structural equation model for $Y_i$. The fourth follows by noting that $\mathbf{A}^{G_n+}_{n}$ is a constant conditional on $\mathbf{V}_{n}=\mathbf{v}_{n}$ and $\delta$, and then using the strict exogeneity of the randomizer, to remove $\delta$ from the conditioning set. The fifth equality follows from DAG  $\mathcal{G}(\mathbf{a}_n)$, where we can read that $Y_i^{a_i} \independent \mathbf{Y}_{-i}^{\mathbf{a}_{-i}} \mid \mathbf{V}_n$. The sixth equality follows by reversing the steps in the first through third equalities. The final equality follows from DAG  $\mathcal{G}(G_n)$, where we can read that  $Y_i^{G_n+} \independent \mathbf{V}_{-i} \mid A_i^{G_n+}, V_i$, and again reversing the steps in the first throught third equalities.

Now we have that $\mathbb{P}_n^F(\mathbb{W}_n^{G_n+}=\mathbb{o}_n \mid \mathbb{B}_n^{G_n+}=\mathbb{b}_n) = \mathbb{Q}^*_{Y, n}(\mathbb{w}_n \mid \mathbb{u}_n)$ if $\mathbb{P}_n(Y_i^{G_n+}=y \mid U^{G_n+}_i=u) = Q^*_{Y}(y\mid a, v)$, and so we establish the latter.

\begin{align*}
    \mathbb{P}^F_n(Y_i^{G_n}=y \mid U^{G_n+}_i=u) 
        =  & \sum\limits_{l:c(l)=v}\mathbb{P}^F_n(Y_i^{G_n}=y \mid A^{G_n+}_i=a, L_i=l)\mathbb{P}^F_n(L_i=l \mid A^{G_n+}_i=a, V_i=v) \\
        =  & \sum\limits_{l:c(l)=v}\mathbb{P}^F_n(Y_i^{G_n}=y \mid A^{G_n+}_i=a, L_i=l)Q_{L\mid v}(l) \\
        = & \sum\limits_{l:c(l)=v}Q_{Y}(y \mid a, l)Q_{L\mid v}(l)
\end{align*}\

The first follows by laws of probability, and noting that $V_i$ is degenerate conditional on $L_i$. The second follows from DAG  $\mathcal{G}(G_n)$, where we can read that  $L_i\independent A_i^{G_n+} \mid  V_i$. The third follows by Lemma \ref{eq: IDkernel}. 

Following arguments isomorphic to the proof of Theorem \ref{theorem: CDTRID}, and using the immediately preceding results, we have that:

\begin{align*}
    \mathbb{P}_n^F(\mathbb{W}_n^{G_n+}=\mathbb{w}_n) &  =  \mathbb{P}_n^F(\mathbb{W}_n^{G_n+}=\mathbb{w}_n, \mathbb{U}_n^{G_n+}=\mathbb{u}_n,  \mathbb{V}_n=\mathbb{v}_n) \\
    & = \mathbb{P}_n^F(\mathbb{W}_n^{G_n+}=\mathbb{w}_n \mid \mathbb{U}_n^{G_n+}=\mathbb{u}_n) 
    \mathbb{P}_n^F(\mathbb{U}_n^{G_n+}=\mathbb{u}_n \mid \mathbb{V}_n=\mathbb{v}_n)
     \mathbb{P}_n^F(\mathbb{V}_n=\mathbb{v}_n) \\
    & = \mathbb{P}_n(\mathbb{W}_n=\mathbb{w}_n \mid \mathbb{u}_n=\mathbb{u}_n) 
    \mathbb{P}_n^F(\mathbb{U}_n^{G_n+}=\mathbb{u}_n \mid \mathbb{V}_n=\mathbb{v}_n)
     \mathbb{P}_n(\mathbb{V}_n=\mathbb{v}_n) = \mathbb{f}_{red}^{G_n}(\mathbb{w}_{n}).
\end{align*}

This concludes the proof for the reduced compositional g-formula identity of Proposition \ref{prop: compgformred}. Now we show that $q^*_{i}(1\mid  l)  = \sum\limits_{\mathbb{v}_n} q^*_{i, \mathbb{v}_n}(1\mid  c(l))\mathbb{Q}_{V, n\mid c(l)}(\mathbb{v}_n)$. The result follows immediately upon noting the following identities:

\begin{align*}
    \mathbb{P}^F_n(A^{G_n+}_i=a \mid L_i=l ) 
= &  \mathbb{P}^F_n(A^{G_n+}_i=a \mid L_i=l, V_i=c(l)) \\
= &  \mathbb{P}^F_n(A^{G_n+}_i=a \mid  V_i=c(l)) \\
= &\sum\limits_{\mathbb{v}_n} q^*_{i, \mathbb{v}_n}(1\mid  c(l))\mathbb{Q}_{V, n\mid c(l)}(\mathbb{v}_n). 
\end{align*}

The first follows by noting that $V_i$ is degenerate conditional on $L_i$. The second follows from DAG  $\mathcal{G}(G_n)$, where we can read that  $L_i\independent A_i^{G_n+} \mid  V_i$. The third follows by laws of probability. This concludes the proof.

\end{proof}

\subsubsection{Large-cluster parameters}

In this appendix, we provide a proof for Theorem \ref{theorem: YbarIDlarge} of Section \ref{sec: largecluster} that provides an identification functional for the large population parameter $\mathbb{E}_{\mathbb{P}^F_0}[ \overline{Y}_0^{\mathbf{G}_0}]$:

\begin{theorem*}[Re-statement of Theorem \ref{theorem: YbarIDlarge} of Section \ref{sec: largecluster}] 
    Consider an asymptotic law $\mathbb{P}_0$ and an asymptotic regime $\mathbf{G}_0$ following conditions $D$ and $E$, respectively. Then, $\mathbb{E}[\overline{Y}^{\mathbf{G}_0}_0] = \sum\limits_{o\in\mathcal{O}_i}yf^{\mathbf{G}_0}(o)$ whenever the right-hand side is well-defined. 
\end{theorem*}

\begin{proposition*}[Re-statement of Proposition \ref{lemma: largeint} of Section \ref{sec: largecluster}]
Consider an asymptotic law $\mathbb{P}_0$ and an asymptotic regime $\mathbf{G}_0$ following conditions $D$ and $E$, respectively. Then, 
\begin{align}
         q^*_{0}(1 \mid l)  
             & =  \begin{cases}
                     \frac{\kappa^* -  \mathbb{P}_0(\Lambda(L_i) > \omega_{0})}{\mathbb{P}_0(\Lambda(L_i) = \omega_{0})}   & : \Lambda(l)= \omega_{0}, \\
                     I\Big(\Lambda(l) > \omega_{0}\Big) & : \text{otherwise},
                \end{cases}  
\end{align}
where $\omega_{0} \coloneqq \textup{inf}\bigg\{ c: \mathbb{P}_0(\Lambda(L_i)>c) \leq \kappa^* \bigg\}.$
\end{proposition*}

We defined $\mathbb{E}_{\mathbb{P}^F_0}[\overline{Y}^{\mathbf{G}_0}_0]$ as the limit of a sequence of expectations of cluster average potential outcomes $\underset{n\to\infty}{\lim}\mathbb{E}_{\mathbb{P}^F_n}[ \overline{Y}^{{G}_n}_n]$ corresponding to clusters indexed by $n$, where this sequence is defined by an asymptotic law $\mathbb{P}^F_0$ and an asymptotic regime $\mathbf{G}_0$. Identification of these limiting parameters follows via restrictions on $\mathbb{P}^F_0$ and $\mathbf{G}_0$, which we restate here for convenience.

\begin{itemize}
    \item [$D1.$] $\mathbb{P}_n \in \mathcal{M}_n^{AB}$, for all $n\in\mathbb{N}^+.$
    \item [$D2.$] $\{f_{L_i}, \mathbb{P}_n^{\epsilon_{L_i}}, f_{Y_i}, \mathbb{P}_n^{\epsilon_{Y_i \mid pa(y)}}\} \equiv \{f_{L_i}, \mathbb{P}_{n'}^{\epsilon_{L_i}}, f_{Y_i}, \mathbb{P}_{n'}^{\epsilon_{Y_i \mid pa(y)}}\}$ for all  $n, n' \in\mathbb{N}^+.$
    \item [$E1.$] For each $n$, $G_n\in \Pi^{d}_n(\kappa_n, L_i, \Lambda)$.
    \item [$E2.$] For each $n$, $\kappa_n = \lfloor n \times \kappa^* \rfloor$, with $\kappa^* \in [0,1].$
\end{itemize}

Then consider the following sequence of identities:

\begin{align}
  &  \underset{n\to\infty}{\lim}\mathbb{E}_{\mathbb{P}^F_n}[ \overline{Y}^{{G}_n}_n], \nonumber\\
= & \underset{n\to\infty}{\lim}\mathbb{E}_{\mathbb{P}^F_n}[ {Y}^{{G}_n}_i], \nonumber\\
= & \underset{n\to\infty}{\lim} \sum\limits_oyQ_{Y}(y \mid a, l) q^{*}_{i}(a \mid l)  Q_L(l), \nonumber\\
= &  \sum\limits_oyQ_{Y}(y \mid a, l) \big\{\underset{n\to\infty}{\lim} q^{*}_{i}(a \mid l)\big\}  Q_L(l). \label{eq: limlemma}
\end{align}

The first equality holds by the proof of Proposition \ref{theorem: YbarIDrelax}, because by conditions $D1$ and $E1$, $G_n$ and $\mathbb{P}_n$ satisfy the conditions of Proposition \ref{theorem: YbarIDrelax} for each $n$. The second equality holds by Theorem \ref{theorem: CDTRID} because condition $D1$ ensures that $\mathbb{P}_n$ satisfies its conditions. The third equality holds by condition $D2$, which ensures that the measures 
$Q_{Y}$ and $Q_L$ are invariant in $n$ for $\mathbb{P}_0$.

It thus remains to prove the following Lemma:

\begin{lemma}
Consider an asymptotic law $\mathbb{P}_0^F$ and an asymptotic regime $\mathbf{G}_0$ following conditions $D$ and $E$, respectively. Then $\underset{n\to\infty}{\lim} q^{*}_{i}(a \mid l) = q^*_0(a \mid l)$ for all $a, l$.
\end{lemma}

\begin{proof}

First, note that by Lemma \ref{lemma: ciid} and condition $D$ that $\boldsymbol{\lim}\frac{\mathbb{L}_{n}(l)}{n} = Q_{L}(l)$ for all $l$. Second, note that, by laws of probability, $$q^*_{i}(1 \mid l)  = \sum\limits_{\mathbb{l}_n} q^*_{i, \mathbb{l}_n}(1 \mid  l)\mathbb{P}_n(\mathbb{L}_n=\mathbb{l}_n \mid L_i=l).$$

Then define the sequence $\mathbb{l}_0 \coloneqq\{ \mathbb{l}^0_1,  \mathbb{l}^0_2,  \mathbb{l}^0_3, \dots\}$ such that  $\lim \mathbb{P}_n( \lVert \frac{\mathbb{l}^0_n}{n} - Q_L\rVert_1  < \epsilon) = 1$ for all $\epsilon>0$ with respect to $\mathbb{l}_0$ and $\mathbb{P}_0$. We thus have that $\lim q^*_{i}(1 \mid l)  = \lim q^*_{i, \mathbb{l}^0_n}(1 \mid  l)$. Recall the following result (from Proposition \ref{prop: intdens_Lrankpres} of Appendix \ref{appsec: gencdtrs}) for $q^*_{i, \mathbb{l}_n}$ under a $L_i$-rank-preserving regime, where we had defined $S_{\Lambda}(m) \coloneqq \sum\limits_{\mathcal{L}}\mathbb{l}_n(l)I(\Lambda(l) \leq m)$, and $S^-_{\Lambda}(m)$ to be the corresponding function denoting the number of such individuals with coarsened rank group strictly greater than $m$:

  \begin{align*}
         q^*_{i,\mathbb{l}_{n}}(1 \mid l)  
             & =  \begin{cases}
                     \frac{\kappa_{n} - S_{\Lambda}^-(\omega_{\mathbb{l}_n})}{S_{\Lambda}(\omega_{\mathbb{l}_n})-S_{\Lambda}^-(\omega_{\mathbb{l}_n})}   & : \Lambda(l)= \omega_{\mathbb{l}_n}, \\
                     I\Big(\Lambda(l) > \omega_{\mathbb{l}_n}\Big) & : \text{otherwise}.
                \end{cases} 
\end{align*}
with 
  \begin{align*}
    \omega_{\mathbb{l}_n} \coloneqq \textup{inf}\bigg\{ c \in \mathbb{R}: S_{\Lambda}(c)  \le \kappa_{n}\bigg\}. 
 \end{align*}

 Note that $\omega_{\mathbb{l}_n}$  is alternatively expressed as $\textup{inf}\bigg\{ c \in \mathbb{R}: 
 \sum\limits_{\mathcal{L}}\frac{\mathbb{l}_n(l)}{n}
 I(\Lambda(l) > c) 
 \leq \frac{\kappa_{n}}{n}\bigg\}$

 Thus we have that:
  \begin{align*}
    \underset{n\to\infty}{\lim}\omega_{\mathbb{l}_n^0} = & \underset{n\to\infty}{\lim}\textup{inf}\bigg\{ c \in \mathbb{R}: \sum\limits_{\mathcal{L}}\frac{\mathbb{l}^0_n(l)}{n}I(\Lambda(l) > c) \le \frac{\kappa_{n}}{n}\bigg\} \\
  = &  \textup{inf}\bigg\{ c \in \mathbb{R}: \sum\limits_{\mathcal{L}}Q_L(l)I(\Lambda(l) > c) \le \kappa^*\bigg\} = \omega_0,
 \end{align*}
 where we use $E2$, so that $\underset{n\to\infty}{\lim}\frac{\kappa_{n}}{n} = \kappa^*$\

 We can also make the following re-formulation:

 \begin{align*}
      \frac{\kappa_{n} - S_{\Lambda}^-(\omega_{\mathbb{l}_n})}{S_{\Lambda}(\omega_{\mathbb{l}_n})-S_{\Lambda}^-(\omega_{\mathbb{l}_n})} =  
      \frac{\frac{\kappa_{n}}{n} - \sum\limits_{\mathcal{L}}\frac{\mathbb{l}_n(l)}{n}I(\Lambda(l) > \omega_{\mathbb{l}_n})}{ \sum\limits_{\mathcal{L}}\frac{\mathbb{l}_n(l)}{n}I(\Lambda(l) = \omega_{\mathbb{l}_n})}
 \end{align*}

 Thus we have that:
  \begin{align*}
    \underset{n\to\infty}{\lim}\frac{\kappa_{n} - S_{\Lambda}^-(\omega_{\mathbb{l}^0_n})}{S_{\Lambda}(\omega_{\mathbb{l}^0_n})-S_{\Lambda}^-(\omega_{\mathbb{l}^0_n})} 
  = & \frac{\kappa^* - \sum\limits_{\mathcal{L}}Q_L(l)I(\Lambda(l) > \omega_{0})}{ \sum\limits_{\mathcal{L}}Q_L(l)I(\Lambda(l) = \omega_{0})}.
 \end{align*}

Because $\Lambda$ is fixed and invariant to $n$ in the asymptotic $\mathbf{G}_0$ by condition $E1$, then we have the result $\lim q^*_{i, \mathbb{l}^0_n}(1 \mid  l) = q_0^*(a\mid l)$, thus also proving Proposition \ref{lemma: largeint} from Section \ref{sec: largecluster}. The result in Theorem \ref{theorem: YbarIDlarge} follows by plugging in $q_0^*(a\mid l)$ for $\underset{n\to\infty}{\lim} q^{*}_{i}(a \mid l)$ in expression \eqref{eq: limlemma}.

\end{proof}

\subsubsection{Optimal regimes}

We restate the optimal regime identification result of Section \ref{sec: largecluster}.

\begin{proposition*}[Re-statement of Proposition \ref{lemma: optimallarge} of Section \ref{sec: largecluster}]
Consider a law $\mathbb{P}_n\in\mathcal{M}_n^{AB}$. The CR  $G^{\mathbf{opt}}_{n}\equiv G_n\in\Pi_n(\kappa_n)$ that maximizes $\mathbb{E}[\overline{Y}_n^{G_n}]$ is the $V$-rank-preserving CR $G_n\in\Pi_n^d(\kappa_n, V_i, \Lambda)$ characterized by $\Lambda \equiv \Delta$.

Consider an asymptotic law $\mathbb{P}_0$ following condition $D$, and an asymptotic regime $\mathbf{G}_0$ following condition $E$. If $G_{n^*}\in\mathbf{G}_0$ is such an optimal regime for $\mathbb{P}_{n^*}(\mathbb{P}_0)$ for some $n^*>1$ then $G_{n}\in\mathbf{G}_0$ is the optimal such regime for $\mathbb{P}_{n}(\mathbb{P}_0)$ for all $n$. Let $q^{\mathbf{opt}}_0$ denote the intervention density under the asymptotic regime $\mathbf{G}^{\mathbf{opt}}_{0} \equiv \Big(G^{\mathbf{opt}}_{n=1}, G^{\mathbf{opt}}_{n=2}, G^{\mathbf{opt}}_{n=3},\dots\Big)$. Then, $\mathbb{E}[\overline{Y}_0^{\mathbf{G}^{\mathbf{opt}}_{0}}]$ is identified as in Theorem \ref{theorem: YbarIDlarge} and its intervention density $q^{\mathbf{opt}}_0$ is identified as in \eqref{eq: gstar_large} of Proposition \ref{lemma: largeint} where we take $\Lambda = \Delta_0\coloneqq \Delta$ and $\omega_0 =\eta_0 \coloneqq \inf\{ c  \mid \mathbb{P}_0(\Delta_0(L_i) > c) \leq \kappa^*\}.$
\end{proposition*}

The formulation of $G^{\mathbf{opt}}_{n}$ for a $\mathbb{P}^F_n\in\mathcal{M}^{AB}_n$ follows from recognizing the optimization problem as a special case of the fractional knapsack problem of \citet{dantzig1957discrete} where rewards for each individual $i$ are characterized by $\Delta_{\mathbb{P}_n^F}(L_i)$, thus corresponding to a regime in $\Pi_n^d(\kappa_n, V_i, \Lambda)$ with $\Lambda=\Delta_{\mathbb{P}_n^F}(L_i)$. The large-cluster result follows from the fact that for an asymptotic law $\mathbb{P}_0$ under conditions $D$, $\Delta_{\mathbb{P}_n^F}(L_i)$ is fixed to some $\Delta_{\mathbb{P}_0^F}$ for all $n$. The identification result for $q^{\mathbf{opt}}_0$ follows simply by substituting $\Delta_0 \coloneqq \Delta_{\mathbb{P}_0^F}$ for $\Lambda$ in $q_0^*$, as formulated in Proposition \ref{lemma: largeint}.

Note that Proposition \ref{lemma: optimallargegen} follows simply by an extension of the fractional knapsack problem that only allocates resources with non-negative rewards, which is captured by a generalized rank-and-treat regime that only provides treatment when $\Delta_{\mathbb{P}_n^F}(l_i)$ is non-negative.

\subsection{Inference results} \label{appsec: proofsEst} 

\subsubsection{Asymptotic consistency for finite-cluster parameters}

We re-state the main asymptotic consistency result of Theorem \ref{theorem: asymptcons} from Section \ref{sec: estim}.
\begin{theorem*}[Re-statement of Theorem \ref{theorem: asymptcons} of Section \ref{sec: estim}] 
    Consider $\mathbb{P}_0$ following conditions $D$ and $G_{n^*}\in \Pi_{n^{\ast}}(\kappa_{n^{\ast}})$ following $C1$. Then,
    \begin{align}
        & \boldsymbol{\lim}\sum\limits_{\mathbb{o}_{n^{\ast}}}h(\mathbb{o}_{n^{\ast}})  \tilde{\mathbb{f}}^{G_{n^{\ast}}}(\mathbb{o}_{n^*}) = \mathbb{E}[ h(\mathbb{O}_{n^{\ast}}^{G_{n^{\ast}}+})], \\
        & \boldsymbol{\lim} \sum\limits_{o}h'(o)\tilde{f}_{i}^{G_{n^*}}(o) = \mathbb{E}[ h'(O_i^{G_{n^*}+})].
    \end{align}
\end{theorem*}

\begin{proof}

Theorem \ref{theorem: asymptcons} follows if $\boldsymbol\lim\tilde{\mathbb{f}}^{G_{n^{\ast}}}(\mathbb{o}_{n^*}) = {\mathbb{f}}^{G_{n^{\ast}}}(\mathbb{o}_{n^*})$ and likewise $\boldsymbol\lim \tilde{f}_{n^*,i}^{G_{n^*}}(o) = {f}_{i}^{G_{n^*}}(o)$.

First we remind the reader that ${f}_{i}^{G_{n^*}}$ and ${\mathbb{f}}^{G_{n^{\ast}}}$ are both smooth functionals of $Q_L$ and $Q_Y$ that are fixed to $n^*$ across the index $n$ of the asymptotic law $\mathbb{P}_0$. Furthermore, we remind the reader that $\tilde{f}_{i}^{G_{n^*}}$ and $\tilde{\mathbb{f}}^{G_{n^{\ast}}}$ are simply plug-in estimators of ${f}_{i}^{G_{n^*}}$ and ${\mathbb{f}}^{G_{n^{\ast}}}$, with respect to $\tilde{Q}_L$ and $\tilde{Q}_Y$. Thus, our result follows via Slutsky's theorem if $\boldsymbol\lim\tilde{Q}_L(l)= Q_L(l)$ for all $l$ and if   $\boldsymbol\lim\tilde{Q}_Y(y \mid a, l)= {Q}_Y(y \mid a, l)$ for all $a,l$ such that  $\overline{q}^*_n(a \mid l)Q_L(l) > 0 \implies \overline{q}_n(a \mid l)>0.$

Under an asymptotic law $\mathbb{P}^F_0$ following condition $D$ we have that $\boldsymbol\lim\frac{\mathbb{L}_n(l)}{n} = Q_L(l)$ for all $l$, since $\boldsymbol\lim\frac{\mathbb{L}_n(l)}{n} = \frac{1}{n}\sum\limits_{i=1}^nI(L_i=l)$ is an empirical average of iid random variables (see Lemma \ref{lemma: ciid}), and Condition $D$ ensures that $Q_L$ remains fixed across laws in $\mathbb{P}^F_0$. $\frac{\mathbb{O}_n(y, a,l)}{\mathbb{B}_n(a,l)}$ is also an empirical average of a randomly varying number of iid random variables. The counterexample of Appendix \ref{appsec: estim} illustrates that weak positivity condition $C0$ is not sufficient for consistency of $\boldsymbol\lim\frac{\mathbb{O}_n(y, a,l)}{\mathbb{B}_n(a,l)}$ because $\mathbb{B}_n(a,l)$ under these conditions will not necessarily grow with some rate that is a function of $n$. Strong positivity condition $C1$ ensures that $\mathbb{B}_n(a,l)$ indeed grows at such a rate so that the number of individuals in this set will almost surely approach infinitely many as $n$ grows. Thus we have that  $\boldsymbol\lim \frac{\mathbb{O}_n(y, a,l)}{\mathbb{B}_n(a,l)} = {Q}_Y(y \mid a, l)$. 
    
\end{proof}

\subsubsection{Semiparametric efficiency for regular large-cluster parameters}

\begin{theorem*}[Re-statement of Theorem \ref{thm: asympeffGopt} from Appendix \ref{appsec: estim}]
    Consider an asymptotic law $\mathbb{P}_0$ following condition $D$ and an asymptotic optimal regime $\mathbf{G}_{0}^{\mathbf{opt}}$ following condition $E$. Suppose further that condition $F^*$ holds. Then the following properties also hold:

    \begin{itemize}
        \item[(1.)] $\Psi^{\mathbf{opt}}$ is a pathwise differentiable parameter with efficient non-parametric influence function $\Phi_0(\overline{P}_0, g^{\mathbf{opt}}_{i}, \eta_0)$.
        \item[(2.)] $\tilde{\Psi}^{\mathbf{opt}}_{n}$ is the asymptotically-efficient, regular and asymptotically linear estimator of $\Psi^{\mathbf{opt}}$.
        \item[(3.)] An asymptotically valid two-sided $1-\alpha$ confidence interval is given by $$\tilde{\Psi}^{\mathbf{opt}}_{n} \pm z_{1-\alpha/2}\frac{\tilde{\sigma}^{2, \mathbf{opt}}_n}{\sqrt{n}}.$$
    \end{itemize}
\end{theorem*}

\begin{proof}
    We adapt the proof from \citetSM{luedtke2016optimal}. Let $Pf := \mathbb{E}_P[f(O_i)]$. Asymptotic linearity and normality, will not follow from standard arguments because the elements of $\mathbf{O}_n$ are not iid. The rest of the proof follows from \citetSM{luedtke2016optimal}, where in a slight abuse of notation, we at times let $\Phi_0(\cdot)$ denote the random variable  $\Phi_0(\cdot)(O_i)$, and likewise for $\Phi_1(\cdot)$ and $\Phi_2(\cdot)$.

\citetSM{luedtke2016optimal} show that
\begin{align}
    \Tilde{\Psi}_n^{\mathbf{opt}} - \Psi^{\mathbf{opt}} &= (\Tilde{P}_n - \Bar{P}_0)\Phi_0(\Bar{P},g_i^{\mathbf{opt}}, \eta) + (\Tilde{P}_n - \Bar{P})(\Phi_0(\Tilde{P}_n, \Tilde{g}_i^{\mathbf{opt}},\eta) - \Phi_0(\Bar{P},g_i^{\mathbf{opt}}, \eta))\nonumber \\
    &+ R_0(\Tilde{g}_i^{\mathbf{opt}}, \Tilde{P}_n) + o_{\Bar{P}}(n^{-1/2}).
    \label{eq: Fundamental equation}
\end{align}
.

The second and third terms of Equation \eqref{eq: Fundamental equation} are $o_{\Bar{P}}(n^{-1/2})$ by the assumptions of Theorem 4 in \citetSM{luedtke2016optimal} (a subset of $F^*$). \citetSM{luedtke2016optimal} use the central limit theorem (and hence the i.i.d. assumption) to show asymptotic normality of $\sqrt{n}(\Tilde{P}_n - \Bar{P})\Phi_0(\Bar{P},g_i^{\mathbf{opt}}, \eta)$. As the elements of $\mathbf{O}_n$ are not iid, we cannot use the central limit theorem to obtain the same conclusion; we need to show that $\sqrt{n}(\Tilde{P}_n - \Bar{P})\Phi_0(\Bar{P},g_i^{\mathbf{opt}}, \eta)$ is asymptotically normal.

Recall the definition of $\Phi_0$:
\begin{align}
    \Phi_0(\Bar{P},g_i^{\mathbf{opt}}, \eta)(o) &= \Phi_1(\Bar{P}, g_i^{\mathbf{opt}})(o) + \Phi_2(g_i^{\mathbf{opt}}, \eta)(o) \nonumber \\
    &= \frac{I(a = g_i^{\mathbf{opt}}(l))}{\Bar{q}_{0}(a \mid l)}(y - \mathbb{E}_{\Bar{P}}(Y_i \mid A_i = a , L_i = l))\nonumber\\
    &+ \{\mathbb{E}_{\Bar{P}}(Y_i \mid A_i = g_i^{\mathbf{opt}}(l), L_i = l) - \mathbb{E}_{\Bar{P}}[\mathbb{E}_{\Bar{P}}(Y_i \mid g_i^{\mathbf{opt}}(L_i), L_i)]\}\nonumber\\
    &- \eta(g_i^{\mathbf{opt}}(l) - \kappa^*).
   \label{eq: D_0 formula} 
\end{align}
The first equality follows from the definition of $\Phi_0$ as $g_i^{\mathbf{opt}}$ is deterministic by Condition $F3^*$ (and thus not a function of $\delta)$ so we can omit the expectation over $\delta$ in expression \ref{eq: IF}, and the second equality follows from the definitions of $\Phi_1$ and $\Phi_2$.

By the law of total expectation and as $\Tilde{P}_n$ has the empirical distribution of $L_i$ marginally (given by $\mathbb{L}_n$),
\begin{align*}
    \Tilde{P}_n \Phi_1(\Bar{P}, g_i^{\mathbf{opt}}) = \frac{1}{n} \sum_{i = 1}^n \Tilde{P}_n(\Phi_1(\Bar{P}, g_i^{\mathbf{opt}}) \mid L_i).
\end{align*}
However,
\begin{align}
   & \Tilde{P}_n(\Phi_1(\Bar{P}, g_i^{\mathbf{opt}}) \mid L_i) \nonumber \\  & = \Tilde{P}_n\left[\frac{I(A_i = g_i^{\mathbf{opt}}(L_i))}{\Bar{q}_0(g_i^{\mathbf{opt}}(L_i) \mid L_i)}(Y_i - \mathbb{E}_{\Bar{P}}(Y_i \mid g_i^{\mathbf{opt}}(L_i), L_i)) \ \Big| \  L_i\right] \nonumber\\
    &= \Tilde{P}_n\left[\frac{I(A_i = g_i^{\mathbf{opt}}(L_i))}{\Bar{q}_0(g_i^{\mathbf{opt}}(L_i) \mid L_i)}Y_i \ \Big| \  L_i\right] - \Tilde{P}_n\left[\frac{I(A_i = g_i^{\mathbf{opt}}(L_i))}{\Bar{q}_0(g_i^{\mathbf{opt}}(L_i) \mid L_i)}\mathbb{E}_{\Bar{P}}(Y_i \mid g_i^{\mathbf{opt}}(L_i), L_i) \ \Big| \  L_i\right]. \label{eq: separation of conditional expectation}
\end{align}
The first equality comes from the definition of $\Phi_1$, and the second equality follows from linearity of the expectation.

We focus on the first term of Equation \eqref{eq: separation of conditional expectation}.
\begin{align*}
    \Tilde{P}_n\left[\frac{I(A_i = g_i^{\mathbf{opt}}(L_i))}{\Bar{q}_0(g_i^{\mathbf{opt}}(L_i) \mid L_i)}Y_i \ \Big| \ L_i\right] &= \frac{1}{\Bar{q}_0(g_i^{\mathbf{opt}}(L_i) \mid L_i)} \Tilde{P}_n\left[I(A_i = g_i^{\mathbf{opt}}(L_i))Y_i \ \Big| \ L_i \right].
\end{align*}
The equality follows by linearity of the expectation as $\Bar{q}_0(g_i^{\mathbf{opt}}(L_i) \mid L_i)$ is a constant conditional on $L_i$.

Then, 
\begin{align*}
    \Tilde{P}_n\left[\frac{I(A_i = g_i^{\mathbf{opt}}(L_i))}{\Bar{q}_0(g_i^{\mathbf{opt}}(L_i) \mid L_i)}Y_i \ \Big| \ L_i\right] &= \frac{1}{\Bar{q}_0(g_i^{\mathbf{opt}}(L_i) \mid L_i)} \Tilde{P}_n\left[I(A_i = g_i^{\mathbf{opt}}(L_i))Y_i \ \Big| \ L_i\right] \\
    &= \frac{1}{\Bar{q}_0(g_i^{\mathbf{opt}}(L_i) \mid L_i)} \mathbb{E}_{\Tilde{P}_n}(Y_i \mid g_i^{\mathbf{opt}}(L_i), L_i) \Tilde{q}_n(g_i^{\mathbf{opt}}(L_i)\mid L_i).
\end{align*}
The first equality is the same as the previous equation, and the second inequality follows as 
\begin{align*}
    \Tilde{P}_n[I(A_i = g_i^{\mathbf{opt}}(L_i))Y_i \mid L_i] &= \Tilde{P}_n[\Tilde{P}_n[I(A_i = g_i^{\mathbf{opt}}(L_i))Y_i \mid A_i, L_i]\mid L_i]\\
    &= \Tilde{P}_n[I(A_i = g_i^{\mathbf{opt}}(L_i))\Tilde{P}_n[Y_i \mid A_i, L_i]\mid L_i]\\
    &= \Tilde{P}_n[I(A_i = g_i^{\mathbf{opt}}(L_i))\Tilde{P}_n[Y_i \mid g_i^{\mathbf{opt}}(L_i), L_i]\mid L_i]\\
    &= \Tilde{P}_n[I(A_i = g_i^{\mathbf{opt}}(L_i))\mid L_i]\Tilde{P}_n[Y_i \mid g_i^{\mathbf{opt}}(L_i), L_i]\\
    &= \Tilde{q}_n(g_i^{\mathbf{opt}}(L_i) \mid L_i)\Tilde{P}_n[Y_i \mid g_i^{\mathbf{opt}}(L_i), L_i].
\end{align*}

The first equality follows from the law of total expectation; the second equality follows as $I(A_i = g_i^{\mathbf{opt}}(L_i))$ is a constant conditional on $A_i$ and $L_i$; the third equality follows as $I(A_i = g_i^{\mathbf{opt}}(L_i))\Tilde{P}_n[Y_i \mid A_i, L_i] = I(A_i = g_i^{\mathbf{opt}}(L_i))\Tilde{P}_n[Y_i \mid g_i^{\mathbf{opt}}(L_i), L_i]$ almost surely; the fourth equality follows as $\Tilde{P}_n[Y_i \mid g_i^{\mathbf{opt}}(L_i), L_i]$ is a constant; the last equality follows by the definition of $\Tilde{q}_n$.

However,
\begin{align*}
    &\frac{1}{n} \sum_{i = 1}^n \Tilde{P}_n\left[\frac{I(A_i = g_i^{\mathbf{opt}}(L_i))}{\Bar{q}_0(g_i^{\mathbf{opt}}(L_i) \mid L_i)}Y_i \ \Big | \ L_i\right]\\
    &= \frac{1}{n} \sum_{i = 1}^n \frac{\Tilde{q}_n(g_i^{\mathbf{opt}}(L_i)\mid L_i)}{\Bar{q}_0(g_i^{\mathbf{opt}}(L_i) \mid L_i)} \mathbb{E}_{\Tilde{P}_n}(Y_i \mid g_i^{\mathbf{opt}}(L_i), L_i) \\
    &= \frac{1}{n} \sum_{i = 1}^n \left(\frac{\Tilde{q}_n(g_i^{\mathbf{opt}}(L_i)\mid L_i)}{\Bar{q}_0(g_i^{\mathbf{opt}}(L_i) \mid L_i)}  - 1\right)\mathbb{E}_{\Tilde{P}_n}(Y_i \mid g_i^{\mathbf{opt}}(L_i), L_i)\\
    &+ \frac{1}{n} \sum_{i = 1}^n (\mathbb{E}_{\Tilde{P}_n}(Y_i \mid g_i^{\mathbf{opt}}(L_i), L_i) - \mathbb{E}_{\Bar{P}}(Y_i \mid g_i^{\mathbf{opt}}(L_i), L_i))\\
    &+ \frac{1}{n} \sum_{i = 1}^n \mathbb{E}_{\Bar{P}}(Y_i \mid g_i^{\mathbf{opt}}(L_i), L_i)\\
    &= \frac{1}{n} \sum_{i = 1}^n \left(\frac{\Tilde{q}_n(g_i^{\mathbf{opt}}(L_i)\mid L_i) - \Bar{q}_0(g_i^{\mathbf{opt}}(L_i) \mid L_i)}{\Bar{q}_0(g_i^{\mathbf{opt}}(L_i) \mid L_i)}\right)\mathbb{E}_{\Tilde{P}_n}(Y_i \mid g_i^{\mathbf{opt}}(L_i), L_i)\\
    &+ \frac{1}{n} \sum_{i = 1}^n (\mathbb{E}_{\Tilde{P}_n}(Y_i \mid g_i^{\mathbf{opt}}(L_i), L_i) - \mathbb{E}_{\Bar{P}}(Y_i \mid g_i^{\mathbf{opt}}(L_i), L_i))\\
    &+ \frac{1}{n} \sum_{i = 1}^n \mathbb{E}_{\Bar{P}}(Y_i \mid g_i^{\mathbf{opt}}(L_i), L_i).\\
\end{align*}
The first equality follows from the previous equation; the second equality follows from adding and subtracting $n^{-1} \sum_{i = 1}^n (\mathbb{E}_{\Tilde{P}_n}(Y_i \mid g_i^{\mathbf{opt}}(L_i), L_i) - \mathbb{E}_{\Bar{P}}(Y_i \mid g_i^{\mathbf{opt}}(L_i), L_i))$; the third equality follows from replacing $1$ by $\Bar{q}_0(g_i^{\mathbf{opt}}(L_i) \mid L_i)/\Bar{q}_0(g_i^{\mathbf{opt}}(L_i) \mid L_i)$.

Similarly,
\begin{align*}
    &\frac{1}{n} \sum_{i = 1}^n \Tilde{P}_n\left[\frac{I(A_i = g_i^{\mathbf{opt}}(L_i))}{\Bar{q}_0(g_i^{\mathbf{opt}}(L_i) \mid L_i)}\mathbb{E}_{\Bar{P}}(Y_i \mid g_i^{\mathbf{opt}}(L_i), L_i) \ \Big | \ L_i\right]\\
    &= \frac{1}{n} \sum_{i = 1}^n \left(\frac{\Tilde{q}_n(g_i^{\mathbf{opt}}(L_i)\mid L_i) - \Bar{q}_0(g_i^{\mathbf{opt}}(L_i) \mid L_i)}{\Bar{q}_0(g_i^{\mathbf{opt}}(L_i) \mid L_i)}\right)\mathbb{E}_{\Bar{P}}(Y_i \mid  g_i^{\mathbf{opt}}(L_i), L_i)\\
    &+ \frac{1}{n} \sum_{i = 1}^n \mathbb{E}_{\Bar{P}}(Y_i \mid g_i^{\mathbf{opt}}(L_i), L_i).
\end{align*}
This equality follows as $\mathbb{E}_{\Bar{P}}(Y_i \mid g_i^{\mathbf{opt}}(L_i), L_i)$ and $\Bar{q}_0(g_i^{\mathbf{opt}}(L_i) \mid L_i)$ are constants, by definition of $\Tilde{q}_n$, and by adding and subtracting $n^{-1} \sum_{i = 1}^n \mathbb{E}_{\Bar{P}}(Y_i \mid g_i^{\mathbf{opt}}(L_i), L_i)$.

Hence, 
\begin{align*}
   & \Tilde{P}_n \Phi_1(\Bar{P}, g_i^{\mathbf{opt}}) \\
    &= \frac{1}{n} \sum_{i = 1}^n \frac{\Tilde{q}_n(g_i^{\mathbf{opt}}(L_i)\mid L_i) - \Bar{q}_0(g_i^{\mathbf{opt}}(L_i) \mid L_i)}{\Bar{q}_0(g_i^{\mathbf{opt}}(L_i) \mid L_i)}(\mathbb{E}_{\Tilde{P}_n}(Y_i \mid g_i^{\mathbf{opt}}(L_i), L_i) - \mathbb{E}_{\Bar{P}}(Y_i \mid g_i^{\mathbf{opt}}(L_i), L_i))\\
    &+ \frac{1}{n} \sum_{i = 1}^n (\mathbb{E}_{\Tilde{P}_n}(Y_i \mid g_i^{\mathbf{opt}}(L_i), L_i) - \mathbb{E}_{\Bar{P}}(Y_i \mid g_i^{\mathbf{opt}}(L_i), L_i)).
\end{align*}
This equality follows from our previous derivations.

Using similar derivations based on the law of total expectation, we can replace empirical averages and expectations under $\Tilde{P}_n$ with expectations with respect to $\Bar{P}$, and replace $\Tilde{q}_n$ with $\Bar{q}_0$. Hence, we find $\Bar{P} \Phi_1(\Bar{P}, g_i^{\mathbf{opt}}) = 0$. Thus,
\begin{align*}
    &(\Tilde{P}_n - \Bar{P})\Phi_0(\Bar{P}, g_i^{\mathbf{opt}}, \eta)\\
    &= (\Tilde{P}_n - \Bar{P})\Phi_1(\Bar{P}, g_i^{\mathbf{opt}}) + (\Tilde{P}_n - \Bar{P})\Phi_2(g_i^{\mathbf{opt}}, \eta) - (\Tilde{P}_n - \Bar{P})[\eta(g_i^{\mathbf{opt}} - \kappa^*)]\\
    &= (\Tilde{P}_n - \Bar{P})(\Phi_2(g_i^{\mathbf{opt}}, \eta) - \eta(g_i^{\mathbf{opt}}(L_i) - \kappa^*))\\
    &+ \frac{1}{n} \sum_{i = 1}^n \left(\frac{\Tilde{q}_n(g_i^{\mathbf{opt}}(L_i)\mid L_i) - \Bar{q}_0(g_i^{\mathbf{opt}}(L_i) \mid L_i)}{\Bar{q}_0(g_i^{\mathbf{opt}}(L_i) \mid L_i)}\right)(\mathbb{E}_{\Tilde{P}_n}(Y_i \mid g_i^{\mathbf{opt}}(L_i), L_i) - \mathbb{E}_{\Bar{P}}(Y_i \mid g_i^{\mathbf{opt}}(L_i), L_i))\\
    &+ \frac{1}{n} \sum_{i = 1}^n (\mathbb{E}_{\Tilde{P}_n}(Y_i \mid g_i^{\mathbf{opt}}(L_i), L_i) - \mathbb{E}_{\Bar{P}}(Y_i \mid g_i^{\mathbf{opt}}(L_i), L_i))
\end{align*}
The first equality follows from definitions of $\Phi_0$, $\Phi_1$, and $\Phi_2$; the second equality follows from our previous derivations.

However, 
\begin{align*}
    \sqrt{n}(\Tilde{P}_n - \Bar{P})(\Phi_2(g_i^{\mathbf{opt}}, \eta)(L_i) - \eta(g_i^{\mathbf{opt}}(L_i) - \kappa^*))
\end{align*}
is asymptotically normal by the central limit theorem, as the $L_i$ are i.i.d and $\Phi_2(g_i^{\mathbf{opt}}, \eta)$ is only a function of $L_i$. Furthermore, 
\begin{align*}
    &\sqrt{n}\left[\frac{1}{n} \sum_{i = 1}^n \bigg(\frac{\Tilde{q}_n(g_i^{\mathbf{opt}}(L_i)\mid L_i) - \Bar{q}_0(g_i^{\mathbf{opt}}(L_i) \mid L_i)}{\Bar{q}_0(g_i^{\mathbf{opt}}(L_i) \mid L_i)}\right)(\mathbb{E}_{\Tilde{P}_n}(Y_i \mid g_i^{\mathbf{opt}}(L_i), L_i) - \mathbb{E}_{\Bar{P}}(Y_i \mid g_i^{\mathbf{opt}}(L_i), L_i))\\
    &+ \frac{1}{n} \sum_{i = 1}^n (\mathbb{E}_{\Tilde{P}_n}(Y_i \mid g_i^{\mathbf{opt}}(L_i), L_i) - \mathbb{E}_{\Bar{P}}(Y_i \mid g_i^{\mathbf{opt}}(L_i), L_i)) \bigg] = o_{\Bar{P}}(n^{-1/2}),
\end{align*}
if $\mathbb{E}_{\Tilde{P}_n}(Y_i \mid g_i^{\mathbf{opt}}(L_i), L_i)$ verifies conditions analogous to Assumptions 4 and 5 in \citetSM{luedtke2016optimal}, that is, Condition $F7^*$, and $\Tilde{q}_n$ is bounded. 
\end{proof}
\begin{remark*}
    The proof of Proposition \ref{prop: asympeffG0} is a direct extension of this proof for a $g_i$ defined by a known $\Lambda$ function instead of $\Delta_0$.
\end{remark*}

\subsubsection{Sub-cluster positivity}

\begin{lemma*}[Re-statment of Lemma \ref{lemma: subclus} from Appendix \ref{appsec: estim}]
    Consider an asymptotic law $\mathbb{P}_0^F$ following conditions $D$ and sub-cluster positivity condition $C2$. Then $\tilde{q}_n$ is $L^2(\mathbb{P}_0)$ consistent for $\overline{q}_0$, i.e. $\mathbb{E}[\lVert \tilde{q}_n - \overline{q}_0\rVert_2^2] = o(1).$
\end{lemma*}

\begin{proof}
Consider the following equations.
\begin{align}
    &\mathbb{E}_{\mathbb{P}_0^F}[(\Tilde{q}_n(a \mid l) - \Bar{q}_0(a \mid l))^2] = o(1) \label{eq: L2-consistency},\\
    & \Tilde{q}_n(a\mid l) - \Bar{q}_0(a \mid l) = o_{\mathbb{P}_0}(1). \label{eq: convergence in prob}
\end{align}
 We aim to prove \eqref{eq: L2-consistency}, which implies \eqref{eq: convergence in prob}, the limit condition in the Strong Positivity Conditions $C1$ and $C1^*$, which we consider in Appendix \ref{appsec: counterexamp_cons}. For $L^2$-consistency, we need \eqref{eq: L2-consistency}.

We prove Equation \eqref{eq: L2-consistency} in several steps. We show that $\Tilde{q}_n(a | l)$ can be expressed as a sum of $m_n$ functions of sub-cluster variables in Lemma \ref{lemma: sub-clusterpos2}, that is, the number of terms in the sum grows proportionately to $n$. However, we cannot apply the usual central limit theorem, since sub-cluster sizes $\mathbb{W}_n(c)$ are not i.i.d. Therefore, we show $\Tilde{q}_n(a|l)$ can be expressed as a sum of functions of $\mathbb{W}_n(c)$ and independent error terms, and can be approximated with a normal distribution in Lemma \ref{lemma: sub-clusterpos central limit}. Finally, we prove Equation \eqref{eq: L2-consistency} holds using this approximation. 

Lemma \ref{lemma: sub-clusterpos central limit} is used to approximate the mean and variance of $\Tilde{q}_n(a|l)$ or any general sum of functions of the $\mathbb{W}_n(c)$ and independent error terms. When Conditions $D$ and $C2$ hold, these approximations will be sufficient as $n$ grows to prove that Equation \eqref{eq: L2-consistency} holds as we show after Lemma \ref{lemma: sub-clusterpos2}. 

For a given $n$, under sub-cluster positivity $C2$, we have that $\mathbb{W}_n(c)$ are only defined for $c \in \{1,\dots, m_n<n\}$. We want to use the central limit theorem of \citepSM[Theorem 27.4]{billingsley_probability_1995} to find estimates of the first two moments of the sum of functions of the sub-cluster sizes $\mathbb{W}_n(c)$. The CLT of \citepSM[Theorem 27.4]{billingsley_probability_1995} requires a stochastic process that is defined for sub-clusters $c$ beyond the original $m_n$. Hence, we define for each $n$, $\mathbb{W}_n(c)$ for $c\in \{m_n+1,\dots, 2n\}$. In particular, we let $\mathbb{W}_n(m_n + 1), \ldots, \mathbb{W}_n(n)$ be defined as random variables with degenerate distributions at $0$, and let $\mathbb{W}_n(n + 1), \ldots, \mathbb{W}_n(2n)$ be another independent set of sub-cluster sizes of a cluster of size $n$ with the same joint distribution as $\mathbb{W}_n(1), \ldots, \mathbb{W}_n(n)$. We define $\mathbb{W}_{n}(jn + 1),\ldots, \mathbb{W}_{n}(\{j+1\}n)$ analogously for $j=2,3,4,\dots$. This construction is used only to give us approximations of the mean and variance of a sum of functions of the sub-cluster sizes, but we will need further argumentation to show that Equation \eqref{eq: L2-consistency} holds.

\begin{lemma} \label{lemma: sub-clusterpos central limit}
Consider an asymptotic law $\mathbb{P}_0^F$ following conditions D and sub-cluster positivity condition C2. Suppose that we have a function $f_n'$ of $\mathbb{W}_n(c)$ and the $\epsilon_{\mathbb{W}_n,c}$, where the elements of $\{{\epsilon}_{\mathbb{w}_n, 1}, \dots, {\epsilon}_{\mathbb{w}_n, m_n}\}$ are mutually independent and jointly independent of $\mathbb{W}_n$. Then, for any $n^*$ there exist $\mu_{n^*}\in \mathbb{R}$ and $\sigma_{n^*} > 0$ such that

\begin{align}
    \sqrt{n'}\bigg(\sum\limits_{c = 1}^{n'} f_{n^*}'(\mathbb{W}_{n^*}(c), \epsilon_{\mathbb{W}_{n^*},c}) - \mu_{n^*} \bigg) \underset{n' \to \infty}{\to} \mathcal{N}(0, \sigma_{n^*}^2). \label{eq: intermediate asymptotic normality}
\end{align}
    
\end{lemma}
\begin{proof}
    We want to apply a central limit theorem for dependent variables \citepSM[Theorem 27.4]{billingsley_probability_1995}. This requires convergence of a different alpha mixing coefficient, namely the strong alpha mixing coefficient. Let $X_1, X_2, \ldots$ be some stationary stochastic process. Let $\alpha_{\mathbb{P}_0}$ be the limit of $\alpha_{\mathbb{P}_n}$. We define the strong alpha $k$-mixing coefficient of the $X_i$ as
\begin{align*}
    &\alpha_k := \alpha_{\mathbb{P}_0}((X_1, \ldots, X_{n'}), (X_{n' + k}, \ldots))= \sup_{\substack{A \in \sigma(X_1, \ldots, X_{n'}),\\ B \in \sigma(X_{n' + k}, \ldots)}} \mid \mathbb{P}_0^F(A \cap B) - \mathbb{P}_0^F(A)\mathbb{P}_0^F(B)\mid.
\end{align*}
The strong alpha $k$-mixing coefficient $\alpha_k$ is the alpha mixing coefficient of random variables $(X_1, \ldots, X_{n'})$ and $(X_{n'+k}, \ldots)$. This condition differs from Condition $C2.3$, which is not required for this result.

To apply the \citetSM{billingsley_probability_1995} central limit theorem, we need the stochastic process $\{X_i\}_{i = 1}^{\infty}$ to be stationary and $\alpha_k = O({k}^{-5})$. In particular, this holds if the stochastic process is $m$-dependent for any $m \in \mathbb{N}$, that is, $(X_1, \ldots, X_{n'})$ is independent of $(X_{n' + k}, \ldots, X_{n' + k + l})$ for any $l \in \mathbb{N}$ whenever $k \geq m$. We want this condition to hold for the stochastic process given by $\{f_n'(\mathbb{W}_n(c), \epsilon_{\mathbb{W}_n,c})\}_{c = 1}^{\infty}$. This sequence is $n$-dependent by construction, and thus verifies $\alpha_{k} = O({k}^{-5})$. Furthermore, the sequence is stationary by definition of the sub-clusters beyond $n$; a negligible proportion can be allowed to not be identically distributed, which we guarantee by sub-cluster positivity conditions $C2.1$ for the degenerate $\frac{n - m_n}{n}$ proportion of sub-clusters and $C2.2$ for the $m_n$ first sub-clusters.

\end{proof}

\begin{lemma} \label{lemma: sub-clusterpos2}
    Consider an asymptotic law $\mathbb{P}_0$ following conditions $D$ and sub-cluster positivity condition $C2$. 
    For each $n$, there exists a transformation $\boldsymbol{\epsilon}_{\mathbb{w}_n} \equiv \{{\epsilon}_{\mathbb{w}_n, 1}, \dots, {\epsilon}_{\mathbb{w}_n, m_n}\}$ of $\mathbf{O}_n$ and a fixed function $f_n$ such that the following properties hold:
        \begin{itemize}
            \item [(1)] For each $c\in \{1,\dots, m_n\}$, ${\epsilon}_{\mathbb{w}_n, c} \independent \mathbb{W}_n(c)$; 
            \item [(2)] The elements of $\boldsymbol{\epsilon}_{\mathbb{w}_n}$ are mutually independent; and
            \item [(3)] $\tilde{q}_n(a \mid l) = \frac{1}{m_n} \sum\limits_{c=1}^{m_n} \frac{m_n}{\mathbb{L}_n(l)}f_n(\mathbb{W}_n(c), {\epsilon}_{\mathbb{w}_n, c})$.
        \end{itemize}
\end{lemma}

\begin{proof}

Let $\mathbb{B}_{n,c}(a, l)$ denote the random variable for a cluster of size $n$ indicating the number of individuals in sub-cluster $c$ with treatment level $a$ and covariate level $l$, and let $\mathbf{L}_{n,c}$ denote the matrix of covariates for individuals in sub-cluster $c$. 

By definition then we have that: $$\tilde{q}_n(a \mid l) = \frac{1}{m_n}\sum\limits_{c=1}^{m_n}\frac{m_n}{\mathbb{L}_n(l)}\mathbb{B}_{n,c}(a, l).$$

We have by sub-cluster positivity condition $C2.4$ that $\mathbb{B}_{n}\equiv \mathbb{B}^{G_n+}_{n}$ for a sub-cluster regime $G_n\in \Pi^F_{n, m_n}(\mathbf{G}_m)$. We thus also have that $\mathbb{B}_{n,c}(a, l) = \mathbb{B}^{G_n+}_{n,c}(a, l)$. Note that by condition $C2.1$, sub-clusters have size bounded by $m$. A defining feature of sub-cluster regimes is that treatment in each cluster is assigned according to a regime in the $m\times m$ matrix of regimes $\mathbf{G}_m$ - the regime implemented in sub-cluster $c$ corresponds to the element in $\mathbf{G}_m$ in the row indicated by cluster size, $\mathbb{W}_n(c)$, and the column indicated by the number of treatment units available to that sub-cluster, $\kappa_{n, c}$. Therefore,  $\mathbb{B}^{G_n+}_{n,c}(a, l) \equiv \mathbb{B}^{G_{\mathbb{W}_n(c)}^*+}_{\mathbb{W}_n(c)}(a, l)$, with $G_{\mathbb{W}_n(c)}^* \equiv [\mathbf{G}_m]_{\mathbb{W}_n(c),\kappa_{n,c}}$. By definition, then we have that:

\begin{align}
    \mathbb{B}_{n,c}(a, l) = \sum\limits_{i=1}^{\mathbb{W}_n(c)}I\Bigg([\mathbf{G}_m]_{\mathbb{W}_n(c),\kappa_{n,c}, i}(\mathbf{L}_{n, c})=a, L_i=l\Bigg). \label{eq: sub-clustercounts}
\end{align}

First we prove property $(1)$ of Lemma \ref{lemma: sub-clusterpos2}. Note that $\mathbf{L}_{n, c}$ is a function of $\epsilon_{\mathbf{L}_n} \in \boldsymbol{\epsilon}_n$, and that $\mathbb{W}_n(c)$ is a function of $\epsilon_{\mathbf{W}_n}$. Take $\epsilon_{\mathbb{w}_n, c} \subset \{ \boldsymbol{\epsilon}_n, \kappa_{n, c}\}$ to be the subset of error terms that generate $\kappa_{n, c}$, $\mathbf{L}_{n, c}$, and $L_i$, which with $\mathbb{W}_n(c)$ comprise the random arguments to the function generating $\mathbb{B}_{n,c}(a, l)$, the right hand side of \eqref{eq: sub-clustercounts}. Due to sub-cluster positivity condition $C2.5$ and  $C2.6$ then $\mathbb{W}_n(c) \independent \epsilon_{\mathbb{w}_n, c}$.  Thus, we have property $(1)$ of Lemma \ref{lemma: sub-clusterpos2}.

Second, we prove property $(2)$ of Lemma \ref{lemma: sub-clusterpos2}. Note that each $L_i$ in the right hand side of \eqref{eq: sub-clustercounts} is an element of $\mathbf{L}_{n,c}\equiv \{f_L(\epsilon_{L_i}): W_i=c\}$ and that $\{\mathbf{L}_{n,1}, \dots, \mathbf{L}_{n,m_n}\}$ form a partition of $\mathbf{L}_n$. Let $\epsilon_{\mathbf{L}_{n, c}} \equiv \{\epsilon_{L_i}: W_i=c\}$ and so  $\epsilon_{\mathbb{w}_n, c} \equiv \Big\{ \epsilon_{\mathbf{L}_{n, c}}, \kappa_{n, c}\Big\}.$ Due to the mutual independence of the elements of $\epsilon_{\mathbf{L}_n}$, via condition $A1$ (conditional noninterference) for all $n$ (by condition $D$), and sub-cluster positivity condition $C2.5$, then the elements of $\{\epsilon_{\mathbb{w}_n, 1}, \dots, \epsilon_{\mathbb{w}_n, m_n}\}$ are mutually independent. Thus we have property $(2)$ of Lemma \ref{lemma: sub-clusterpos2}. 

Finally, we prove property $(3)$ of Lemma \ref{lemma: sub-clusterpos2}, where it suffices to show that $\mathbb{B}_{n,c}(a, l)\equiv f_n(\mathbb{W}_n(c), {\epsilon}_{\mathbb{w}_n, c})$. Note that by condition $A2$ (structural invariance) for all $n$ (by condition $D$), the elements of $\mathbf{L}_{n,c}$ are generated from $\epsilon_{\mathbf{L}_{n, c}}$ via the same structural equations $f_L$ that are invariant in $c$. Note also that $\mathbf{A}_{n, c}$ is equivalent to $\mathbf{A}^{G_{\mathbb{W}_n(c)}^*+}_{\mathbb{W}_n(c)}$ which is generated by the function $[\mathbf{G}_m]_{\cdot,\cdot}(\cdot)$ taking arguments, $\mathbb{W}_n(c)$, $\kappa_{n, c}$ and $\mathbf{L}_{n, c}\equiv \{f_L(\epsilon_{L_i}): W_i=c\}$. Composing $[\mathbf{G}_m]_{\cdot,\cdot}(\cdot)$ and the functions $f_L$ generating $\mathbf{L}_{n, c}$, we have a function of $\epsilon_{\mathbb{w}_n, c}$ and $\mathbb{W}_n(c)$ that is invariant in $c$. Since $\mathbb{B}_{n,c}(a, l)$ is a function of $\{\mathbf{A}_{n,c}, \mathbf{L}_{n, c}\}$, and since $\{\mathbf{A}_{n,c}, \mathbf{L}_{n, c}\}$ are each at most $c$-invariant functions of $\epsilon_{\mathbb{w}_n, c}$ and $\mathbb{W}_n(c)$, then we have that $\mathbb{B}_{n,c}(a, l)$ is a function of $\epsilon_{\mathbb{w}_n, c}$ and $\mathbb{W}_n(c)$ that is invariant in $c$, which we may denote by $f_n$. Thus we have property $(3)$ of Lemma \ref{lemma: sub-clusterpos2}. Thus concludes the proof of Lemma \ref{lemma: sub-clusterpos2}.
    
\end{proof}

Lemma \ref{lemma: sub-clusterpos2} shows that
\begin{align*}
    \Tilde{q}_n(a \mid l) = \frac{1}{m_n} \sum\limits_{c = 1}^{m_n} \frac{m_n}{\mathbb{L}_n(l)} f_n(\mathbb{W}_n(c), \epsilon_{\mathbb{W}_n,c}).
\end{align*}
Moreover, by dividing and multiplying by $n$,
\begin{align*}
    \Tilde{q}_n(a \mid l) &= \frac{1}{n} \sum\limits_{c = 1}^{m_n} \frac{n}{\mathbb{L}_n(l)} f_n(\mathbb{W}_n(c), \epsilon_{\mathbb{W}_n, c})\\
    &=  \sum\limits_{c = 1}^{m_n}  f_n'(\mathbb{W}_n(c), \epsilon_{\mathbb{W}_n, c})\\
    &=  \sum\limits_{c = 1}^{n}  f_n'(\mathbb{W}_n(c), \epsilon_{\mathbb{W}_n, c}),
\end{align*}
where $f'_n(\mathbb{W}_n(c), \epsilon_{\mathbb{W}_n, c}) := \frac{1}{\mathbb{L}_n(l)}f_n(\mathbb{W}_n(c), \epsilon_{\mathbb{W}_n, c}) = \frac{1}{\mathbb{L}_n(l)}\mathbb{B}_{n,c}(a,l) \leq 1$ for $c = 1, \ldots, m_n$ and $f'_n(\mathbb{W}_n(c), \epsilon_{\mathbb{W}_n, c}) := 0$ for $c = m_n + 1, \ldots, n$.

 Furthermore, by Lemma \ref{lemma: sub-clusterpos central limit}, for any $n^*$ there exist $\mu_{n^*}\in \mathbb{R}$ and $\sigma_{n^*} > 0$ such that

\begin{align}
    \sqrt{n'}\bigg(\sum\limits_{c = 1}^{n'} f'_{n^*}(\mathbb{W}_{n^*}(c), \epsilon_{\mathbb{W}_{n^*},c}) - \mu_{n^*} \bigg) \underset{n' \to \infty}{\to} \mathcal{N}(0, \sigma_{n^*}^2). 
\end{align}
In fact, assuming that the first sub-cluster is non-degenerate in the sense of Condition $C2.2$, we have that:
\begin{align}
    \sigma_{n^*}^2 &= \text{Var}_{\mathbb{P}_0}[f'_{n^*}(\mathbb{W}_{n^*}(1), \epsilon_{\mathbb{W}_{n^*},1})]\nonumber\\
    &+ 2 \sum\limits_{k = 1}^\infty \text{Cov}_{\mathbb{P}_0}[f'_{n^*}(\mathbb{W}_{n^*} (1), \epsilon_{\mathbb{W}_{n^*},1}), f'_{n^*}(\mathbb{W}_{n^*}(1 + k), \epsilon_{\mathbb{W}_{n^*},1 + k})] \nonumber\\
    &= \text{Var}_{\mathbb{P}_0}[f'_{n^*}(\mathbb{W}_{n^*}(1), \epsilon_{\mathbb{W}_{n^*},1})]\nonumber\\
    &+ 2 \sum\limits_{k = 1}^{n^* - 1} \text{Cov}_{\mathbb{P}_0}[f'_{n^*}(\mathbb{W}_{n^*} (1), \epsilon_{\mathbb{W}_{n^*},1}), f'_{n^*}(\mathbb{W}_{n^*}(1 + k), \epsilon_{\mathbb{W}_{n^*}, 1 + k})]\nonumber\\
    &= \text{Var}_{\mathbb{P}_0}[f'_{n^*}(\mathbb{W}_{n^*}(1), \epsilon_{\mathbb{W}_{n^*},1})]\nonumber\\
    &+ 2 \sum\limits_{k = 1}^{n^* - 1} \mathbb{E}_{\mathbb{P}_0}[f'_{n^*}(\mathbb{W}_{n^*} (1), \epsilon_{\mathbb{W}_{n^*}, 1}) f'_{n^*}(\mathbb{W}_{n^*}(1 + k), \epsilon_{\mathbb{W}_{n^*}, 1 + k})]\nonumber\\
    &- \mathbb{E}_{\mathbb{P}_0}[f'_{n^*}(\mathbb{W}_{n^*} (1), \epsilon_{\mathbb{W}_{n^*}, 1})]\mathbb{E}_{\mathbb{P}_0}[f'_{n^*}(\mathbb{W}_{n^*}(1 + k), \epsilon_{\mathbb{W}_{n^*}, 1 + k})] \nonumber\\
    &= \text{Var}_{\mathbb{P}_0}[f'_{n^*}(\mathbb{W}_{n^*}(1), \epsilon_{1})]+ 2 \sum\limits_{k = 1}^{n^* - 1} \sum\limits_{\substack{\mathbb{w}_{n^*}(j), \varepsilon_j,\\j \in \{1, 1 + k\}}} f'_{n^*}(\mathbb{w}_{n^*}(1), \varepsilon_1)f'_{n^*}(\mathbb{w}_{n^*}(1 + k), \varepsilon_{1 + k})\nonumber\\
    &\cdot(\mathbb{P}_0(\mathbb{W}_{n^*}(1) = \mathbb{w}_{n^*}(1), \epsilon_{\mathbb{W}_{n^*}, 1} = \varepsilon_1, \mathbb{W}_{n^*}(1 + k) = \mathbb{w}_{n^*}(1 + k), \epsilon_{\mathbb{W}_{n^*}, 1 + k} = \varepsilon_{1 + k})\nonumber\\
    &- \mathbb{P}_0(\mathbb{W}_{n^*}(1) = \mathbb{w}_{n^*}(1), \epsilon_{\mathbb{W}_{n^*},1} = \varepsilon_1)\mathbb{P}_0(\mathbb{W}_{n^*}(1 + k) = \mathbb{w}_{n^*}(1 + k), \epsilon_{\mathbb{W}_{n^*}, 1 + k} = \varepsilon_{1 + k}))\nonumber\\
    &= \text{Var}_{\mathbb{P}_0}[f'_{n^*}(\mathbb{W}_{n^*}(1), \epsilon_{\mathbb{W}_{n^*},1})] +  O(n^*)\nonumber\\
    &\cdot\sup_{\mathbb{w}_{n^*}(1), \mathbb{w}_{n^*}(1 + k), \varepsilon_1, \varepsilon_{1 + k}}(\mathbb{P}_0(\mathbb{W}_{n^*}(1) = \mathbb{w}_{n^*}(1), \epsilon_{\mathbb{W}_{n^*}, 1} = \varepsilon_1, \mathbb{W}_{n^*}(1 + k) = \mathbb{w}_{n^*}(1 + k), \epsilon_{\mathbb{W}_{n^*}, 1 + k} = \varepsilon_{1 + k})\nonumber\\
    &- \mathbb{P}_0(\mathbb{W}_{n^*}(1) = \mathbb{w}_{n^*}(1), \epsilon_{\mathbb{W}_{n^*}, 1} = \varepsilon_1)\mathbb{P}_0(\mathbb{W}_{n^*}(1 + k) = \mathbb{w}_{n^*}(1 + k), \epsilon_{\mathbb{W}_{n^*}, 1 + k} = \varepsilon_{1 + k}))\nonumber\\
    &= \text{Var}_{\mathbb{P}_0}[f'_{n^*}(\mathbb{W}_{n^*}(1), \epsilon_{\mathbb{W}_{n^*}, 1})] +  O(n^*) \cdot\alpha_{\mathbb{P}_0}(\mathbb{W}_{n^*}(1), \mathbb{W}_{n^*}(1 + k))\nonumber\\
    &- \mathbb{P}_0(\mathbb{W}_{n^*}(1) = \mathbb{w}_{n^*}(1))\mathbb{P}_0(\mathbb{W}_{n^*}(1 + k) = \mathbb{w}_{n^*}(1 + k))\nonumber\\
    &= \text{Var}_{\mathbb{P}_0}[f'_{n^*}(\mathbb{W}_{n^*}(1), \epsilon_{\mathbb{W}_{n^*}, 1})] +  o(n^*)\nonumber\\
    &= o(n^*). \label{eq: vanishing intermediate variance}
\end{align}

The first equality follows from the statement of Theorem 27.4 in \citetSM{billingsley_probability_1995}; the second equality follows from the $n$-dependence property; the third equality follows from the definition of covariance; the fourth equality follows from the definition of expectations, where we write $\varepsilon_j$ for an instantiation of $\epsilon_{\mathbb{W}_{n^*},j}$; the fifth equality follows from the definition of $\sup$, condition $C2.1$, and as $f_n' \leq 1$; the sixth equality follows from $(\epsilon_{\mathbb{W}_n,1}, \ldots, \epsilon_{\mathbb{W}_n, m_n}) \independent (\mathbb{W}_n(1), \ldots, \mathbb{W}_n(m_n))$, and $\mathbb{P}_0(\epsilon_{\mathbb{W}_n,1} = \epsilon_1')\mathbb{P}_0(\epsilon_{\mathbb{W}_n, 1 + k} = \epsilon_{1 + k}') \leq 1$; the seventh equality follows from condition $C2.3$; the last equality follows as $\underset{n \to \infty}{\lim}f_n'(\mathbb{W}_n(c), \epsilon_{\mathbb{W}_n, c}) = \underset{n \to \infty}{\lim} \frac{\mathbb{B}_{n,c}(a,l)}{\mathbb{L}_n(l)} = 0$ by condition $C2.1$.

Hence,
\begin{align*}
    \mathbb{E}_{\mathbb{P}_0^F}[(\Tilde{q}_n(a \mid l) - \Bar{q}_0(a \mid l))^2] &= \mathbb{E}_{\mathbb{P}_0^F}[(\Tilde{q}_n(a \mid l) - \mu_n + \mu_n - \Bar{q}_0(a \mid l))^2]\\
    &=  \mathbb{E}_{\mathbb{P}_0^F}[(\Tilde{q}_n(a \mid l) - \mu_n)^2 + (\mu_n - \Bar{q}_0(a \mid l))^2\\
    &+ 2(\Tilde{q}_n(a \mid l) - \mu_n)(\mu_n - \Bar{q}_0(a \mid l))]\\
    &= \frac{1}{n}(\mathbb{E}_{\mathbb{P}_0^F}[n(\Tilde{q}_n(a \mid l) - \mu_n)^2] - \sigma_n^2) + \frac{\sigma_n^2}{n} + (\mu_n - \Bar{q}_0(a \mid l))^2\\
    &+ 2(\mu_n - \Bar{q}_0(a \mid l))\mathbb{E}_{P_0^F}[\Tilde{q}_n(a \mid l) - \mu_n]\\
    &= \frac{1}{n} o(1) + o(1) + o(1)^2 + 2o(1)o(1)\\
    &= o(1),
\end{align*}
where the first equality comes from adding and subtracting $\mu_n$; the second equality comes from expanding the square; the third equality comes from linearity of expectation, multiplying and dividing $(\Tilde{q}_n(a \mid l) - \mu_n)^2$ by $n$, adding and subtracting $\sigma_n^2$, and taking constants out of expectations; the fourth equality comes from Equations \eqref{eq: intermediate asymptotic normality} and \eqref{eq: vanishing intermediate variance} and the continuous mapping theorem, and the fact that $\mu_n \to \Bar{q}_0(a \mid l)$ as $\Bar{q}_0(a \mid l) := \underset{n \to \infty}{\lim} \Tilde{q}_n(a \mid l)$ and $\mu_n = \mathbb{E}_{\mathbb{P}_0^F}[\Tilde{q}_n(a \mid l)] + o(1)$ by Equation \eqref{eq: intermediate asymptotic normality}; the last equality follows from the definition of $o(1) \underset{n \to \infty}{\to} 0$.
\end{proof}

\clearpage
\section{Additional remarks and supporting materials}\label{appsec: misc}

In this appendix, we present additional remarks and additional supporting materials for the main text. Section \ref{appsec: analysis} provides tables of numerical results for the analysis of ICU data presented in Section \ref{sec: icu example}. Section \ref{appsec: graphs} provides directed acyclic graphs illustrating the conditions defining model $\mathcal{M}^{AB}_n$. Section \ref{appsec: modelrobust} provides a remark illustrating a robustness property of optimal rank-preserving CRs that is not enjoyed by conventional optimal IRs for limited resource settings considered previously in the literature. Section \ref{appsec: opt} provides simple elaborations to definitions of optimal CRs (from Proposition \ref{lemma: optimallarge} in Section \ref{sec: largecluster}) to accommodate the more realistic settings in which treatment may be harmful for some non-empty $v$-subset of the population with positive probability.  

\subsection{Data Analysis} \label{appsec: analysis}

In this appendix, Tables \ref{tab: LargeTable} - \ref{tab: CountTable} provide numerical results supporting Figures \ref{fig: LargeComp} and \ref{fig: finite} for the data analysis in Section \ref{sec: icu example}. Figures \ref{fig: OptFullGridComp} and \ref{fig: OptTightGrid} illustrate the function that permits translation of the asymptotic optimal $V$-rank-preserving regime $\mathbf{G}_{0}^{\mathbf{opt}_1}$ into a individualized dynamic treatment regime, conditional on the cumulative density function of $\Delta_0(V_i)$, an individual's average treatment effect conditional on $V_i$. Figure \ref{fig: Candidate_full_nat} provides the analogous illustration for the candidate asymptotic $V$-rank-preserving regime $\mathbf{G}_{0}^{{\Lambda}_1}$, with the key distinction being that, under $\mathbf{G}_{0}^{{\Lambda}_1}$, individuals will continue to be assigned treatment so long as a treatment unit is available, whereas under $\mathbf{G}_{0}^{\mathbf{opt}_1}$, individuals with negative CATEs will no longer be assigned treatment.

\begin{table}[ht]
\centering
\footnotesize
\begin{tabular}{r| llll}
  \hline \hline
\normalsize{  $\kappa^*$}  & \normalsize{ $\mathbf{G}_{0}^{\mathbf{opt}_1}$ }  & \normalsize{  $\mathbf{G}_{0}^{\Lambda_1}$ }  & \normalsize{  $\mathbf{G}_{0}^{\mathbf{opt}_2}$ }  & \normalsize{  $\mathbf{G}_{0}^{\Lambda_2}$ }\\ 
  \hline
  0.00      & 0.695 (0.685, 0.705) & 0.695 (0.685, 0.705) & 0.695 (0.685, 0.705) & 0.695 (0.685, 0.705) \\ 
  0.05      & 0.700 (0.690, 0.710) & 0.684 (0.672, 0.696) & 0.696 (0.686, 0.706) & 0.695 (0.686, 0.704) \\ 
  0.10      & 0.711 (0.701, 0.721) & 0.683 (0.670, 0.696) & 0.697 (0.688, 0.707) & 0.695 (0.686, 0.705) \\ 
  0.15      & 0.703 (0.693, 0.714) & 0.679 (0.665, 0.692) & 0.699 (0.689, 0.709) & 0.696 (0.686, 0.706) \\ 
  0.20      & 0.705 (0.695, 0.716) & 0.678 (0.664, 0.692) & 0.702 (0.692, 0.712) & 0.699 (0.689, 0.708) \\ 
  0.25      & 0.711 (0.700, 0.722) & 0.673 (0.659, 0.687) & 0.705 (0.695, 0.715) & 0.701 (0.692, 0.710) \\ 
  0.30      & 0.714 (0.703, 0.725) & 0.676 (0.662, 0.690) & 0.707 (0.697, 0.718) & 0.703 (0.694, 0.712) \\ 
  0.35      & 0.715 (0.704, 0.727) & 0.675 (0.660, 0.690) & 0.707 (0.697, 0.718) & 0.705 (0.696, 0.715) \\ 
  0.40      & 0.716 (0.705, 0.728) & 0.675 (0.660, 0.690) & 0.708 (0.696, 0.719) & 0.708 (0.698, 0.718) \\ 
  0.45      & 0.722 (0.710, 0.734) & 0.674 (0.659, 0.690) & 0.709 (0.697, 0.720) & 0.710 (0.699, 0.721) \\ 
  0.50      & 0.718 (0.705, 0.730) & 0.671 (0.655, 0.687) & 0.711 (0.699, 0.723) & 0.712 (0.701, 0.724) \\ 
  0.55      & 0.718 (0.705, 0.730) & 0.675 (0.659, 0.692) & 0.712 (0.700, 0.725) & 0.713 (0.701, 0.726) \\ 
  0.60      & 0.718 (0.705, 0.730) & 0.673 (0.657, 0.690) & 0.712 (0.700, 0.725) & 0.711 (0.699, 0.723) \\ 
  0.65      & 0.718 (0.705, 0.730) & 0.676 (0.660, 0.693) & 0.713 (0.700, 0.725) & 0.708 (0.696, 0.720) \\ 
  0.70      & 0.718 (0.705, 0.730) & 0.674 (0.657, 0.691) & 0.713 (0.700, 0.725) & 0.705 (0.692, 0.717) \\ 
  0.75      & 0.718 (0.705, 0.730) & 0.673 (0.656, 0.691) & 0.713 (0.700, 0.725) & 0.701 (0.689, 0.714) \\ 
  0.80      & 0.718 (0.705, 0.730) & 0.677 (0.660, 0.695) & 0.713 (0.700, 0.725) & 0.698 (0.685, 0.712) \\ 
  0.85      & 0.718 (0.705, 0.730) & 0.680 (0.662, 0.697) & 0.713 (0.700, 0.725) & 0.695 (0.681, 0.709) \\ 
  0.90      & 0.718 (0.705, 0.730) & 0.685 (0.667, 0.702) & 0.713 (0.700, 0.725) & 0.692 (0.677, 0.707) \\ 
  0.95      & 0.718 (0.705, 0.730) & 0.685 (0.668, 0.703) & 0.713 (0.700, 0.725) & 0.689 (0.673, 0.705) \\ 
  1.00      & 0.718 (0.705, 0.730) & 0.685 (0.668, 0.702) & 0.713 (0.700, 0.725) & 0.686 (0.669, 0.703) \\ 
  $\widetilde{\mathbb{E}}[\overline{A}_0] \approx 0.27$ & 0.711 (0.700, 0.722) & 0.675 (0.661, 0.689) & 0.706 (0.696, 0.717) & 0.702 (0.693, 0.711) \\ 
   \hline \hline
\end{tabular}\normalsize
    \caption{Estimated values of the expected cluster-average potential outcome under large-cluster regimes, $\mathbb{E}[\overline{Y}^{\mathbf{G}_0}_0]$, with 95\% confidence intervals, under varying resource constraints ($\kappa^*$). Estimates for parameters defined by $\mathbf{G}_{0}^{\mathbf{opt}_1}$ and $\mathbf{G}_{0}^{\Lambda_1}$ use semiparametric efficient methods adapted from \citet{luedtke2016optimal}. Estimates for parameters defined by $\mathbf{G}_{0}^{\mathbf{opt}_2}$ and $\mathbf{G}_{0}^{\Lambda_2}$ use online-learning-based methods adapted from \citet{luedtke2016statistical}. See Appendix \ref{appsec: conditions} for details on inferential procedures.}
    \label{tab: LargeTable}
\end{table}

\begin{table}[ht]
\centering
\begin{tabular}{r|lll}
  \hline \hline
$\kappa_{n^*=20}$ & IPW$^{\dagger}$ & g-formula$^{\ddagger}$ & Doubly-robust$^{\circ}$ \\ 
  \hline
 0  & 0.700 (0.690, 0.709) & 0.705 (0.695, 0.714)  & 0.695 (0.685, 0.705)\\ 
 1  & 0.701 (0.692, 0.709) & 0.704 (0.695, 0.713)  & 0.695 (0.686, 0.704)\\ 
 2  & 0.702 (0.693, 0.710) & 0.703 (0.695, 0.712)  & 0.695 (0.686, 0.705)\\ 
 3  & 0.703 (0.694, 0.711) & 0.703 (0.694, 0.712)  & 0.696 (0.686, 0.70 6)\\ 
 4  & 0.703 (0.694, 0.711) & 0.703 (0.695, 0.712)  & 0.699 (0.689, 0.708)\\ 
 5  & 0.704 (0.695, 0.712) & 0.704 (0.695, 0.712)  & 0.701 (0.692, 0.710)\\ 
 6  & 0.704 (0.695, 0.713) & 0.704 (0.695, 0.713)  & 0.703 (0.694, 0.712)\\ 
 7  & 0.705 (0.695, 0.713) & 0.705 (0.695, 0.714)  & 0.705 (0.696, 0.715)\\ 
 8  & 0.705 (0.695, 0.714) & 0.705 (0.696, 0.714)  & 0.708 (0.698, 0.718)\\ 
 9  & 0.705 (0.695, 0.714) & 0.705 (0.695, 0.715)  & 0.710 (0.699, 0.721)\\ 
 10 & 0.704 (0.694, 0.714) & 0.704 (0.694, 0.715)  & 0.712 (0.701, 0.724)\\ 
 11 & 0.702 (0.692, 0.713) & 0.703 (0.692, 0.714)  & 0.713 (0.701, 0.726)\\ 
 12 & 0.700 (0.689, 0.710) & 0.701 (0.690, 0.712)  & 0.711 (0.699, 0.723)\\ 
 13 & 0.696 (0.685, 0.707) & 0.698 (0.688, 0.710)  & 0.708 (0.696, 0.720)\\ 
 14 & 0.692 (0.681, 0.703) & 0.695 (0.684, 0.707)  & 0.705 (0.692, 0.717)\\ 
 15 & 0.688 (0.676, 0.699) & 0.691 (0.679, 0.703)  & 0.701 (0.689, 0.714)\\ 
 16 & 0.683 (0.671, 0.695) & 0.687 (0.674, 0.699)  & 0.698 (0.685, 0.712)\\ 
 17 & 0.678 (0.666, 0.691) & 0.683 (0.670, 0.696)  & 0.695 (0.681, 0.709)\\ 
 18 & 0.674 (0.659, 0.688) & 0.679 (0.665, 0.693)  & 0.692 (0.677, 0.707)\\ 
 19 & 0.669 (0.653, 0.685) & 0.675 (0.660, 0.691)  & 0.689 (0.673, 0.705)\\ 
 20 & 0.665 (0.648, 0.681) & 0.671 (0.655, 0.688)  & 0.686 (0.669, 0.703)\\ 
   \hline\hline
\end{tabular}
    \caption{Estimated values of the expected cluster-average potential under regimes defined by the rank function, $\Lambda_2$, with 95\% confidence intervals, under varying resource constraints ($\kappa_{n^*}$). $\dagger$ Estimates of finite cluster parameters $\mathbb{E}[\overline{Y}^{G^{\Lambda_2}_{n^*=20}}_{n^*=20}]$ using a semi-parameteric IPW estimator of the individual-level g-formula, based on parametric models for $\overline{q}_0$ with 95\% confidence intervals computed using 500 nonparametric bootstrap samples; $\ddagger$ estimates of  $\mathbb{E}[\overline{Y}^{G^{\Lambda_2}_{n^*=20}}_{n^*=20}]$ using a parametric plug-in estimator of the individual-level g-formula based on parametric models for $Q_Y$ with 95\% confidence intervals computed using 500 nonparametric bootstrap samples; $\circ$ Estimates of the analogous large-cluster parameter $\mathbb{E}[\overline{Y}^{\mathbf{G}^{\Lambda_2}_{0}}_{0}]$ using online-learning-based methods adapted from \citet{luedtke2016statistical}.}
    \label{tab: FiniteTable}

\end{table}

\begin{table}
\centering
\tiny
\begin{tabular}{r|lllll}
 \hline
  \hline
 $\kappa_{n^*}$  & $x=$6 & $x=$7 & $x=$8 & $x=$9 & $x=$10   \\ 
  \hline
 0  & 0.000 (0.000, 0.000) & 0.001 (0.001, 0.001) & 0.004 (0.003, 0.005) & 0.012 (0.010, 0.015) & 0.031 (0.026, 0.037)  \\ 
 1  & 0.000 (0.000, 0.000) & 0.001 (0.001, 0.001) & 0.004 (0.003, 0.005) & 0.012 (0.010, 0.015) & 0.030 (0.026, 0.036)  \\ 
 2  & 0.000 (0.000, 0.000) & 0.001 (0.001, 0.001) & 0.004 (0.003, 0.005) & 0.012 (0.009, 0.014) & 0.030 (0.025, 0.035)  \\ 
 3  & 0.000 (0.000, 0.000) & 0.001 (0.001, 0.001) & 0.004 (0.003, 0.005) & 0.011 (0.009, 0.014) & 0.029 (0.025, 0.035)  \\ 
 4  & 0.000 (0.000, 0.000) & 0.001 (0.001, 0.001) & 0.004 (0.003, 0.005) & 0.011 (0.009, 0.014) & 0.029 (0.025, 0.034)  \\ 
 5  & 0.000 (0.000, 0.000) & 0.001 (0.001, 0.001) & 0.003 (0.003, 0.004) & 0.011 (0.009, 0.014) & 0.029 (0.024, 0.034)  \\ 
 6  & 0.000 (0.000, 0.000) & 0.001 (0.001, 0.001) & 0.003 (0.003, 0.004) & 0.011 (0.009, 0.013) & 0.028 (0.024, 0.034)  \\ 
 7  & 0.000 (0.000, 0.000) & 0.001 (0.001, 0.001) & 0.003 (0.003, 0.004) & 0.011 (0.009, 0.013) & 0.028 (0.024, 0.034)  \\ 
 8  & 0.000 (0.000, 0.000) & 0.001 (0.001, 0.001) & 0.003 (0.003, 0.004) & 0.011 (0.009, 0.013) & 0.028 (0.024, 0.034)  \\ 
 9  & 0.000 (0.000, 0.000) & 0.001 (0.001, 0.001) & 0.003 (0.003, 0.004) & 0.011 (0.008, 0.013) & 0.028 (0.023, 0.034)  \\ 
 10 & 0.000 (0.000, 0.000) & 0.001 (0.001, 0.001) & 0.003 (0.003, 0.004) & 0.011 (0.008, 0.014) & 0.028 (0.023, 0.034)  \\ 
 11 & 0.000 (0.000, 0.000) & 0.001 (0.001, 0.001) & 0.003 (0.003, 0.005) & 0.011 (0.009, 0.014) & 0.029 (0.024, 0.035)  \\ 
 12 & 0.000 (0.000, 0.000) & 0.001 (0.001, 0.001) & 0.004 (0.003, 0.005) & 0.012 (0.009, 0.015) & 0.030 (0.025, 0.037)  \\ 
 13 & 0.000 (0.000, 0.000) & 0.001 (0.001, 0.002) & 0.004 (0.003, 0.006) & 0.013 (0.010, 0.016) & 0.033 (0.026, 0.040)  \\ 
 14 & 0.000 (0.000, 0.000) & 0.001 (0.001, 0.002) & 0.005 (0.003, 0.006) & 0.014 (0.011, 0.018) & 0.035 (0.028, 0.043)  \\ 
 15 & 0.000 (0.000, 0.000) & 0.001 (0.001, 0.002) & 0.005 (0.004, 0.007) & 0.015 (0.012, 0.020) & 0.038 (0.031, 0.047)  \\ 
 16 & 0.000 (0.000, 0.001) & 0.002 (0.001, 0.002) & 0.006 (0.004, 0.008) & 0.017 (0.013, 0.022) & 0.041 (0.033, 0.051)  \\ 
 17 & 0.000 (0.000, 0.001) & 0.002 (0.001, 0.003) & 0.007 (0.005, 0.009) & 0.019 (0.014, 0.025) & 0.045 (0.036, 0.055)  \\ 
 18 & 0.001 (0.000, 0.001) & 0.002 (0.001, 0.003) & 0.008 (0.005, 0.011) & 0.021 (0.016, 0.028) & 0.048 (0.038, 0.060)  \\ 
 19 & 0.001 (0.000, 0.001) & 0.003 (0.002, 0.004) & 0.009 (0.006, 0.012) & 0.023 (0.017, 0.031) & 0.052 (0.041, 0.065)  \\ 
 20 & 0.001 (0.000, 0.001) & 0.003 (0.002, 0.005) & 0.010 (0.006, 0.014) & 0.026 (0.018, 0.035) & 0.056 (0.043, 0.070)  \\ 
  \hline
 $\kappa_{n^*}$ & $x=$11 & $x=$12 & $x=$13 & $x=$14 &$x=$15   \\ 
  \hline
 0  &  0.066 (0.057, 0.075) & 0.115 (0.105, 0.125) & 0.164 (0.157, 0.171) & 0.192 (0.191, 0.192) & 0.179 (0.170, 0.186) \\
 1  &  0.065 (0.057, 0.073) & 0.114 (0.104, 0.123) & 0.164 (0.156, 0.170) & 0.191 (0.191, 0.192) & 0.179 (0.172, 0.186) \\
 2  &  0.064 (0.056, 0.072) & 0.112 (0.103, 0.122) & 0.163 (0.156, 0.170) & 0.191 (0.190, 0.192) & 0.180 (0.172, 0.187) \\
 3  &  0.063 (0.056, 0.072) & 0.112 (0.103, 0.121) & 0.162 (0.155, 0.169) & 0.191 (0.190, 0.192) & 0.181 (0.173, 0.187) \\
 4  &  0.062 (0.055, 0.071) & 0.111 (0.102, 0.120) & 0.161 (0.154, 0.168) & 0.191 (0.190, 0.192) & 0.181 (0.174, 0.188) \\
 5  &  0.062 (0.055, 0.070) & 0.110 (0.101, 0.120) & 0.161 (0.154, 0.168) & 0.191 (0.190, 0.192) & 0.182 (0.174, 0.188) \\
 6  &  0.061 (0.054, 0.070) & 0.109 (0.100, 0.120) & 0.160 (0.153, 0.168) & 0.191 (0.190, 0.192) & 0.182 (0.175, 0.188) \\
 7  &  0.061 (0.053, 0.070) & 0.109 (0.100, 0.119) & 0.160 (0.153, 0.168) & 0.191 (0.189, 0.192) & 0.183 (0.175, 0.189) \\
 8  &  0.061 (0.053, 0.070) & 0.109 (0.099, 0.120) & 0.160 (0.152, 0.168) & 0.191 (0.189, 0.192) & 0.183 (0.175, 0.189) \\
 9  &  0.061 (0.053, 0.070) & 0.109 (0.099, 0.120) & 0.160 (0.152, 0.168) & 0.192 (0.190, 0.193) & 0.183 (0.174, 0.190) \\
 10 &  0.061 (0.053, 0.071) & 0.110 (0.099, 0.121) & 0.161 (0.153, 0.170) & 0.192 (0.190, 0.193) & 0.183 (0.174, 0.190) \\
 11 &  0.063 (0.054, 0.073) & 0.112 (0.101, 0.124) & 0.163 (0.154, 0.171) & 0.193 (0.191, 0.193) & 0.182 (0.172, 0.189) \\
 12 &  0.065 (0.056, 0.076) & 0.115 (0.103, 0.127) & 0.166 (0.156, 0.174) & 0.193 (0.192, 0.194) & 0.179 (0.169, 0.188) \\
 13 &  0.069 (0.058, 0.080) & 0.119 (0.107, 0.131) & 0.168 (0.159, 0.176) & 0.193 (0.191, 0.194) & 0.176 (0.166, 0.186) \\
 14 &  0.073 (0.062, 0.084) & 0.123 (0.111, 0.136) & 0.171 (0.163, 0.178) & 0.193 (0.190, 0.193) & 0.173 (0.161, 0.183) \\
 15 &  0.077 (0.066, 0.089) & 0.128 (0.116, 0.140) & 0.174 (0.166, 0.180) & 0.192 (0.188, 0.193) & 0.168 (0.156, 0.179) \\
 16 &  0.082 (0.070, 0.094) & 0.133 (0.120, 0.145) & 0.177 (0.169, 0.182) & 0.190 (0.185, 0.192) & 0.163 (0.150, 0.175) \\
 17 &  0.086 (0.073, 0.100) & 0.137 (0.124, 0.150) & 0.179 (0.171, 0.184) & 0.188 (0.182, 0.192) & 0.159 (0.145, 0.171) \\
 18 &  0.091 (0.077, 0.106) & 0.142 (0.127, 0.155) & 0.180 (0.173, 0.185) & 0.186 (0.178, 0.191) & 0.154 (0.138, 0.168) \\
 19 &  0.096 (0.080, 0.112) & 0.146 (0.131, 0.159) & 0.182 (0.175, 0.185) & 0.184 (0.174, 0.190) & 0.148 (0.131, 0.165) \\
 20 &  0.101 (0.084, 0.118) & 0.150 (0.134, 0.163) & 0.183 (0.176, 0.184) & 0.181 (0.170, 0.189) & 0.143 (0.125, 0.161) \\
  \hline
 $\kappa_{n^*}$  & $x=$16 & $x=$17 & $x=$18 & $x=$19 &$x=$20   \\ 
  \hline
 0  & 0.130 (0.119, 0.141) & 0.071 (0.062, 0.081) & 0.028 (0.023, 0.033) & 0.007 (0.005, 0.008) & 0.001 (0.001, 0.001) \\
 1  & 0.131 (0.120, 0.142) & 0.072 (0.064, 0.082) & 0.028 (0.024, 0.033) & 0.007 (0.006, 0.009) & 0.001 (0.001, 0.001) \\
 2  & 0.133 (0.121, 0.143) & 0.074 (0.064, 0.083) & 0.029 (0.024, 0.034) & 0.007 (0.006, 0.009) & 0.001 (0.001, 0.001) \\
 3  & 0.134 (0.122, 0.144) & 0.074 (0.065, 0.083) & 0.029 (0.025, 0.034) & 0.007 (0.006, 0.009) & 0.001 (0.001, 0.001) \\
 4  & 0.135 (0.124, 0.145) & 0.075 (0.066, 0.084) & 0.030 (0.025, 0.035) & 0.007 (0.006, 0.009) & 0.001 (0.001, 0.001) \\
 5  & 0.135 (0.124, 0.145) & 0.076 (0.067, 0.085) & 0.030 (0.025, 0.035) & 0.008 (0.006, 0.009) & 0.001 (0.001, 0.001) \\
 6  & 0.136 (0.124, 0.146) & 0.076 (0.067, 0.086) & 0.030 (0.025, 0.036) & 0.008 (0.006, 0.009) & 0.001 (0.001, 0.001) \\
 7  & 0.137 (0.125, 0.147) & 0.077 (0.067, 0.086) & 0.031 (0.025, 0.036) & 0.008 (0.006, 0.009) & 0.001 (0.001, 0.001) \\
 8  & 0.137 (0.124, 0.148) & 0.077 (0.067, 0.087) & 0.031 (0.025, 0.036) & 0.008 (0.006, 0.010) & 0.001 (0.001, 0.001) \\
 9  & 0.137 (0.124, 0.148) & 0.077 (0.066, 0.087) & 0.030 (0.025, 0.036) & 0.008 (0.006, 0.009) & 0.001 (0.001, 0.001) \\
 10 & 0.135 (0.123, 0.148) & 0.075 (0.065, 0.086) & 0.030 (0.024, 0.036) & 0.007 (0.006, 0.009) & 0.001 (0.001, 0.001) \\
 11 & 0.133 (0.120, 0.146) & 0.073 (0.063, 0.085) & 0.028 (0.023, 0.035) & 0.007 (0.005, 0.009) & 0.001 (0.001, 0.001) \\
 12 & 0.130 (0.116, 0.143) & 0.070 (0.060, 0.082) & 0.027 (0.022, 0.033) & 0.006 (0.005, 0.008) & 0.001 (0.001, 0.001) \\
 13 & 0.125 (0.112, 0.139) & 0.067 (0.056, 0.079) & 0.025 (0.020, 0.031) & 0.006 (0.005, 0.008) & 0.001 (0.000, 0.001) \\
 14 & 0.120 (0.106, 0.134) & 0.063 (0.052, 0.074) & 0.023 (0.018, 0.029) & 0.005 (0.004, 0.007) & 0.001 (0.000, 0.001) \\
 15 & 0.115 (0.101, 0.129) & 0.059 (0.049, 0.070) & 0.021 (0.017, 0.027) & 0.005 (0.004, 0.006) & 0.001 (0.000, 0.001) \\
 16 & 0.109 (0.095, 0.124) & 0.055 (0.045, 0.066) & 0.020 (0.015, 0.025) & 0.004 (0.003, 0.006) & 0.000 (0.000, 0.001) \\
 17 & 0.104 (0.089, 0.120) & 0.051 (0.042, 0.063) & 0.018 (0.014, 0.023) & 0.004 (0.003, 0.005) & 0.000 (0.000, 0.001) \\
 18 & 0.099 (0.083, 0.116) & 0.048 (0.037, 0.060) & 0.016 (0.012, 0.022) & 0.004 (0.002, 0.005) & 0.000 (0.000, 0.001) \\
 19 & 0.094 (0.077, 0.112) & 0.044 (0.034, 0.057) & 0.015 (0.011, 0.021) & 0.003 (0.002, 0.005) & 0.000 (0.000, 0.001) \\
 20 & 0.089 (0.072, 0.108) & 0.041 (0.031, 0.054) & 0.014 (0.010, 0.019) & 0.003 (0.002, 0.004) & 0.000 (0.000, 0.000) \\
   \hline
 \hline
\end{tabular}
\normalsize
    \caption{Estimated probabilities of exactly $x$ individuals surviving under the counterfactual regime $\Lambda_2$ in a finite cluster of $n^*=20$, that is, $\mathbb{P}_{n^*=20}\Big(\mathbb{Y}_{n^*=20}^{G^{\Lambda_2}_{n^*=20}}(1)=x\Big)$, under varying resource constraints ($\kappa_{n^*}$). Estimates are obtained using a parametric plug-in estimator of the compositional g-formula based on parametric models for $Q_Y$, with 95\% confidence intervals computed using 500 nonparametric bootstrap samples. Probabilities for $x<6$ are $0.000$ $(0.000, 0.000)$ for all $\kappa_{n^*}$.}
    \label{tab: CountTable}

\end{table}

\begin{figure}[ht]
    \centering
    \begin{tikzpicture}[scale=1]
        \node[anchor=north, scale=.5833333] at (0,0){%
            \pgfimage{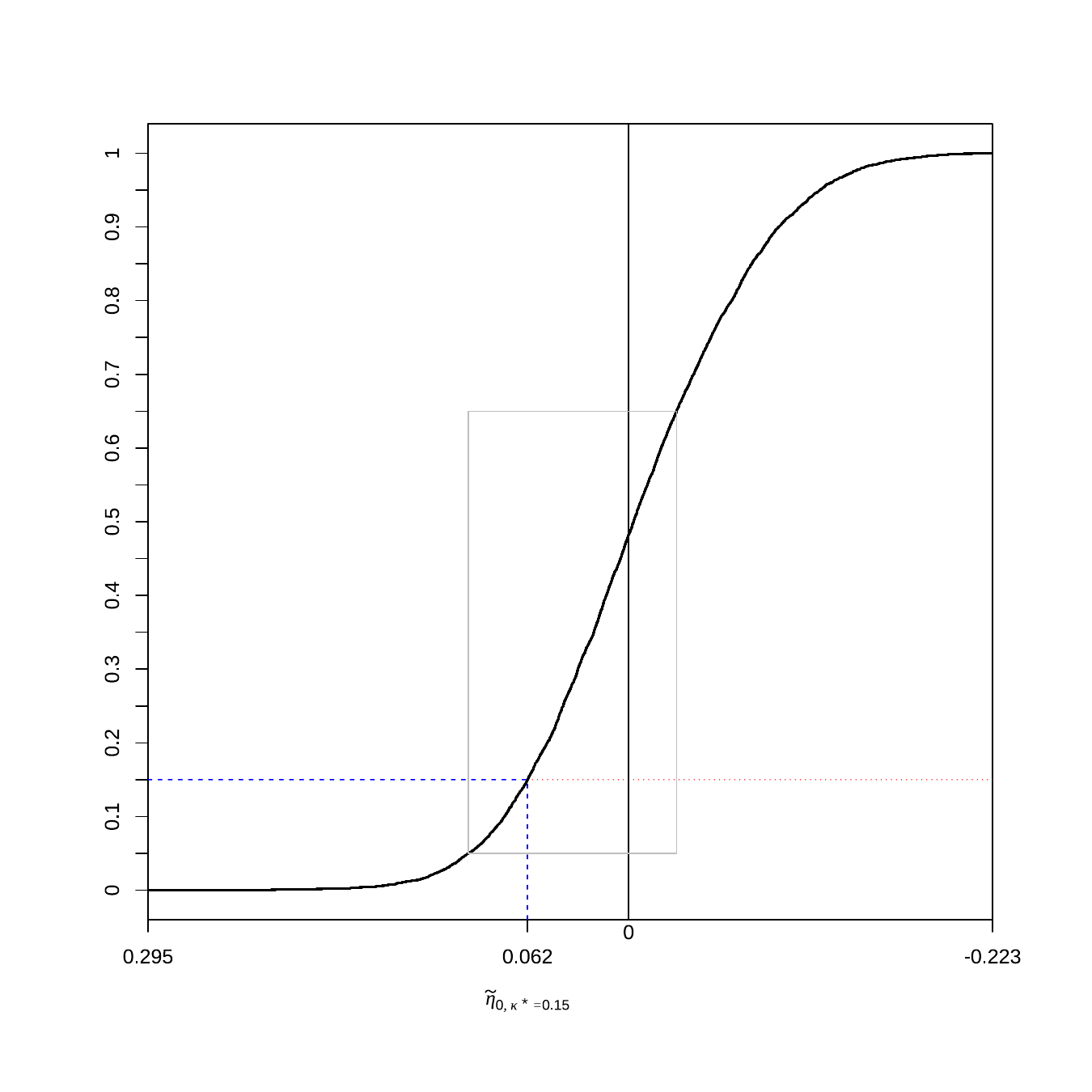}};

        \node[anchor = south east] at (-6,-7) {$\widetilde{\mathbb{P}}_0\Big(\overline{\Delta}_0(L_i)>x\Big)$};
        \node[anchor = south west] at (6,-7) {$\widetilde{\mathbb{E}}_{\mathbb{P}_0}\Big[\overline{A}^{\mathbf{G}_{0}^{\mathbf{opt}_1}}_0\Big]$};
        \node[anchor = south west] at (0,-13) {$x$};
\end{tikzpicture}
\caption{Function for the estimated probability that the proportion of individuals in a large cluster has a $V$-conditional average treatment effect (CATE) of prompt ICU admission on 90-day survival greater than $x$. Here, $V$ includes all individual covariates in $L_i$. The blue dashed line depicts the treatment assignment mechanism of the optimal $V$-rank-preserving regime, $\mathbf{G}_{0}^{\mathbf{opt}_1}$: for a given $\kappa^*$, the regime assigns treatment to individuals with $\Delta_0(L_i)$ greater than some number $\eta_{0, \kappa^*}$ such that precisely ${\mathbb{P}}_0\Big(\overline{\Delta}_0(L_i)>\eta_{0, \kappa^*}\Big) = \kappa^*$, or such that $\eta_{0, \kappa^*}=0$ and this probability is less than or equal to $\kappa^*$. The red dashed line indicates that when  ${\mathbb{P}}_0\Big(\overline{\Delta}_0(L_i)>\eta_{0, \kappa^*}\Big) = \kappa^*$, then so will be the proportion of treated individuals under the optimal regime in the large cluster, ${\mathbb{E}}_{\mathbb{P}_0}\Big[\overline{A}_0^{\mathbf{G}_{0}^{\mathbf{opt}_1}}\Big] = \kappa^*$. The black line depicts the CATE beyond which no individuals are treated under the optimal regime $\mathbf{G}_{0}^{\mathbf{opt}_1}$. Figure \ref{fig: OptTightGrid} focuses on the region of this plot within the gray rectangle.}
\label{fig: OptFullGridComp}
    \end{figure}

\begin{figure}[ht]
    \centering
    \begin{tikzpicture}[scale=1]
        \node[anchor=north, scale=.625] at (0,0){%
            \pgfimage{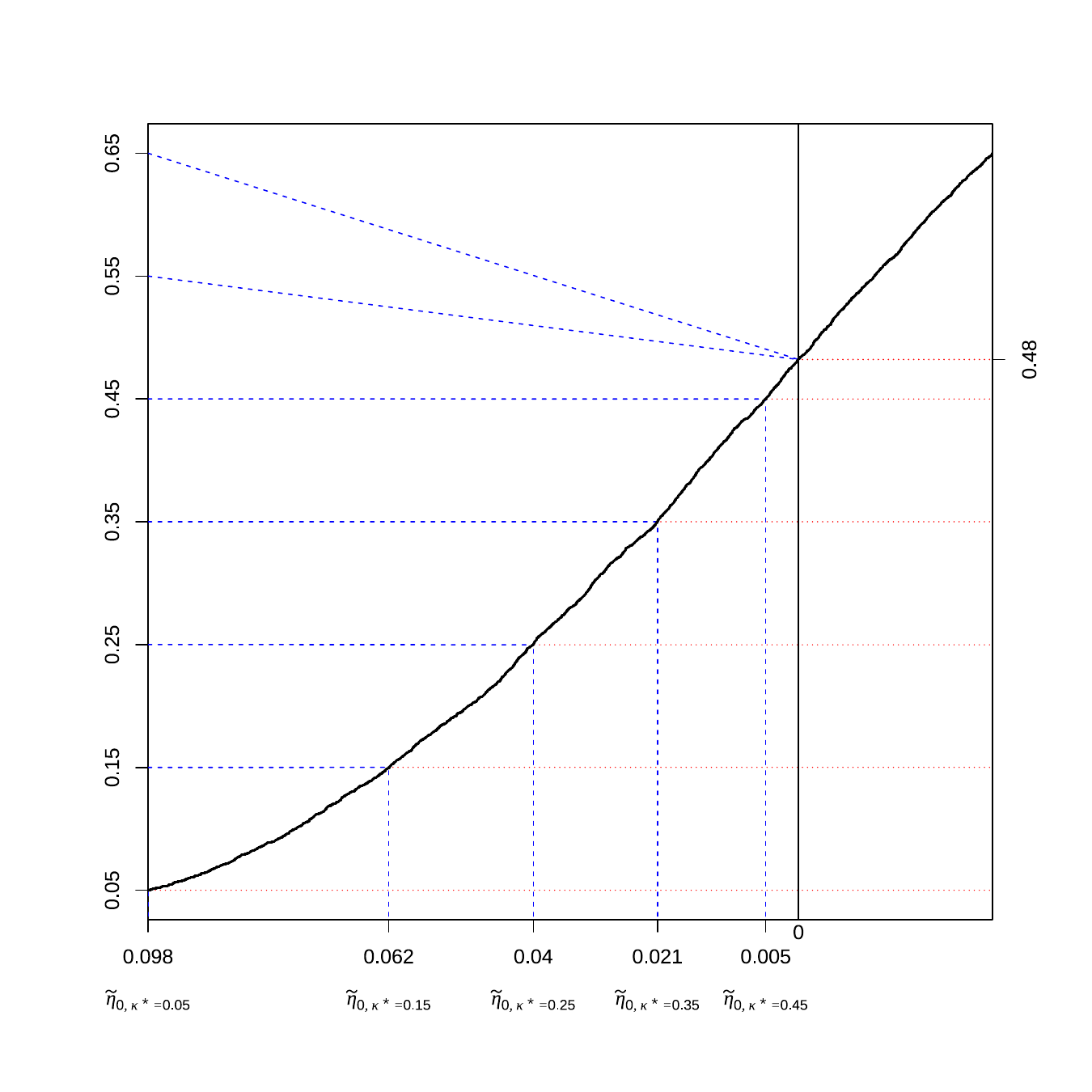}};

        \node[anchor = south east] at (-6,-7) {$\widetilde{\mathbb{P}}_0\Big(\overline{\Delta}_0(L_i)>x\Big)$};
        \node[anchor = south west] at (6,-7) {$\widetilde{\mathbb{E}}\Big[\overline{A}^{\mathbf{G}_{0}^{\mathbf{opt}_1}}_0\Big]$};
        \node[anchor = south west] at (0,-14) {$x$};
\end{tikzpicture}
\caption{Inset of Figure \ref{fig: OptFullGridComp}, with varying $\kappa^*$. The estimated expected proportion of treated individuals under the optimal regime in the large cluster is maximized at $\approx0.48$, the probability that the proportion of individuals in a large cluster has a positive $V$-conditional average treatment effect (CATE) of prompt ICU admission on 90-day survival.}
\label{fig: OptTightGrid}
    \end{figure}

\newpage
\begin{landscape}
\begin{figure}[ht]
    \centering
    \begin{tikzpicture}[scale=1]
        \node[anchor=north, scale=.5] at (0,0){%
            \pgfimage{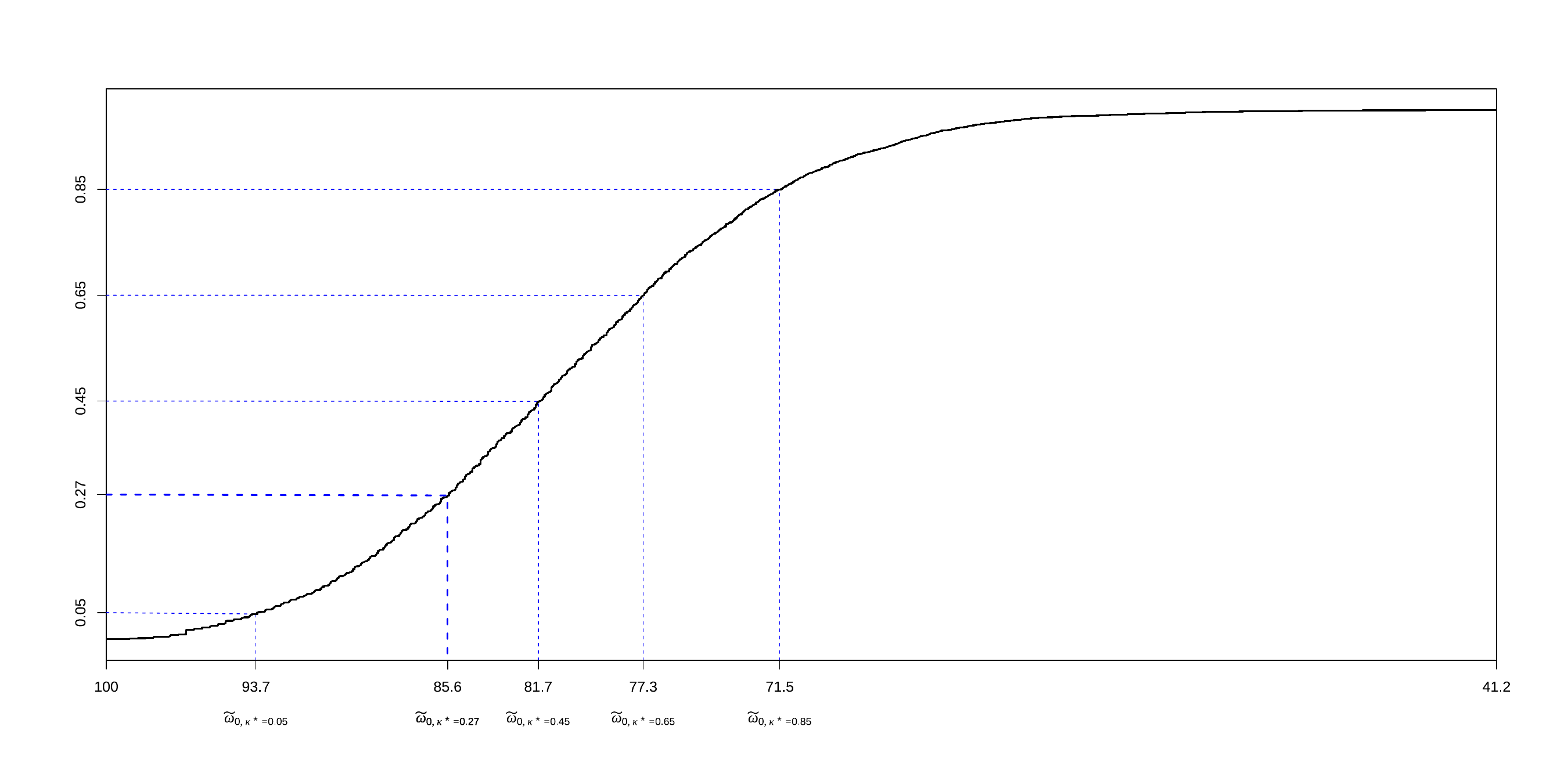}};

        \node[anchor = south east] at (-11,-7) {$\widetilde{\mathbb{P}}_0\Big({\Lambda}_1(L_i)>x\Big)$};
        \node[anchor = south west] at (0,-11.5) {$x$};
\end{tikzpicture}
\caption{Function for the estimated probability that the proportion of individuals in a large cluster has a ${\Lambda}_1$-rank-ordering greater than $x$. Here, ${\Lambda}_1(V_i)$ is an individual's weighted average of their clinical risk scores (NEWS, ICNARC, SOFA). In contrast to Figures \ref{fig: OptFullGridComp} and \ref{fig: OptTightGrid}, there is no black line indicating the value beyond which no individual is treated, regardless of resource availability; under $\mathbf{G}_{0}^{{\Lambda}_1}$, individuals are treated in order of their rank until the treatment resource is exhausted.}
\label{fig: Candidate_full_nat}
    \end{figure}

\end{landscape}
\newpage

\subsection{Directed acyclic graphs illustrating conditions sets A and B} \label{appsec: graphs}

\begin{figure}[ht] 
\centering

\begin{tikzpicture}[inner sep=0.3, outer sep=0.3, scale=.5]
        \node (A1)    at  (4   *3, 4  ) {$A_1$};
        \node (A2)    at  (4   *3, 0   ) {$A_{2}$};
        
        \node (L1)    at  (1   *3, 4  )  {$L_1$};
        \node (L2)    at  (1   *3, 0   ) {$L_{2}$};

        \node (Y1)    at  (7   *3, 4   ) {$Y_{1}$};
        \node (Y2)    at  (7   *3, 0   ) {$Y_{2}$};
        
\begin{scope}

 
        \path (L1)          edge[bend left=10, very thick ]  (Y1);
        \path (L1)          edge[very thick]                  (A1);
        \path (A1)          edge[very thick]                   (Y1);
        
        \path (L2)          edge[bend right=10, very thick  ]  (Y2);
        \path (L2)          edge[very thick]                   (A2);
        \path (A2)          edge[very thick]                   (Y2);

\path[{latex}-{latex}] (L1)    edge[bend left=30, very thick, black ]                  (A1);
\path[{latex}-{latex}] (L1)    edge[bend left=45, very thick, black ]                  (Y1);
\path[{latex}-{latex}] (A1)    edge[bend left=30, very thick, black ]                  (Y1);

\path[{latex}-{latex}] (L2)    edge[bend right=30,  very thick, black]                  (A2);
\path[{latex}-{latex}] (L2)    edge[bend right=45,  very thick, black]                  (Y2);
\path[{latex}-{latex}] (A2)    edge[bend right=30,  very thick, black]                  (Y2);

\path[{latex}-{latex}] (L1)    edge[bend right=30, blue ]                  (L2);
\path[{latex}-{latex}] (A1)    edge[bend right=30, red]                  (A2);
\path[{latex}-{latex}] (Y1)    edge[bend right=30, blue ]                  (Y2);

\path (L1)    edge[ blue]                 (L2);
\path (A1)    edge [very thick, red ]                (A2);
\path (Y1)    edge[ blue]                  (Y2);               
        
\path (L1)    edge  [orange]             (A2);
\path (L2)    edge  [orange]             (A1);
\path (A1)    edge [ blue]                 (Y2);
\path (A2)    edge [ blue]                 (Y1);
\path (L1)    edge [ blue]                 (Y2);
\path (L2)    edge [ blue]                 (Y1);

\path[{latex}-{latex}] (L1)    edge[bend left=20, orange ]                  (A2);
\path[{latex}-{latex}] (L2)    edge[bend right=20, orange ]                  (A1);
\path[{latex}-{latex}] (A1)[ blue]    edge[bend left= 20 ]                  (Y2);
\path[{latex}-{latex}] (A2)[ blue]    edge[bend right=20 ]                  (Y1);

\path[{latex}-{latex}] (L1) [ blue]   edge[bend left=120,min distance=10cm ]                  (Y2);
\path[{latex}-{latex}] (L2) [ blue]   edge[bend right=120,min distance=10cm ]                  (Y1);

 \end{scope}       
        
\end{tikzpicture}
\caption{DAG for a cluster of size $n=2$ illustrating condition $A1$ (Conditional noninterference). Bi-directed arrows represent unmeasured common causes, and the bold red arrow represents the characteristic feature of the limited resource setting. This DAG includes all possible directed paths and all possible bi-directed paths between all pairs of nodes. Blue paths include all between-unit paths with the exception of those between covariates $\{L_1, L_2\}$ and treatments $\{A_1, A_2\}$ (in orange) and those among treatments $\{A_1, A_2\}$ (in red). The presence of any single blue path would in general contradict condition $A1$.} 
\label{fig: CondIll1}
\end{figure}
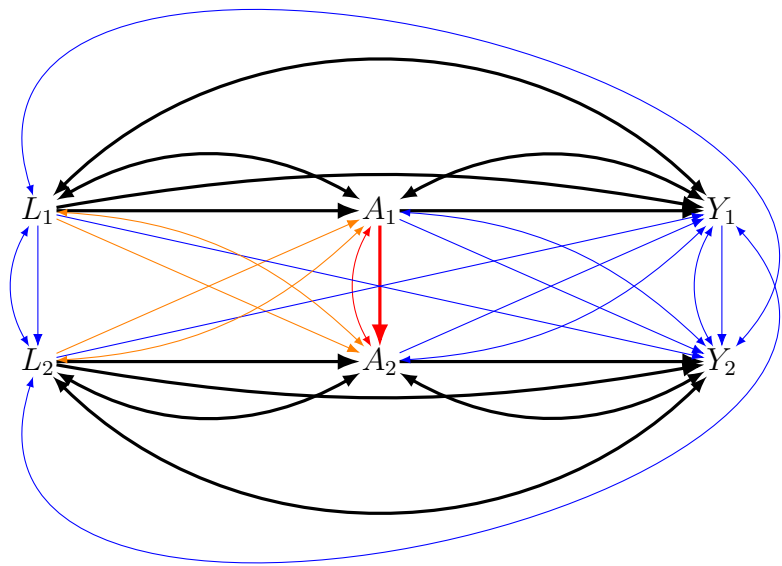
\clearpage

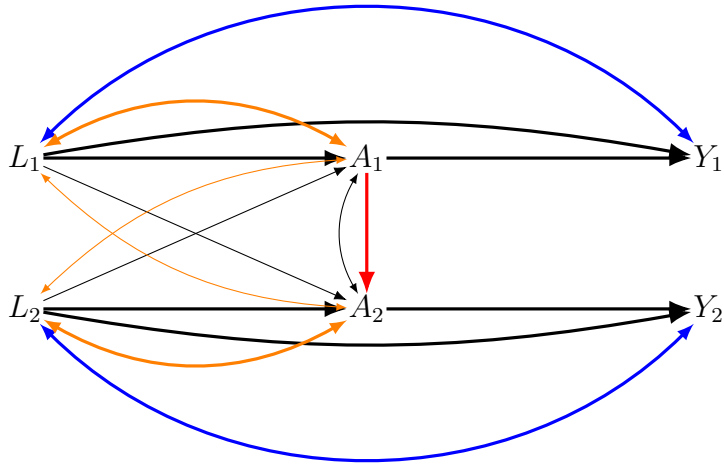
\begin{figure}[ht] 
\centering
\begin{tikzpicture}[inner sep=0.3, outer sep=0.3, scale=.5]
        \node (A1)    at  (4   *3, 4  ) {$A_1$};
        \node (A2)    at  (4   *3, 0   ) {$A_{2}$};
        
        \node (L1)    at  (1   *3, 4  )  {$L_1$};
        \node (L2)    at  (1   *3, 0   ) {$L_{2}$};

        \node (Y1)    at  (7   *3, 4   ) {$Y_{1}$};
        \node (Y2)    at  (7   *3, 0   ) {$Y_{2}$};
        
\begin{scope}

 
        \path (L1)          edge[bend left=10, very thick ]  (Y1);
        \path (L1)          edge[very thick]                  (A1);
        \path (A1)          edge[very thick]                   (Y1);
        
        \path (L2)          edge[bend right=10, very thick  ]  (Y2);
        \path (L2)          edge[very thick]                   (A2);
        \path (A2)          edge[very thick]                   (Y2);

\path[{latex}-{latex}] (L1)    edge[bend left=45, very thick, blue ]                  (Y1);

\path[{latex}-{latex}] (L2)    edge[bend right=45,  very thick, blue]                  (Y2);

\path[{latex}-{latex}] (A1)    edge[bend right=30, black]                  (A2);

\path (A1)    edge [very thick, red ]                (A2);
        
\path (L1)    edge  [black]             (A2);
\path (L2)    edge  [black]             (A1);

\path[{latex}-{latex}] (L1)    edge[bend left=30, very thick, orange ]                  (A1);
\path[{latex}-{latex}] (L2)    edge[bend right=30,  very thick, orange]                  (A2);
\path[{latex}-{latex}] (L1)    edge[bend right=20, orange ]                  (A2);
\path[{latex}-{latex}] (L2)    edge[bend left=20, orange ]                  (A1);
 \end{scope}       
        
\end{tikzpicture}
\caption{DAGs for a cluster of size $n=2$ illustrating condition $B$  (No unit-level confounding of outcomes). Bi-directed arrows represent unmeasured common causes, and the bold red arrow represents the characteristic feature of the limited resource setting. The DAG only includes arrows that are in general consistent with condition $A1$ (Conditional noninterference), i.e., it excludes the blue arrows in the DAG of figure \ref{fig: CondIll1}. Bi-directed arrows between $A_i$ and $Y_i$, which classically contradict $B1$, are excluded.  In this DAG, blue arrows represent un-measured common causes of $L_i$ and $Y_i$. Orange arrows represent un-measured common causes of $L_i$ and $A_i$, as well as those of $L_i$ and $A_j$ for $j\neq i$. In the iid setting, it is well-known that the presence of both $L_i \leftrightarrow Y_i$ and $L_i \leftrightarrow A_i$ will contradict condition $B$ due to the back-door path $A_i \leftrightarrow L_i \leftrightarrow Y_i$, which is open conditional on $L_i$. However, the presence of both $L_i \leftrightarrow Y_i$ and $L_i \leftrightarrow A_j$ will also in general contradict condition $B$. For example, the back-door path $A_2 \leftarrow A_1 \leftrightarrow L_2 \leftrightarrow Y_2$ would be open conditional on $L_2$. Thus, a model with either the blue set of arrows but not the orange set of arrows, or vice versa would be consistent with condition $B$, but a model with both sets will in general not.} 
\label{fig: CondIll2}
\end{figure}
\clearpage

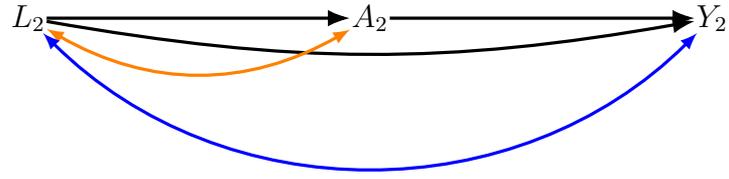
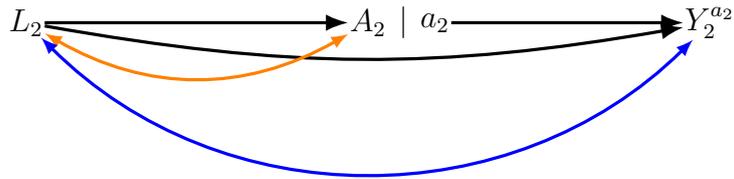
\begin{figure}[ht] 
\centering\
\subfloat[]{
\begin{tikzpicture}[inner sep=0.3, outer sep=0.3, scale=.5]
        \node (A2)    at  (4   *3, 0   ) {$A_{2}$};
        
        \node (L2)    at  (1   *3, 0   ) {$L_{2}$};

        \node (Y2)    at  (7   *3, 0   ) {$Y_{2}$};
        
\begin{scope}

 
        
        \path (L2)          edge[bend right=10, very thick  ]  (Y2);
        \path (L2)          edge[very thick]                   (A2);
        \path (A2)          edge[very thick]                   (Y2);
\path[{latex}-{latex}] (L2)    edge[bend right=30,  very thick, orange]                  (A2);
\path[{latex}-{latex}] (L2)    edge[bend right=45,  very thick, blue]                  (Y2);


        

 \end{scope}       
        
\end{tikzpicture}
}\\ \par\bigskip 

\subfloat[]{
\begin{tikzpicture}[inner sep=0.3, outer sep=0.3, scale=.5]
        \node (A2)    at  (4   *3, 0   ) {$A_{2}$};
        \node (midA2)[right = .125 of A2,    anchor=west] {$\mid$};
            \node (A2+)  [right = .125 of midA2, anchor=west] {$a_2$};
        \node (L2)    at  (1   *3, 0   ) {$L_{2}$};

        \node (Y2)    at  (7   *3, 0   ) {$Y_{2}^{a_2}$};
        
\begin{scope}

 
        
        \path (L2)          edge[bend right=10, very thick  ]  (Y2);
        \path (L2)          edge[very thick]                   (A2);
        \path (A2+)          edge[very thick]                   (Y2);
\path[{latex}-{latex}] (L2)    edge[bend right=30,  very thick, orange]                  (A2);
\path[{latex}-{latex}] (L2)    edge[bend right=45,  very thick, blue]                  (Y2);


        

 \end{scope}       
        
\end{tikzpicture}
}
\caption{Marginalized DAG and Single World Intervention Graph (SWIG) for individual $i=2$ in a cluster of size $n=2$, illustrating explicit evaluation of condition $B$ (No unit-level confounding of outcomes). Bi-directed arrows represent unmeasured common causes. DAG (a) is a marginalization of the DAG in figure \ref{fig: CondIll2}, obtained using standard graph-reduction rules. Notably, DAG (a) would remain unchanged even had the DAG in figure \ref{fig: CondIll2} omitted the bi-directed path $L_2 \leftrightarrow A_2$: the bi-directed path $A_2 \leftrightarrow L_2$ in DAG (a) would follow from marginalizing over $A_1$ in the causal path $A_2 \leftarrow A_1 \leftrightarrow L_2$ in the DAG in figure \ref{fig: CondIll2}. SWIG (b) is obtained from DAG (a) via a simple graphical algorithm described in \citet{richardson2013single}. Therein, it is apparent that $Y_2^{a_2} \not\independent A_2 \mid L_2$ via a classical M-bias structure, and thus there is a violation of condition $B$.} 
\label{fig: CondIll3}
\end{figure}

\newpage
\subsection{Model robustness of optimal cluster- vs. individualized regimes} \label{appsec: modelrobust}

In Section \ref{sec:policyclass} of the main text, we emphasized that the regime class considered in the constrained optimization literature,
$\Pi_{P^F}(\kappa)$, is necessarily indexed by a particular law $P^F$; the class of IR's guaranteed to satisfy a constraint $\kappa$ across all laws contains only the regimes that randomly assign treatment with some probability less than or equal to $\kappa$, including the trivial deterministic regime that assigns $a=0$ to all individuals. On the other hand, the class of CRs  $\Pi^F_n(\kappa_n)$ is not a function of $P^F$. Here we illustrate by example differences in the robustness and transportability of constrained optimal IR vs.\ CR regimes, that stems from differences in their dependence on $P^F$ vs. $\mathbb{P}_n^F$. 

Consider an iid model $\mathcal{M}^{iid}_{n}$ for $\mathbb{P}_n^F$. Let $P_{i}^F \equiv P_{i}^F(\mathbb{P}_{n}^F)$ be the subdistribution of $O^F_i$ induced by $\mathbb{P}_{n}^F$. For each $\mathbb{P}_n^F\in\mathcal{M}^{iid}_{n}$, there exists a $P^F$ such that  $P_{i}^F = P^F$ for all $i$, i.e., all individuals are identically distributed. Then consider the submodel $\mathcal{M}^{iid, F}_{n} (\Delta) \subset \mathcal{M}^{iid, F}_{n}$ which fixes the CATE function to $\Delta$. It can be easily shown that there is a unique regime $G_n\in \Pi_n(\kappa_n)$ that maximizes $\mathbb{E}_{\mathbb{P}_{n}^F}[Y_i^{G_n}]$ for all laws in $\mathcal{M}^{iid, F}_{n} (\Delta)$. In contrast, however, there will be a \textit{set} of policy classes $\Pi_{P^F}(\kappa)$ under $\mathcal{M}^{iid, F}_{n} (\Delta)$: one for each law $\mathbb{P}_{n}^F$ in the model, each of which is uniquely characterized by $P^F$. It can also easily be shown that there does not in general exist a single law $g$ that is a common to all $\Pi_{P^F}(\kappa)$, where $\kappa$ is fixed, under $\mathcal{M}^{iid, F}_{n} (\Delta)$ that maximizes  $\mathbb{E}_{\mathbb{P}_{n}^F}[Y_i^{g}]$ at all laws in the model. As a contradiction to the alternative, consider a setting with $L_i$ taking values in $\{0, 1\}$ and suppose $\Delta$ is such that $\Delta(1) > \Delta(0)>0$. Suppose $\kappa=0.5$, and consider two laws $P_1^F$ and $P_2^F\in \mathcal{M}^{iid, F}_{n}$ distinguished by $P_1(L_i=1)=1$ and $P_2(L_i=1)=0.5$. The optimal IR $g$ in $\Pi_{P_1^F}(\kappa)$ that maximizes $\mathbb{E}_{P_1^F}[Y_i^{g}]$ is one under which treatment is withheld from individuals with $L_i=0$ and is randomly assigned to individuals with $L_i=1$ with probability 0.5, whereas the  optimal IR $g$ in $\Pi_{P_2^F}(\kappa)$ maximizing $\mathbb{E}_{P_2^F}[Y_i^{g}]$ is the distinct regime in which treatment is withheld from individuals with $L_i=0$ and is deterministically provided to individuals with $L_i=1$. Thus, in order for the constrained optimal IR to remain fixed across all laws in the model, further restrictions on the model must be made. One example of such a model is an  $\mathcal{M}^{iid, F}_{n} (\Delta, Q_L) \subset \mathcal{M}^{iid, F}_{n} (\Delta)$, in which the distribution of $L_i$ is fixed across all laws to $Q_L$. This is obviously a highly constrained model, in which the candidate laws are only distinguished by the treatment propensities.  

The purpose of this example is to emphasize that, while constrained optimal IRs and CRs are in general law-dependent parameters, they may become fixed (i.e. law-independent) when the space of possible laws is constrained by a model. This property is important whenever an investigator is interested in transporting an optimal regime learned in a particular observed data setting $\mathbb{P}^F_{n, 0}$ to some other future setting which may not share exactly the same law, say $\mathbb{P}^F_{n, 1}$. An investigator could justify such transportation by adopting a model for  $\mathbb{P}^F_{n, 0}$ and $\mathbb{P}^F_{n, 1}$, in which optimal regime to be fixed across all laws. The example illustrates that an investigator must make strictly stronger assumptions to validly transport a $\kappa$-constrained IR $g$ than the assumptions that would be necessary to validly transport a $\kappa_n$-constrained CR $G_n$, thus compromising the robustness of the transportation task for IRs. This is especially critical in optimal regime learning where the application of the learned parameter in future settings is the principal motivation.

\subsection{Optimal regimes with unrestricted conditional average treatment effects} \label{appsec: opt}

To simplify the presentation in the main text, Proposition \ref{lemma: optimallarge} provides results under the restricted class of regime $\Pi_n(\kappa_n)$, in which exactly $\kappa_n$ individuals are treated. We remarked that this regime will not in general be the optimal regime in the unrestricted class of constrained CRs $\Pi^F_n(\kappa_n)$ in which \textit{at most} $\kappa_n$ individuals are treated. In this appendix we define a sub-class of the generalized constrained CRs, $\Pi^F_n(\kappa_n)$ introduced in Appendix \ref{appsec: gencdtrs}, that generalizes the class of $L$-rank-preserving CRs, $\Pi_n(\kappa_n, L_i, \Lambda)$. In this definition, we leverage the rank-and-treat formulation of generalized CRs presented in Appendix \ref{appsec: gencdtrs}. Then we find that, for laws in $\mathcal{M}_n^{AB}$, the optimal such CR is the optimal CR in the unrestricted class of constrained CRs,  $\Pi^F_n(\kappa_n)$.

\begin{definition}[$L$-rank-preserving generalized CRs]\label{def: LRPregimesgen}

The class of CRs $\Pi^{F,d}_n(\kappa_n, L_i, \Lambda, g)$ is the subset of regimes $G_n$ in  $\Pi^F_n(\kappa_n)$ with rank-and-treat formulation characterized by $\{G^*_{n. \mathbf{R}}, g_1, \dots, g_n\}$ such that: 
\begin{enumerate}
    \item $G^*_{n. \mathbf{R}}$ characterizes an $L$-rank-preserving regime $G^{\dagger}_{n} \in \Pi_n(\kappa_n, L_i, \Lambda)$; and
    \item for all $i$, $g_i$ is a fixed function $g$ of at most $L_i$ and $\delta_{A_i}$. 
\end{enumerate}
\end{definition}

Now consider condition $E^{\ast}$, defined as a modification of condition $E$, wherein condition $E1$ is exchanged with the following:

\begin{itemize}
    \item [$E1^*.$] For each $n$, $G_n\in \Pi^{F,d}_n(\kappa_n, L_i, \Lambda, g)$.
\end{itemize}

\begin{proposition}[Optimal generalized CR] \label{lemma: optimallargegen}
Consider a law $\mathbb{P}_n^F\in\mathcal{M}_n^{AB}$. The CR  $G^{\mathbf{opt}, F}_{n}\equiv G_n\in\Pi^F_n(\kappa_n)$ that maximizes $\mathbb{E}_{\mathbb{P}_n^F}[\overline{Y}_n^{G_n}]$ is that $V$-rank-preserving generalized CR $G_n\in\Pi^{F,d}_n(\kappa_n, L_i, \Lambda, g)$ characterized by $\Lambda \equiv \Delta_0$ and $g$ such that $g(v) \equiv I\big(\Delta_0(v)>0\big)$.

Furthermore, consider an asymptotic law $\mathbb{P}_0^F$ following condition $D$, and an asymptotic regime $\mathbf{G}_0$ following condition $E^*$. If $G_{n^*}\in\mathbf{G}_0$ is such an optimal regime for $\mathbb{P}_{n^*}^F(\mathbb{P}_0^F)$ for some $n^*>1$ then $G_{n}\in\mathbf{G}_0$ is the optimal such regime for $\mathbb{P}_{n}^F(\mathbb{P}_0^F)$ for all $n$. Let $q^{\mathbf{opt}}_0$ denote the intervention density under the asymptotic regime $\mathbf{G}^{\mathbf{opt}, F}_{0} \equiv \Big(G^{\mathbf{opt},F}_{1}, G^{\mathbf{opt},F}_{2}, G^{\mathbf{opt},F}_{3},\dots\Big)$. Then, $\mathbb{E}_{\mathbb{P}_n^F}[\overline{Y}_0^{\mathbf{G}^{\mathbf{opt}, F}_{0}}]$ is identified as in Theorem \ref{theorem: YbarIDlarge} and its intervention density $q^{\mathbf{opt}, F}_0$ is identified as in \eqref{eq: gstar_large} of Proposition \ref{lemma: largeint} where we take $\Lambda = \Delta_0$ and $\omega_0 =\eta_0 \coloneqq \max\Big\{\inf\{ c  \mid \mathbb{P}_0(\Delta_0(L_i) > c) \leq \kappa^*\}, 0 \Big\}.$

\end{proposition}

\clearpage

\bibliographystyleSM{unsrtnat-init}
\bibliographySM{Limited1.bib}

\end{appendices}

\end{document}